\documentclass[final,times,12p]{elsarticle}

\usepackage{amsmath}
\usepackage{amssymb}
\usepackage{enumerate}
\usepackage[margin=1.in]{geometry}
\usepackage{epsfig}
\usepackage{color}
\usepackage{appendix}

\usepackage[utf8]{inputenc}

\usepackage{amsmath,bm}
\usepackage{array}

\usepackage{multirow}

\usepackage{graphicx,mathrsfs}
\usepackage{amsthm}
\usepackage{xspace}
\usepackage{verbatim}
\usepackage{hyperref}

\newtheorem{theorem}{Theorem}[section]
\newtheorem{lemma}[theorem]{Lemma}
\newtheorem{proposition}[theorem]{Proposition}
\newtheorem{corollary}[theorem]{Corollary}
\theoremstyle{definition}
\newtheorem{definition}[theorem]{Definition}

\theoremstyle{remark}
\newtheorem{remark}[theorem]{Remark}
\numberwithin{equation}{section}
\newcommand {\beq} {\begin{equation}}
\newcommand {\eeq} {\end{equation}}
\newcommand{\rmnum}[1]{\romannumeral #1}
\newcommand{\Rmnum}[1]

\begin{document}

\begin{frontmatter}

\title{Emergence of purely nonlinear localized states with frequencies exited from spectral bands}
\author[1]{Huajie Song \corref{cor1}} \ead{songhj@hust.edu.cn}
\author[1]{Haitao Xu} \ead{hxumath@hust.edu.cn}

\address[1]{Center for Mathematical Sciences, Huazhong University of Science and Technology, Wuhan, Hubei 430074, People's Republic of China}

\cortext[cor1]{Corresponding author}

\begin{abstract}

In this work, we revisit the classic model of diatomic chain with cubic nonlinearity and investigate the formation mechanism of nonlinear localized time-periodic solutions (breathers) with frequencies exited the spectral bands. First we employ the long-chain limit to obtain estimates of linear eigenstates, especially those with frequencies near the band edges. As the strength of nonlinearity grows, some frequencies can cross the band edges to turn isolated while the corresponding states gradually change from non-localized to localized. Based on the estimates of linear eigenstates, we derive analytical approximations of these nonlinear states and prove their validity for frequencies within the bands. Moreover, the process of states growing localized are illustrated in both analytical and numerical approaches. Although here we place emphasis on nonlinear middle-localized states with the most generic boundary conditions, the results can also be extended to a wider range of localized states including edge states.
~\\~\\
\noindent{\textbf{Keywords}: Diatomic chain, Discrete breathers, Edge states, Long-chain limit, Band edge, Localization formation}

\end{abstract}

\end{frontmatter}

\section{Introduction}
\label{sec:intro}

The localization of energy in lattices has been a classical topic for years. In nonlinear lattices the localized states such as solitons and breathers have been extensively studied and relevant researches have emerged in many different fields~\cite{FPU, Yaroslav, CDN, Lederer, Kartashov1,Kartashov2,Fibich, Malomed, Ovchinnikov, Kosevich, Sievers, Flach, Flach2, Campbell2}. On the other hand, topological insulators and edge states in recent decades have provided various examples and applications for localized states in linear lattices, which have attracted tremendous attention of researchers~\cite{Bernevig. B. A, Klitzing. K. V,Kane1,Kane2,Bernevig}.
Topological insulators are usually periodically-structured lattices with "edge states" propagating or being localized at the boundary (such as edges, corners and surfaces).  The existence of these localized states at the boundary is related to bulk topological invariants, which is also known as "bulk-boundary correspondence". As a result, the edge states are robust regarding perturbations that do not change the bulk topology, which can also been characterized by the topology of the energy bands. This framework has emerged essentially for linear systems and it has been powerful in analyzing and developing structures and materials in various fields, such as electronic systems, photonics and phononics.

More recently, there have been growing interests on the topological materials with nonlinearity~\cite{Zhang2020,Kirsch,Rajesh2022,Magnus}. Among others an important problem is the effects of nonlinearity on the existence of edge states (or topologically-protected states) and it mainly has two branches: One is whether linear edge states continue to be localized as the strength of nonlinearity grows; The other is what new edge states will appear in the nonlinear regime. 
For the first question, it has been commonly observed that the edge states with nonlinearity can possibly remain localized at the boundary while the corresponding energy (or frequency) becomes amplitude-dependent. For example, there have been studies on the existence~\cite{Y. Hadad2016,Y. Hadad2018,R. S. Savelev2018,R. Chaunsali2019,F. Zangeneh2019} and stability\cite{A. Bisianov2019,K. Mochizuki2020,Y. Lumer2016,Rajesh Chaunsali2020,E. Prodan2009,T. Shi2017} of nonlinear edge states continued from the linear limit. On the other hand, the latter concerns purely nonlinear states and we will place our focus on this more exotic branch. To be specific, since the frequency of a localized state must be outside the bands, one possible approach to constructing a purely nonlinear edge state is making a frequency exit the bands and become isolated. It is worthy noticing that this is a typical scenario for the emergence of nonlinear edge states in finite-size lattices but existing explanation for its mechanism is still limited~\cite{D.K. Campbell,Aubry1}. Therefore it is our goal of this work to characterize the states with frequencies near the band edges and understand the process of these frequencies exiting the bands and the states becoming localized. 

Here we consider the model of a finite-size long diatomic chain with nonlinear interactions and study the localized states (perhaps including the edge states and the middle-localized states) emerged as the strength of nonlinearity grows. The eigenstates in the corresponding linear chain with frequencies outside the bands or near the band edges can be well approximated in the long chain limit. This enables us to perform expansion analysis and track the change of them from the linear limit to the nonlinear regime.

The structure of this article is arranged as follows. In Section~\ref{sec:2}, we introduce the linear model of one-dimensional diatomic chain and explicitly calculate the eigenfrequencies and eigenstates; Then Section~\ref{sec:edgestates} shows the existence of edge states in long linear chains and its dependence on boundary conditions; In Section~\ref{sec:estimate}, estimates on the eigenfrequencies near the band edges and the corresponding eigenstates in linear chains are discussed; In Section~\ref{add nonlinearility}, with all the estimates on the linear states from previous sections, we obtain an approximation of the state with frequency near the band edge in the weakly nonlinear regime. Then the process of this state getting localized is analytically and numerically studied.

\section{One-dimensional linear diatomic chain}
\label{sec:2}

In this section, we start with a simple linear diatomic chain with identical masses $m$ (for simplicity we assume $m=1$) and alternating nearest-neighbour interactions. This is effectively a spring-mass system, whose equations of motion are as follows
%
\begin{equation}
\begin{split}
  \ddot{q}_{1} = & k_{3,1}(0-q_{1})+k_{1}(q_{2}-q_{1}),\\
  \ddot{q}_{2m} = & k_{1}(q_{2m-1}-q_{2m})+k_{2}(q_{2m+1}-q_{2m}), \\
  \ddot{q}_{2m+1} = & k_{2}(q_{2m}-q_{2m+1})+k_{1}(q_{2m+2}-q_{2m+1}),1\leq m\leq n-1 \\
  \ddot{q}_{2n} = & k_{1}(q_{2n-1}-q_{2n})+k_{3,2}(0-q_{2n}),
\end{split}
\label{eq:motion}
\end{equation}
where $q_j$ denotes the displacement of 
$j$-th mass, double dots denote differentiation twice in time, $k_1$ and $k_2$ are the two alternating spring stiffness constants, $k_{3,1}$ and $k_{3,2}$ are the stiffness constants of two springs connected to the left and right ends, respectively. Unless later specifically explained, it will be assumed that the stiffness constants are non-negative, $k_{3,1}^2+k_{3,2}^2\neq 0$ and $k_{2}>k_{1}>0$.

In the matrix form, the equations \eqref{eq:motion} can be rewritten as 
\begin{equation}
  \ddot{{q}}=\mathcal{L}{q}
  \label{eq:motion1}
\end{equation}
where ${q}=(q_{1},q_{2},
\cdots,q_{2n})^{\top}$. 
It can be easily seen that the matrix $\mathcal{L}$ is
diagonally-dominant, symmetric and invertible. Focusing on the time-periodic solutions and adopting the ansatz ${q}(t)={u}e^{i\omega t}$, we find that the matrix $\mathcal{L}$ 
has $2n$ distinct negative eigenvalues 
$\{-(\omega^{(j)})^{2}\}_{1\leq j\leq 2n}$ and the corresponding eigenvectors are denoted by $\{{u^{(j)}}\}_{1\leq j\leq 2n}$. 
\begin{remark}
If $k_{3,1}=k_{3,2}$, then the chain is symmetric about the middle. As a result, if $\vec{u}=(u_1,u_2,\cdots,u_{2n})^T$ is an eigenvector of $\mathcal{L}$ for eigenvalue $-\omega^2$, then $\vec{u'}=(u_{2n},u_{2n-1},\cdots,u_{1})^T$ also satisfies $-\omega^2 \vec{u'}=\mathcal{L}\vec{u'}$. Since the eigenvalues of $\mathcal{L}$ are distinct, there exists some $c\neq 0$ such that $\vec{u}=c\vec{u'}$. 
\end{remark}

Due to periodic structure of the chain, the adjacent two masses form a unit cell and two adjacent unit cells satisfy the following relation
\begin{equation}
\begin{split}
  -\omega^{2}u_{2m+1}&= k_{2}(u_{2m}-u_{2m+1})+k_{1}(u_{2m+2}-u_{2m+1}), \\
  -\omega^{2}u_{2m+2}&= k_{1}(u_{2m+1}-u_{2m+2})+k_{2}(u_{2m+3}-u_{2m+2}), \quad 1\leq m < n-2
\end{split}
\end{equation}
or
\begin{equation}
\begin{split}
  \left(
  \begin{array}{c}
    u_{2m+1} \\
    u_{2m+2}
  \end{array}
  \right)
  =&
  \frac{1}{k_{1}k_{2}}
  \left(
  \begin{array}{cc}
    -k_{1}^{2} & -k_{1}(\omega^{2}-k_{1}-k_{2}) \\
    k_{1}(\omega^{2}-k_{1}-k_{2}) & (\omega^{2}-k_{1}-k_{2})^{2}-k_{2}^{2}
  \end{array}
  \right)
  \left(
  \begin{array}{c}
    u_{2m-1} \\
    u_{2m}
  \end{array}
  \right)\\
  =&
  \mathcal{T}(\omega)
  \left(
  \begin{array}{c}
    u_{2m-1} \\
    u_{2m}
  \end{array}
  \right).
  \end{split}
  \label{eq:iteration}
\end{equation}
It should be noticed that this relation \eqref{eq:iteration} also holds for infinite diatomic chains where $\omega^2$ can take negative values (meaning time-growing or time-decaying solutions). However in what follows, we will only place emphasis on the time-periodic solutions where $\omega^2\geq 0$. 

If the eigenvalues of $\mathcal{T}$ are denoted by $\lambda$, then they can be calculated from
\begin{equation}
\label{iterative eigenvalue}
  \lambda^{2}+\lambda(\frac{-\omega^{4}
  +2\omega^{2}(k_{1}+k_{2})-2k_{1}k_{2}}
  {k_{1}k_{2}})+1=0
\end{equation}
where two roots of the above equation satisfy $\lambda_{1}\lambda_{2}=1$ and $\lambda_1+\lambda_2\in\mathbb{R}$. This implies that 
\begin{itemize}
\item either (A1): $|\lambda|=1$ (when $\omega^2\in [0,2k_{1}]\cup[2k_{2},2k_{1}+2k_{2}]$);
\item or (A2): $\lambda\in\mathbb{R}$ and $\lambda\neq \pm 1$ (when $\omega^2\in (2k_{1},2k_2)\cup(2k_{1}+2k_2,\infty)$). 
\end{itemize}
Without loss of generality, we assume that $|\lambda_1|\leq|\lambda_2|$ and use $a$ to denote the eigenvalue $\lambda_1$, then the eigenfrequency $\omega$ can be expressed as
\begin{equation}
  \omega^{2}=k_{1}+k_{2}
  \pm\sqrt{(k_{1}+k_{2}a)
  (k_{1}+k_{2}/a)}.
  \label{eq:omega}
\end{equation}
As a special case, the eigenvalue of $\mathcal{T}$ is $1$ or $-1$ (with algebraic multiplicity two and geometric multiplicity one) when $\omega^2$ sits exactly at some edge of the two bands $[0,2k_{1}]\cup[2k_{2},2k_{1}+2k_{2}]$, namely $\omega^2\in\{0, 2k_1, 2k_2, 2k_1+2k_2\}$. On the other hand, in the generic situation where $a\neq\pm1$ and $(k_{1}+k_{2}a)(k_{1}+k_{2}/a)\geq 0$, the eigenvectors of $\mathcal{T}$ are explicitly
\begin{equation}
\begin{split}
  \vec{v}_{1}(a) =
  \left(
  \begin{array}{c}
    v_{11} \\
    \sigma v_{12}
  \end{array}
  \right)=
  \left(
  \begin{array}{c}
    \sqrt{k_{1}+k_{2}/a} \\
    \sigma\sqrt{k_{1}+k_{2}a}
  \end{array}
  \right);
  \vec{v}_{2}(a) =
    \left(
  \begin{array}{c}
    v_{21} \\
   \sigma v_{22}
  \end{array}
  \right)=
  \left(
  \begin{array}{c}
    \sqrt{k_{1}+k_{2}a} \\
    \sigma\sqrt{k_{1}+k_{2}/a}
  \end{array}
  \right)
\end{split}
\label{eq:eigenvector}
\end{equation}
for the eigenvalues $a$ and $1/a$,
respectively, where $\omega^2$ here adopts the form
\begin{equation}
  \omega^{2}=k_{1}+k_{2}-
  \sigma\sqrt{k_{1}+k_{2}a}\sqrt{k_{1}+k_{2}/a}, \quad \sigma=\pm 1.
\end{equation}
To obtain the corresponding eigenvector $u$ of $\mathcal{L}$, if $(u_1, u_2)^T=c_1 v_1+ c_2 v_2$, then $(u_{2m+1}, u_{2m+2})^T=c_1 a^m v_1+c_2 a^{-m} v_2$. 
It can be observed that choosing $c_2=0$ ($c_1=0$) and $|a|<1$ gives an eigenvector $u$ localized at the left (right) edge of the chain and a solution of this shape is usually known as an ``edge state'' in one-dimensional lattices.

In a relatively long chain, a typical scenario is that most of the frequencies $(\omega^{(j)})^{2}$ belong to the two bands $[0,2k_1]\cup[2 k_2, 2k_1+2k_2]$ (usually called ``acoustic band'' and ``optical band'' respectively) and no more than two frequencies fall outside the bands. For those outside frequencies, although $c_1$ and $c_2$ in their corresponding eigenvectors generically do not vanish, we will show in the next section that they commonly give `` edge states'' in diatomic chains for a quite wide range of choices for boundary springs $k_{3,1}$ and $k_{3,2}$.

\section{``Edge states'' in a long linear diatomic chain}
\label{sec:edgestates}

 It can be inferred from the previous section that if $k_{3,1}$ and $k_{3,2}$ are carefully chosen such that $\mathcal{L}$ has an eigenstate $u$ with $|a|<1$ and $c_1 c_2=0$, then this state is a genuine edge state. On the other hand, observations suggest that states looking localized at one end are actually common in long chains and the constraints on $k_{3,1}$ and $k_{3,2}$ can be possibly relaxed.
%
%
To be specific, the chains studied in this section are assumed to be long enough (namely $n\gg 1$) and we consider ``edge states'' in the sense 
$$
(u_{2m-1}, u_{2m})\approx (u_{1}, u_{2})a^{m-1},~{\rm or}~(u_{1}, u_{2})\approx a^{m-1}(u_{2m-1}, u_{2m})
$$
Apparently the edge states with $c_1 c_2=0$ are included in this category and we will show that there exist other localized states with $|c_1|\ll |c_2|$ or $|c_2|\ll |c_1|$. Focusing on edge states, we first assume that $1-|a|\sim O(1)$ (then $|a-1|\sim O(1)$ and $|a+1|\sim O(1)$) and the case $a\approx \pm 1$ will be discussed separately.

In the same spirit of section~\ref{sec:2}, if $|a|<1$ we write the left end as
$
(u_{1},u_{2})^{\top}=c_{1}\vec{v}_{1}+c_{2}\vec{v}_{2}
$,
then the right end yields $(u_{2n-1},u_{2n})^{\top}=c_{1}a^{n-1}\vec{v}_{1}+c_{2}a^{1-n}\vec{v}_{2}$. 
Substituting these into $-\omega^2 u=\mathcal{L} u$, we get the following equations
\begin{eqnarray}
\label{eq:boundary conditions_1}
  \frac{u_{1}}{u_{2}}&
  =&\frac{c_{1}v_{11}+c_{2}v_{12}}
  {\sigma c_{1}v_{12}+\sigma c_{2}v_{11}}=\frac{k_{1}}
  {k_{1}+k_{3,1}-\omega^{2}}
  \\
  \label{eq:boundary conditions_2}
  \frac{u_{2n-1}}{u_{2n}}&
  =&\frac{c_{1}v_{11}a^{n-1}+c_{2}v_{12}a^{1-n}}
  {\sigma c_{1}v_{12}a^{n-1}
  +\sigma c_{2}v_{11}a^{1-n}}
  =\frac{k_{1}+k_{3,2}-\omega^{2}}{k_{1}}
\end{eqnarray}
where $\omega$ is defined in (\ref{eq:omega}).
Without loss of generality, we assume $c_{1}=1$ (the case $c_1=0$ can be studied by assuming $c_2=1$ due to the symmetry) for simplicity. Then \eqref{eq:boundary conditions_1} becomes
\begin{equation}
  (k_{1}+k_{3,1}-\omega^{2})v_{11}
  +(k_{1}+k_{3,1}-\omega^{2})c_{2}v_{12}
  =\sigma k_{1}v_{12}+\sigma k_{1}c_{2}v_{11}.
  \label{eq:1}
\end{equation}
At first, we consider the special case with $c_{2}=0$, which corresponds to a genuine edge state localized at the left end. Substituting \eqref{eq:eigenvector} into \eqref{eq:1}, then we obtain the equation of $a$ for the solution with $c_2=0$ (this specific $a$ is denoted by $\tilde{a}$):
\begin{equation}
  (k_{3,1}-k_{2})^{2}k_{1}\tilde{a}^{2}
  +[k_{2}(k_{3,1}-k_{2})^{2}
  -k_{2}^{3}]\tilde{a}-k_{1}k_{2}^{2}=0
\label{eq:5}  
\end{equation}
Here \eqref{eq:5} becomes degenerate and belongs to a special situation when $k_{3,1}=k_{2}$, while $|k_{3,1}-k_2|=k_2$ leads to another special case with $\tilde{a}=\pm 1$. In the generic case with $(k_{3,1}-k_2)[(k_{3,1}-k_2)^2-k_2^2]\neq 0$, \eqref{eq:5} always has one root $|\tilde{a}_1|<1$ and another root $|\tilde{a}_2|>1$. In what follows, we will discuss the linear ``edge states'' in three cases based on whether $k_{3,1}\approx k_{2}$ or $|k_{3,1}-k_2|\approx k_2$.

\subsection{Special Case: $k_{3,1}\approx k_{2}$, $a\approx -\frac{k_1}{k_2}$}
\label{subsec:k31_approx_k2}
\subsubsection{Subcase: $k_{3,1}=k_2$}

First we assume $k_{3,1}=k_2$, then equation \eqref{eq:5} yields $\tilde{a}=-\frac{k_{1}}{k_{2}}$ which corresponds to the the edge state with $c_{2}=0$. However, after substituting these into \eqref{eq:boundary conditions_2}, we find that the equation will only hold if $k_{3,2}=\infty$. In other words, the edge state with $k_{3,1}=k_2$ and $c_{2}=0$ only exists in semi-infinite chains but not in realistic finite chains with $k_{3,2}<\infty$.
For this reason, we define $\tilde{k}_{3,2}
 =\frac{1}{k_{3,2}}$ so that $c_2=0$ now leads to $\tilde{k}_{3,2}=0$. Accordingly equation \eqref{eq:boundary conditions_2} becomes
  \begin{equation}
 \frac{v_{11}a^{n-1}+c_{2}v_{12}a^{1-n}}
  {\sigma v_{12}a^{n-1}
  +\sigma c_{2}v_{11}a^{1-n}}
  =\frac{\tilde{k}_{3,2}k_{1}+1-\tilde{k}_{3,2}\omega^{2}}{\tilde{k}_{3,2}k_{1}}
  \label{eq:boundary conditions_2 k32}
  \end{equation}
To study the "edge states" in long ($n\gg 1$) finite chains, we perturb $\tilde{a}$ as $a=-\frac{k_1}{k_2}+\Delta a$ where $0<-\Delta a\ll 1$. This implies that 
\begin{equation}
\label{eq:eigenvector_app}
v_{12}=\sqrt{k_2\Delta a}=i\sqrt{-k_2\Delta a}, \quad v_{11}=\sqrt{\frac{k_1^2-k_2^2}{k_1}}+\mathcal{O}(\Delta a)\approx i\sqrt{\frac{k_2^2-k_1^2}{k_1}}.
\end{equation}
Plugging these into \eqref{eq:boundary conditions_1}, we obtain 
\begin{equation}
\label{eq:3}
\begin{split}
  c_{2}(a)&=-\frac{(k_{3,1}-k_2)\sqrt{k_1+k_2/a}+\sigma k_2/a\sqrt{k_1+k_2 a}}{(k_{3,1}-k_2)\sqrt{k_1+k_2 a}+\sigma k_2 a\sqrt{k_1+k_2/a}} \\
  &=-\frac{\sqrt{k_{1}+k_{2}a}}{a^{2}
  \sqrt{k_{1}+k_{2}/a}}=
  -\frac{k_{2}^{2}}{k_{1}^{2}}
  \sqrt{\frac{k_{1}k_{2}\Delta a}
  {k_{1}^{2}-k_{2}^{2}}}+\mathcal{O}
  ((\Delta a)^{\frac{3}{2}}).
\end{split}
\end{equation}
On the other hand, ${k}_{3,2}$ can be explicitly expressed from \eqref{eq:boundary conditions_2 k32} as 
\begin{equation}
\label{eq:k32_original}
k_{3,2}=\omega^2-k_1+\sigma k_1\frac{v_{11}+v_{12}c_2 a^{2-2n}}{v_{12}+v_{11}c_2 a^{2-2n}}
\end{equation}
or 
\begin{equation}
\label{eq:k32}
   \tilde{k}_{3,2} =
   \frac{\sigma v_{12}a^{2n-2}+\sigma c_{2}v_{11}}
   {k_{1}(v_{11}a^{2n-2}+c_{2}v_{12})-(k_{1}-\omega^{2})
   (\sigma v_{12}a^{2n-2}+\sigma c_{2}v_{11})}
\end{equation}
where $\omega^2\approx k_1+k_2+\sigma \sqrt{\frac{k_2\Delta a(k_1^2-k_2^2)}{k_1}}$. Here $\tilde{k}_{3,2}$ or $k_{3,2}$ is a function of $a$ and particularly $\tilde{k}_{3,2}(-\frac{k_{1}}{k_{2}})=0$. 
Under the assumptions $n\gg 1$ and $|\Delta a|\ll 1$ (hence $a^{n}=(-\frac{k_1}{k_2}+\Delta a)^n\ll 1$), the following results can be derived:
\begin{eqnarray}
&&|v_{12}a^{2n-2}|\sim\mathcal{O}(a^{2n-2}\sqrt{\Delta a})\ll |c_2v_{11}|\sim\mathcal{O}(\sqrt{\Delta a}), \quad\quad\\
    &&\tilde{k}_{3,2}
    \approx
    \frac{\sigma c_{2}v_{11}}{k_{1}v_{11}a^{2n-2}+k_{2}
    \sigma c_{2}v_{11}}
    \approx
    \frac{-\sigma\frac{k_{2}^{2}\sqrt{-k_{2}\Delta a}}{k_{1}^{2}}}
    { k_{1}\sqrt{\frac{k_{2}^{2}-k_{1}^{2}}{k_{1}}}a^{2n-2}-\sigma k_2\frac{k_{2}^{2}\sqrt{-k_{2}\Delta a}}{k_{1}^{2}}}.\quad\quad
    \label{tilde k32}
\end{eqnarray}
Since the order of $\tilde{k}_{3,2}$ in \eqref{tilde k32} depends on the relation between $a^{2n-2}$ and $\sqrt{\Delta a}$, then we consider the following three cases:
\begin{itemize}
  \item $|a^{2n-2}|
  >>\sqrt{|\Delta a|}$
  
Now \eqref{tilde k32} yields 
\begin{equation}
  \tilde{k}_{3,2}
\approx-\frac{\sigma
\frac{k_{2}^{2}}
{k_{1}^{2}}\sqrt{-\Delta a}}
{k_{1}\sqrt{\frac{k_{2}^{2}-k_{1}^{2}}{k_{1}}}a^{2n-2}}\sim O(\frac{\sqrt{|\Delta a|}}{a^{2n-2}})\ll 1
\label{order of tilde k32 1}
\end{equation}
hence $k_{3,2}\sim\mathcal{O}(\frac{a^{2n-2}}{\sqrt{|\Delta a|}})\gg 1$. This is the case when $|\Delta a|$ is very small such that $|\Delta a|\ll (\frac{k_1}{k_2})^{4n-4}$.
Now the eigenstate $u$ reads
\begin{equation}
    (u_{2m-1},u_{2m})^{\top}
    =a^{m-1}\vec{v}_{1}
    +c_{2}a^{1-m}\vec{v}_{2},\quad
    (1\leq m\leq n)
    \label{iteration m-th}.
\end{equation}
We know
\begin{equation}
a^{m-1}\approx
(-\frac{k_{1}}{k_{2}})^{m-1};
\quad
|c_{2}a^{1-m}|\ll\mathcal{O}
((\frac{k_{1}}
{k_{2}})^{2n-m-1})
\end{equation}
hence $a^{m-1}\vec{v}_{1}$ plays a dominant role in the shape of the state. Therefore, the eigenstate is a localized mode similar to the state with $c_{2}=0$, namely an edge state localized at the left.
However the corresponding large $k_{3,2}$ may be difficult to implement in real chains.

  \item $a^{2n-2}\sim\mathcal{O}(\sqrt{|\Delta a|})$

As $|\Delta a|$ grows, we encounter this scenario where $|\Delta a|\sim (\frac{k_1}{k_2})^{4n-4}$. Then it can be inferred that $k_{3,2}\approx k_2+\sigma \frac{k_1}{c_2 a^{2-2n}}\leq O(1)$ and $k_{3,2}\not\approx k_2$.

Similar to the discussion on \eqref{iteration m-th}, now we have
\begin{equation}
a^{m-1}\approx
(-\frac{k_{1}}{k_{2}})^{m-1};
\quad
c_{2}a^{1-m}\sim\mathcal{O}
((-\frac{k_{1}}
{k_{2}})^{2n-m-1})
\end{equation}
thus $a^{m}\vec{v}_{1}$ still be the dominant part in the eigenstate for most of the time. It should be noted that at the right end of the chain ($m=n$) the order of $c_{2}a^{1-m}$ just catches up with $a^{m-1}$ thus the state $u$ still looks like the case $c_2=0$. 

  \item  $|a^{2n-2}
|<<\sqrt{|\Delta a|}$

Similarly \eqref{eq:k32_original} implies
\begin{equation}
{k}_{3,2}
    \approx k_2+\sigma[ k_2^2\sqrt{\frac{k_2\Delta a}{k_1(k_1^2-k_2^2)}}+\frac{k_1}{c_2 a^{2-2n}} ]\approx k_2
   \label{order of Delta k32 2}
\end{equation}
As $a$ deviates from $\tilde{a}=-\frac{k_1}{k_2}$, $c_2$ in \eqref{iteration m-th} also grows and the state $u$ accordingly looks more localized at both ends. Moreover, as $k_{3,2}$ approaches $k_2$, the order of $\Delta a$ becomes close to $a^{2n-2}$. In particular, when $k_{3,1}=k_2=k_{3,2}$ the eigenstate is symmetric such that $c_2^2=a^{2n-2}$(see Fig.~\ref{fig:edge state}(a) as an example). In order to focus on states localized at one end, we restrict ourselves to previous two cases (especially the case $a^{2n-2}\sim\mathcal{O}(\sqrt{|\Delta a|})$ where the value of $k_{3,2}$ is easier to achieve). For example, a left edge states with $k_{3,1}=k_2$ in a chain of length $100$ is shown in Fig.~\ref{fig:edge state}(b). 
\end{itemize}

\begin{figure}[!htp]
\begin{minipage}{0.5\linewidth}
\leftline{(a)}
\centerline{\includegraphics[height = 5cm, width = 6cm]{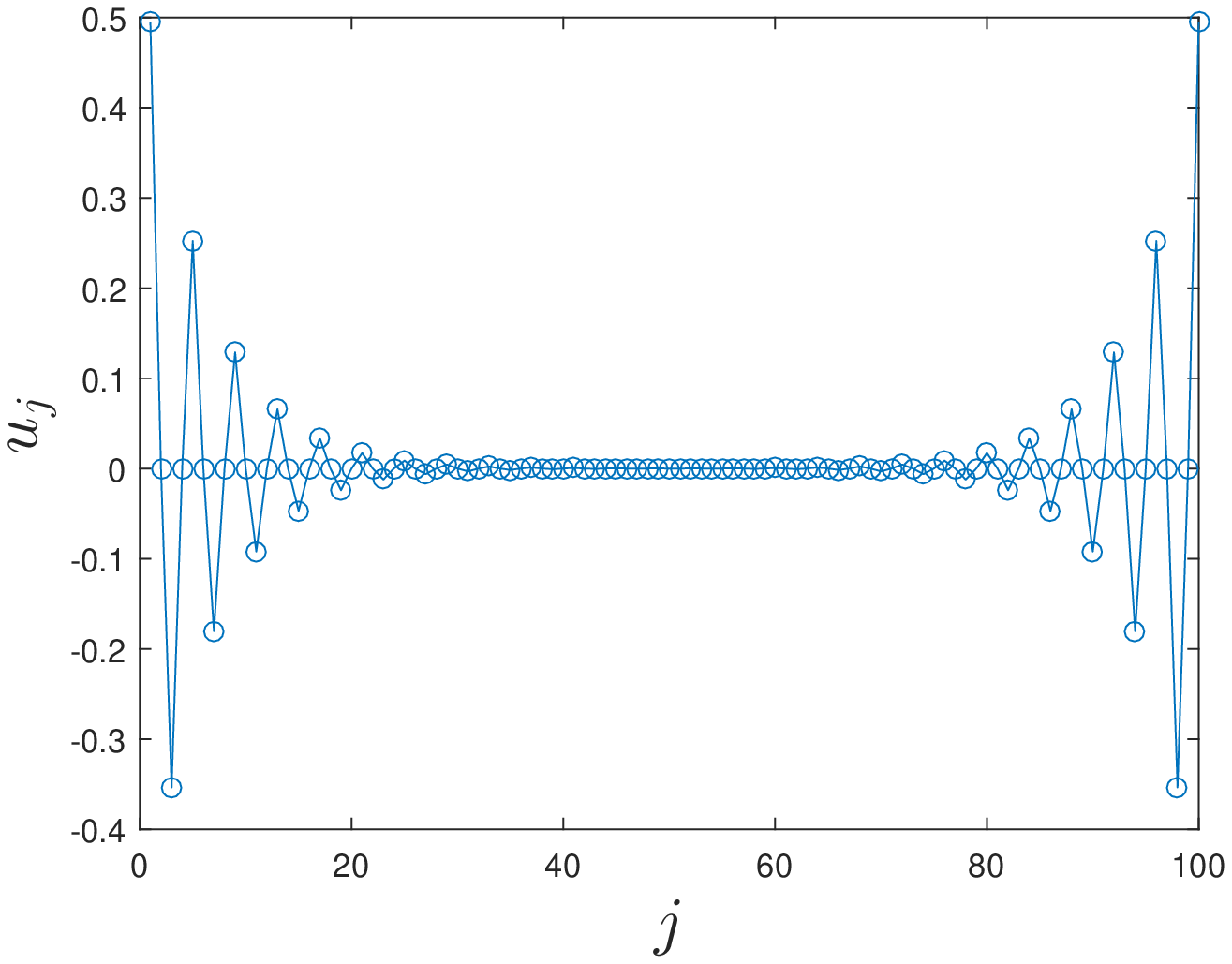}}
\end{minipage}
\hfill
\begin{minipage}{0.5\linewidth}
\leftline{(b)}
\centerline{\includegraphics[height = 5cm, width = 6cm]{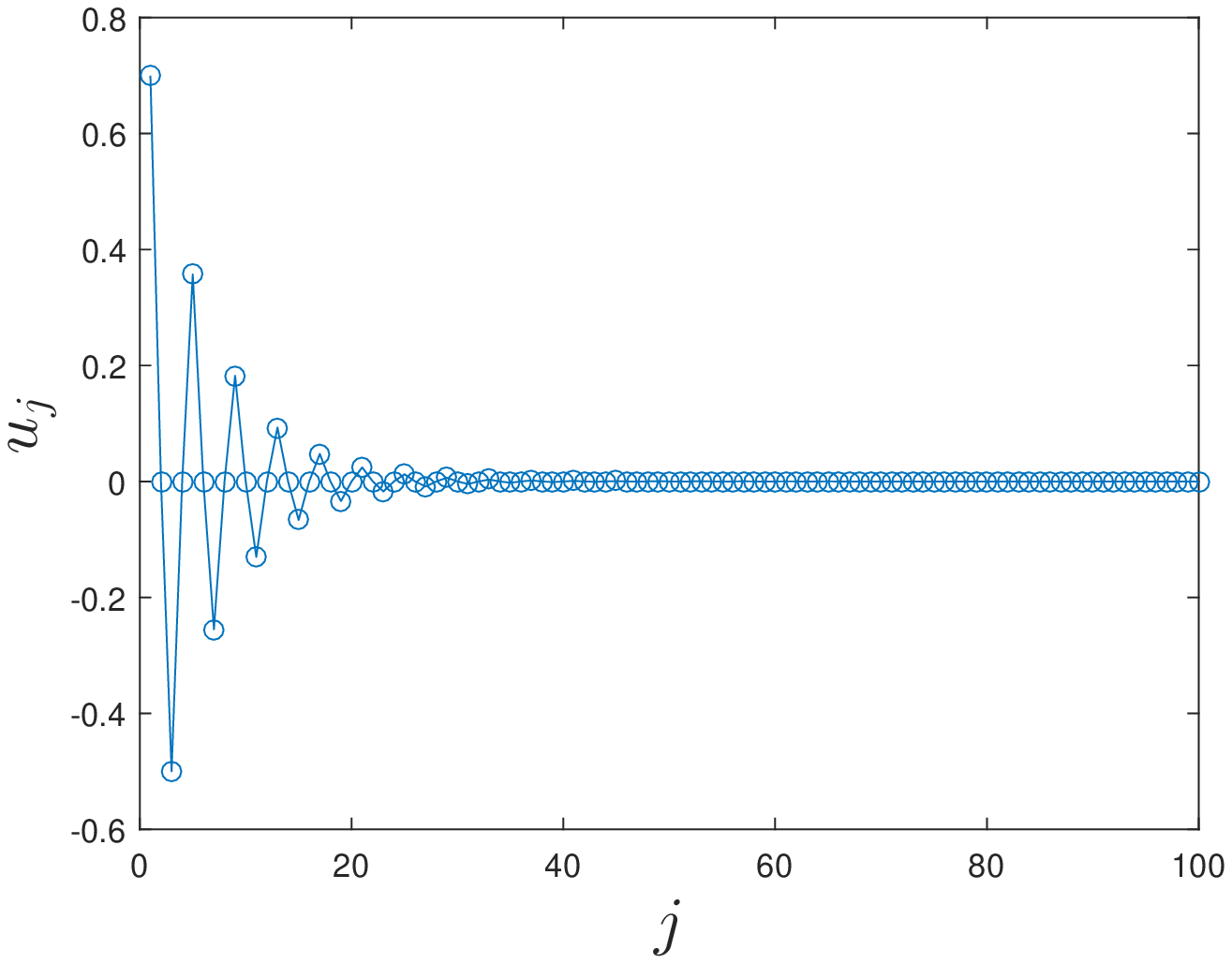}}
\end{minipage}

\begin{minipage}{0.5\linewidth}
\leftline{(c)}
\centerline{\includegraphics[height = 5cm, width = 6cm]
{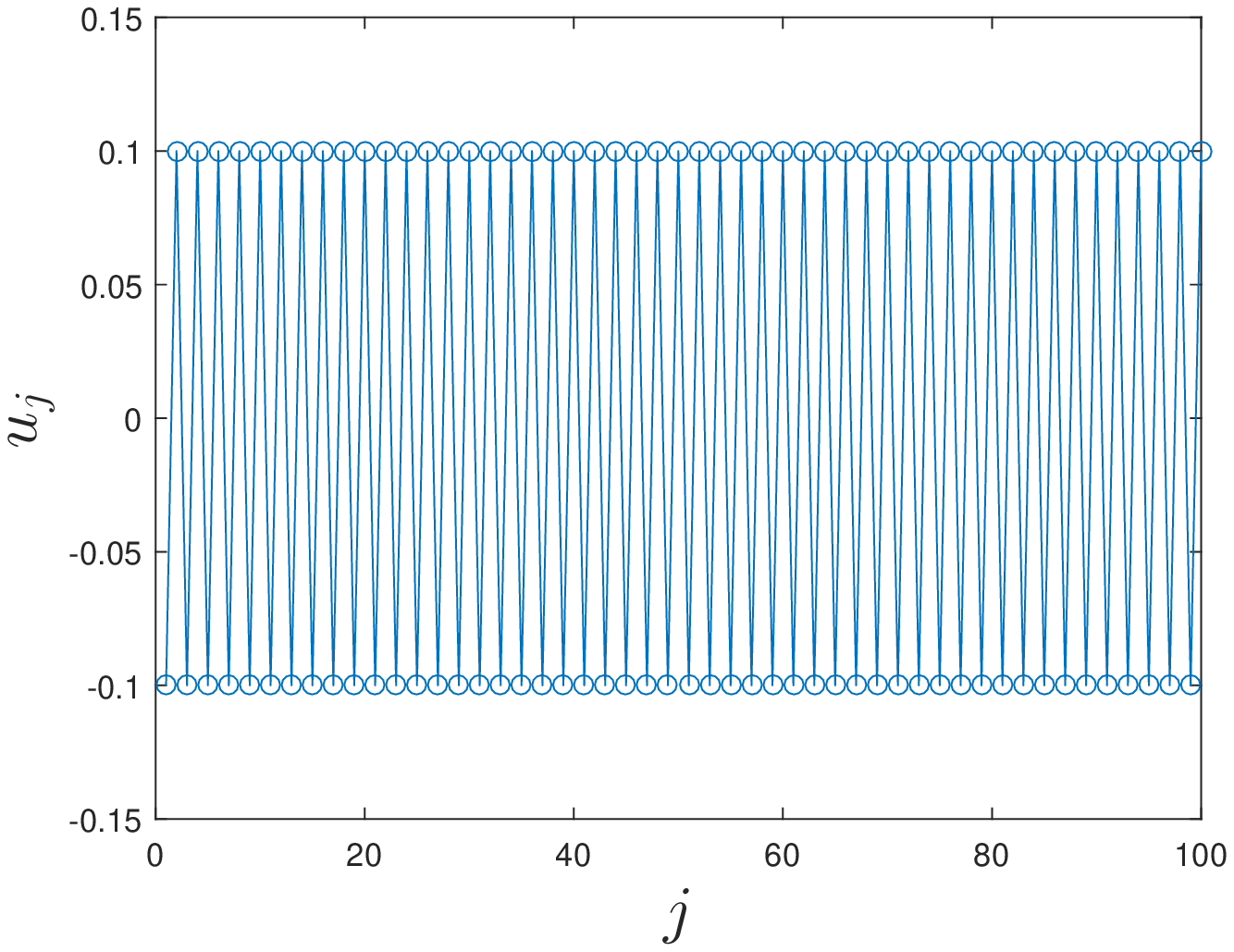}}
\end{minipage}
\hfill
\begin{minipage}{0.5\linewidth}
\leftline{(d)}
\centerline{\includegraphics[height = 5cm, width = 6cm]{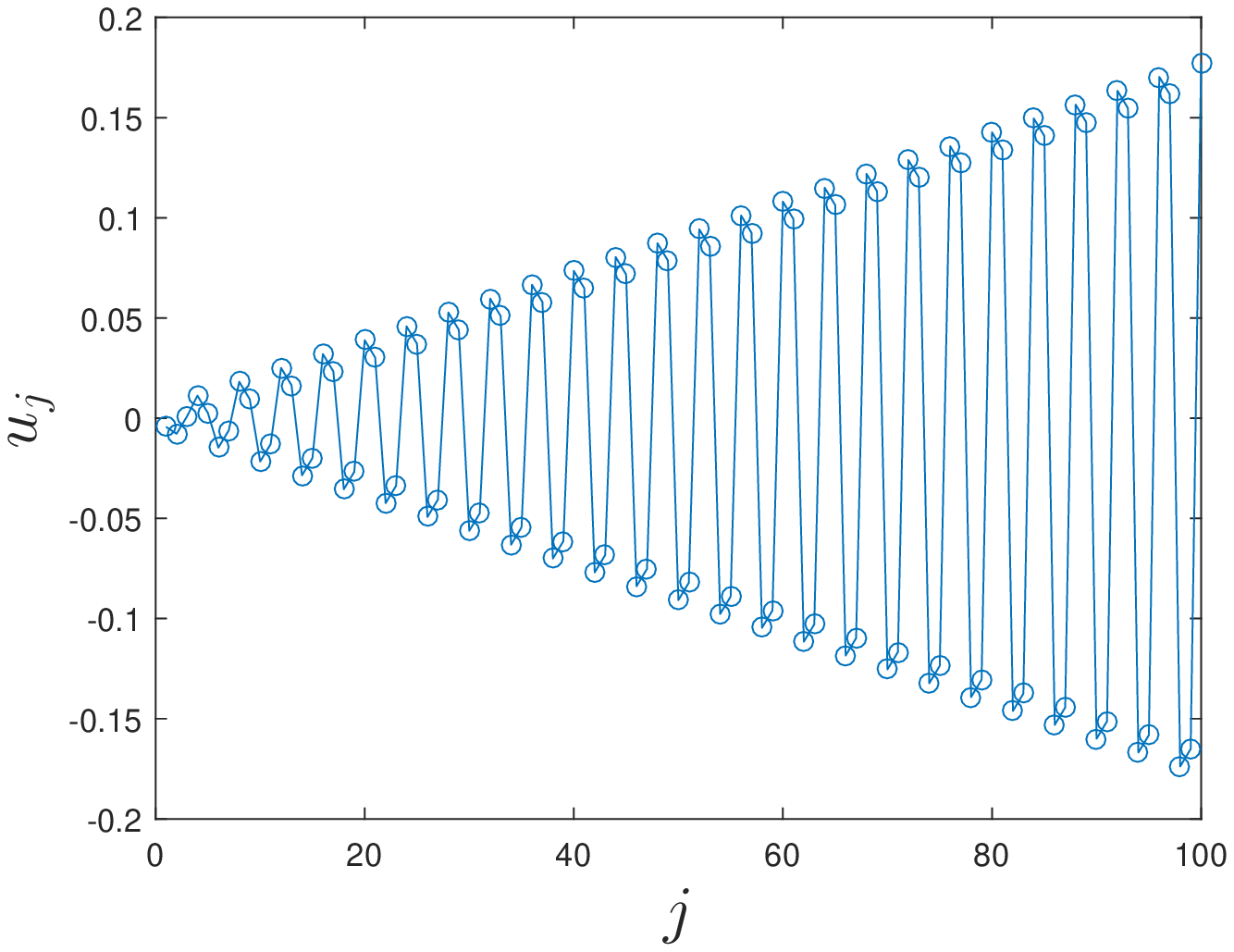}}
\end{minipage}

\begin{minipage}{0.5\linewidth}
\leftline{(e)}
\centerline{\includegraphics[height = 5cm, width = 6cm]
{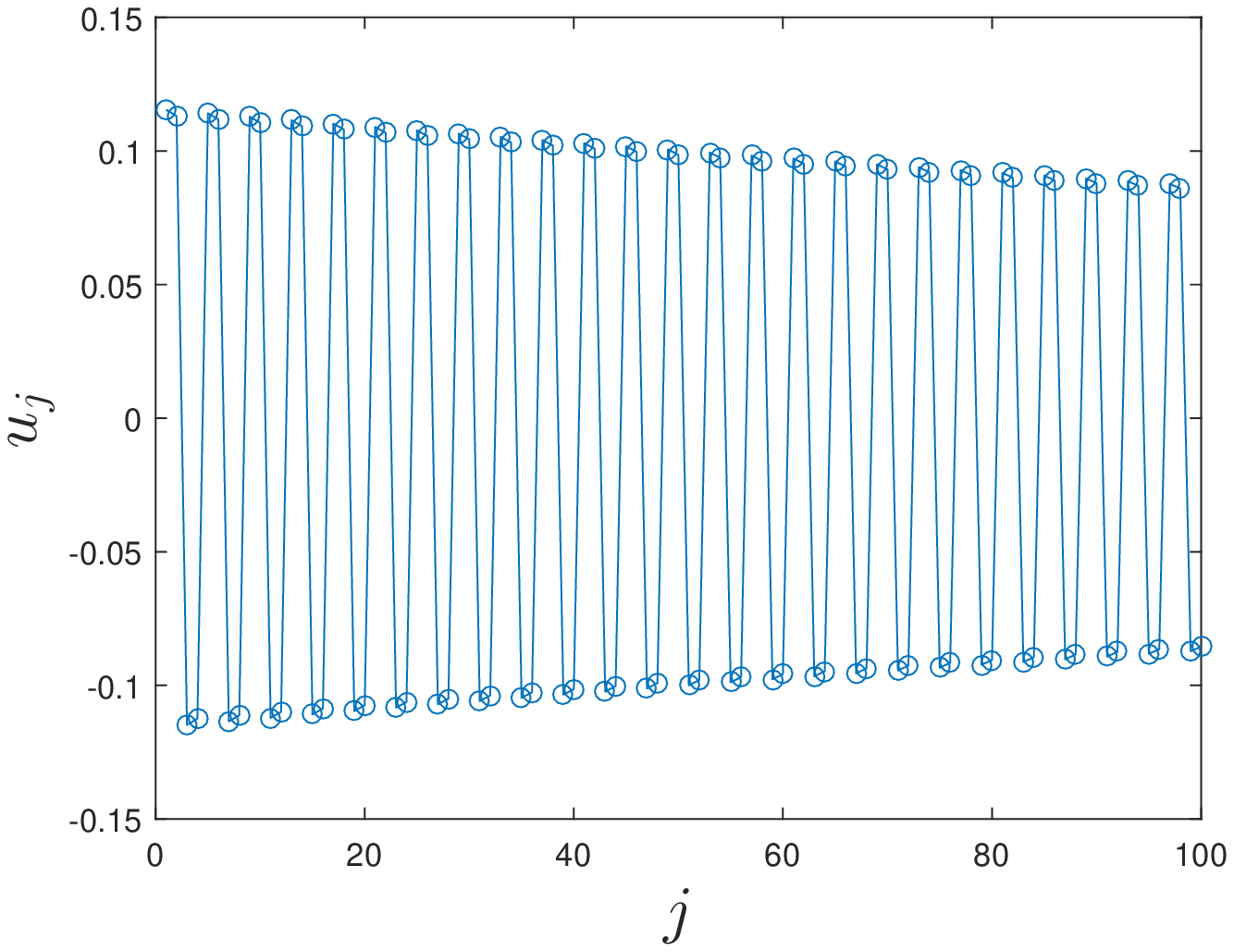}}
\end{minipage}
\hfill
\begin{minipage}{0.5\linewidth}
\leftline{(f)}
\centerline{\includegraphics[height = 5cm, width = 6cm]{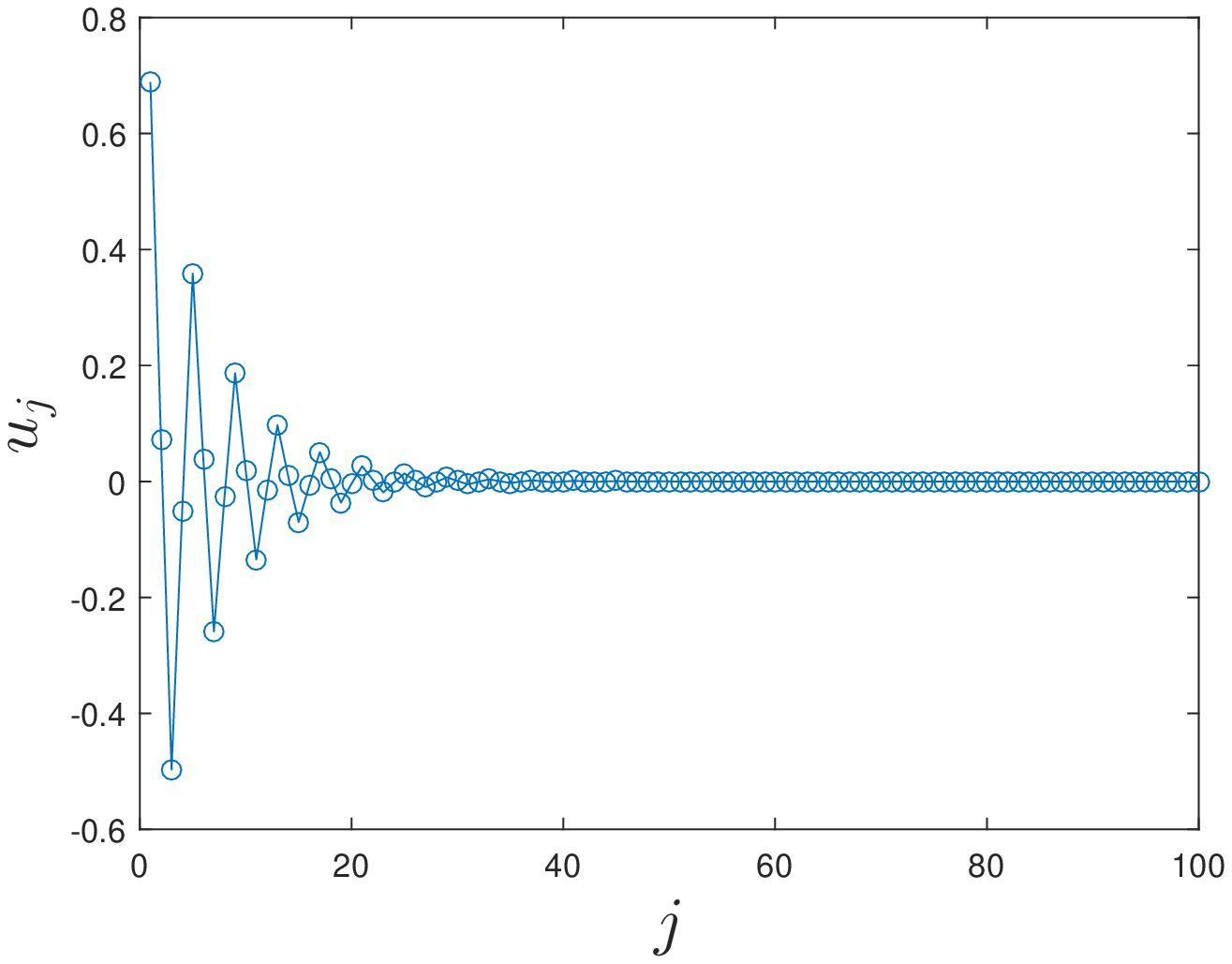}}
\end{minipage}

\caption{Here we plot some representative eigenstates of the linear diatomic chain with $n=50$ (length $100$), $k_{1}=1$ and $k_{2}=1.4$. The panel (a) shows an eigenstate localized at both ends when $k_{3,1}=k_{3,2}=k_{2}$. The panel (b) shows a left edge state for case \ref{subsec:k31_approx_k2} with 
$k_{3,1}=k_{2}$ and 
$k_{3,2}=1.2$. The panel (c) in the middle left (panel (d) in the middle right) shows a ``band edge state'' with $\omega^{2}=2k_{2}$ and $k_{3,1}=2k_{2}=k_{3,2}$ ($\omega^{2}=2k_{1}$,  $k_{3,1}=2k_{2}$ and $k_{3,2}=\frac{2k_{1}k_{2}}{(2n-1)(k_{2}-k_{1})+k_{2}}$) for case~\ref{subsec: k31_approx_2k2}. The panel (e) shows a state decays ``slowly'' with $a\approx-0.9941582\approx-1$ when $2.78=k_{3,1}\approx 2k_{2}\approx k_{3,2}=2.82$. The panel (f) shows a left edge state for the
generic case \ref {generic case} when $k_{3,1}=1.6$ and $k_{3,2}=2$.}
\label{fig:edge state}
\end{figure}

\subsubsection{Subcase: $0<|k_{3,1}-k_2|\ll 1$}
\label{subsubsec:k31_approx_k2}

If $k_{3,1}$ is not exactly $k_2$, we write $k_{3,1}=k_2+\delta k_{3,1}$ and $\tilde{a}=-\frac{k_1}{k_2}+\delta a$ (here $\delta a<0$). Based on equation~(\ref{eq:5}), it can be obtained that 
\begin{equation}
\delta k_{3,1}\approx \sigma k_2^2\sqrt{\frac{k_2\delta a}{k_1(k_1^2-k_2^2)}}\sim O(|\delta a|^{1/2}).
\end{equation}
Different from the subcase $k_{3,1}=k_2$, now $c_2=0$ and $a=\tilde{a}$ in (\ref{eq:boundary conditions_2}) yields an edge state with $k_{3,2}\sim O(\frac{1}{|\delta a|^{1/2}})\sim O(\frac{1}{|\delta k_{3,1}|})$ that can be implemented in a finite chain.

In order to consider edge states with $c_2\neq 0$, we write $$a=\tilde{a}+\delta\tilde{a}=-\frac{k_1}{k_2}+\delta a+\delta\tilde{a}=-\frac{k_1}{k_2}+\Delta a$$ and (\ref{eq:eigenvector_app}) still holds.
Accordingly 
(\ref{eq:boundary conditions_1}) yields 
\begin{equation}
\label{eq:c_2}
c_2\approx \frac{\delta k_{3,1}\sqrt{\frac{k_1^2-k_2^2}{k_1}}-\sigma\frac{k_2^2}{k_1}\sqrt{k_2\Delta a}}{\sigma\sqrt{k_1(k_1^2-k_2^2)}}\approx\frac{k_2^2}{k_1}\sqrt{\frac{k_2}{k_1(k_1^2-k_2^2)}}(\sqrt{\delta a}-\sqrt{\Delta a}).
\end{equation} 

By comparing $|\delta a|$ and $|\delta\tilde{a}|$, we consider different situations to obtain the dependence of $\Delta a$ on $k_{3,1}$ and $k_{3,2}$. The main results are summarized in Lemma~\ref{lemma: k31_approx_k2} and one can consult Appendix~\ref{sec:special cases proof} for more details.

\begin{lemma}
\label{lemma: k31_approx_k2}
In \eqref{eq:boundary conditions_1}\eqref{eq:boundary conditions_2} with $n\gg 1$, if $k_{3,1}=k_2+\delta k_{3,1}$ and $|\delta k_{3,1}|\ll 1$, then $k_{3,2}\not\approx k_2$ is a necessary condition to the existence of a left edge state with $a\approx -\frac{k_1}{k_2}$ and $|c_2 a^{2-2n}|\leq O(1)$.
\end{lemma}

\begin{remark}
\label{remark: k31_approx_k2}
Utilizing continuity of $a$ with respect to $k_{3,2}$ (when $k_{3,1}$ is fixed) in the proof of Lemma~\ref{lemma: k31_approx_k2}, we can further conclude that $k_{3,1}\approx k_2 \not \approx k_{3,2}$ suffices to imply the existence of a left edgestate with $a\approx -\frac{k_1}{k_2}$ and $|c_2 a^{2-2n}|\leq O(1)$. On the other hand, $k_{3,1}\approx k_2\approx k_{3,2}$ leads to an eigenstate with $a\approx -\frac{k_1}{k_2}$ and $|c_2 a^{2-2n}|\gg O(1)$, which is possibly localized at both ends. With that being said, if $k_{3,1}\not\approx k_2\not\approx k_{3,2}$, the edge states (when exist) have $a\not\approx -\frac{k_1}{k_2}$ and $\omega^2\not\approx k_1+k_2$.
\end{remark}

\subsection{Special Case: $|k_{3,1}-k_2|\approx k_2$, $a\approx \pm 1$}
\label{subsec: k31_approx_2k2}

\subsubsection{Subcase: $a=\pm 1$}
\label{subsubsec: a=pm1}
In order to study the states with $|a|<1$ and $a\approx \pm 1$, we first characterize the states with exactly $a=\pm 1$. Since the frequencies of these states with $a=\pm 1$ are exactly on the edges of the bands, we name them as ``band edge state''(see Fig.~\ref{fig:edge state}(c)(d) as two examples). 
When ${a}=1$, $\mathcal{T}$ has an eigenvector $\vec{v}_{1}=(1,\sigma)^{\top}$ and a generalized eigenvector $\vec{v}_{2}=(1,-\sigma)^{\top}$. Then
the eigenstate $u$ of $\mathcal{L}$ has the form
\begin{equation}
\label{relation:iteration1}
\begin{split}
  (u_{1},u_{2})^{\top}
&=c_{1}\vec{v}_{1}+c_{2}\vec{v}_{2};\\
  (u_{2n-1},u_{2n})^{\top}
&=c_{1}\vec{v}_{1}+c_{2}(n-1)(\frac{-2(k_{1}+k_{2})}{k_{2}})\vec{v}_{1}+c_{2}\vec{v}_{2}
\end{split}
\end{equation}
where $\omega^2=k_1+k_2-\sigma\sqrt{k_1+k_2 a}\sqrt{k_1+k_2/a}=(k_1+k_2)(1-\sigma)$. 
Similarly, ${a}=-1$ corresponds to the eigenstate $u$ of $\mathcal{L}$ with
\begin{equation}
\label{relation:iteration2}
\begin{split}
(u_{1},u_{2})^{\top}
&=c_{1}\vec{v}_{1}+c_{2}\vec{v}_{2};\\
(u_{2n-1},u_{2n})^{\top}
&=(-1)^{n-1}[c_{1}\vec{v}_{1}-c_{2}(n-1)\frac{2(k_{2}-k_{1})}{k_{2}}\vec{v}_{1}+c_{2}\vec{v}_{2}].
\end{split}
\end{equation}
Substituting these into $-\omega^2 u=\mathcal{L}u$, we obtain the conditions of $k_{3,1}$ and $k_{3,2}$ for ``band edge states'' in the following Remark~\ref{remark: band edge states} (see Appendix~\ref{sec:special cases proof} for more detailed derivation).

\begin{remark}
\label{remark: band edge states}
Suppose the system bears a ``band edge state'' with $a=\pm 1$ and $\omega^2=k_1+k_2-\sigma(k_1+k_2 a)$.
    \begin{itemize}
        \item If $k_{3,1}$ is not close enough to $0$ or $2k_2$, say $|k_{3,1}-k_2(1-\sigma a)|\gg \frac{1}{n}$, then $k_{3,2}\approx k_2(1-\sigma a)- \sigma \frac{k_1k_2}{n(k_2+a k_1)}$.
        \item If $k_{3,1}$ is very close to $0$ or $2k_2$, say $|k_{3,1}-k_2(1-\sigma a)|\ll \frac{1}{n}$, then $k_{3,2}-k_2(1-\sigma a)\approx k_2(1-\sigma a)-k_{3,1}\ll O(\frac{1}{n})$.
        \item If $|k_{3,1}-k_2(1-\sigma a)|\sim O(\frac{1}{n})$ and $k_{3,1}-k_2(1-\sigma a)\not\approx -\sigma\frac{k_1 k_2}{n(k_2+a k_1)}$, then $|k_{3,2}-k_2(1-\sigma a)|\sim O(\frac{1}{n})$ and $k_{3,2}-k_2(1-\sigma a)\not\approx -\sigma\frac{k_1 k_2}{n(k_2+a k_1)}$.
    \end{itemize}
\end{remark}
That is to say, at least one of $k_{3,1}$ and $k_{3,2}$ should be close to $0$ or $2k_2$ for any ``band edge state'' with $a=1$ or $a=-1$ to exist. Therefore it suffices to characterize all the ``band edge states'' by considering $|k_{3,1}-k_2|\approx k_2$ and utilizing the symmetry (hence also $|k_{3,2}-k_2|\approx k_2$).

\subsubsection{Subcase: $a\approx \pm 1$}
Now we consider the more generic situation $a\approx \pm 1$. In particular, the case with $\omega^2\approx 2k_2$ ($a\approx -1$ and $\sigma=1$) will be discussed and the rest follow similar strategies. Here we define $a=-1+\Delta a\approx -1$, $k_{3,1}=2k_2+\delta k_{3,1}$ and $k_{3,2}=2k_2+\Delta k_{3,2}$. Plugging these into \eqref{eq:boundary conditions_1}\eqref{eq:boundary conditions_2}, we obtain the relation between $\delta k_{3,1}$, $\Delta k_{3,2}$ and $\Delta a$. which leads to the following remark (see proof in Appendix~\ref{sec:special cases proof}).

\begin{remark}
\label{remark:k31_approx_2k2}
When $|k_{3,1}-2k_2|\geq O(1) \leq |k_{3,2}-2k_2|$, there does not exist any eigenstate $u$ near $\omega^2=2k_2$ with $-1<a\approx -1$.     
\end{remark}
In the same spirit, the other band edges can be discussed and we find that the existence of eigenstates with $a\approx \pm 1$ requires $|k_{3,1}-k_2|\approx k_2$ or $|k_{3,2}-k_2|\approx k_2$.
On the other hand, the states with $a\approx \pm 1$ decay too slow to be 
``localized enough'' ($|c_2 a^{2-2n}|\leq O(1)$), although their corresponding frequencies are outside the special bands(see Fig.~\ref{fig:edge state}(e) as an example). In this work we will place our attention on the more generic scenario where $||a|-1||\geq O(1)$ and $|k_{3,1}-k_2|\not\approx k_2\not\approx |k_{3,2}-k_2|$.

\subsection{Generic case: $k_{3,1}\not\approx k_{2}$ and $|k_{3,1}-2k_{2}|\not\approx k_2$}

\label{generic case}

In this case, equation \eqref{eq:5} always has a root $\tilde{a}$ which is away from $\pm 1$ and $-\frac{k_1}{k_2}$. Again we perturb $\tilde{a}$ as $a=\tilde{a}+\Delta a$ where $|\Delta a|\ll 1$.
Since $v_{11}\sim\mathcal{O}(1)\sim v_{12}$ when $a=\tilde{a}$,
$v_{11}$ and $v_{12}$ can be written as
\begin{equation}
  v_{11}(a) = \sqrt{k_{1}+k_{2}/\tilde{a}}
  +\mathcal{O}(\Delta a); \quad 
  v_{12}(a) = \sqrt{k_{1}+k_{2}\tilde{a}}
  +\mathcal{O}(\Delta a).
  \label{eq:v_11,12}
\end{equation}
In addition, $c_2$ can be obtained from \eqref{eq:boundary conditions_1} as
\begin{equation}
    c_{2}(a) = \frac{-(k_{3,1}-k_2)\sqrt{k_1+k_2/a}-\sigma k_2/a\sqrt{k_1+k_2 a}}{\sigma k_2 a\sqrt{k_1+k_2/a}+(k_{3,1}-k_2)\sqrt{k_1+k_2 a}}
    \approx c_2'(\tilde{a})(\Delta a)\sim O(\Delta a).
  \label{eq:c_2+}
\end{equation}
Substituting this into \eqref{eq:boundary conditions_2}, we can express $k_{3,2}$ as 
\begin{equation}
    k_{3,2}(a) =
    \frac{k_{1}(v_{11}+c_{2}a^{2-2n}v_{12})}
    {\sigma v_{12}+\sigma c_{2}a^{2-2n}v_{11}}
    +\omega^{2}-k_{1}.
\label{k32 of a}    
\end{equation}
Similar to the special case $k_{3,1}=k_2$, the value of $k_{3,2}$ changes as the order of $\Delta a$ varies:
\begin{itemize}
    \item $|a^{2n-2}|\gg|\Delta a|$
    
    In this case $|\Delta a|$ is so small that $|c_2 a^{2-2n}|\ll 1$. Therefore \eqref{k32 of a} now implies
    \begin{equation}
        k_{3,2}\approx k_2-\frac{k_2 \tilde{a}\sqrt{k_1+k_2/\tilde{a}}}{\sigma\sqrt{k_1+k_2 \tilde{a}}};
    \end{equation}
    
    \item $|a^{2n-2}|\ll|\Delta a|$
    
    When $|\Delta a|$ is relatively large, $k_{3,2}$ becomes
    \begin{equation}
        k_{3,2}\approx k_2-\frac{k_2 \sqrt{k_1+k_2\tilde{a}}}{\sigma\tilde{a}\sqrt{k_1+k_2/ \tilde{a}}};
    \end{equation}
    
    \item $|a^{2n-2}|\sim|\Delta a|$
    
    When $|c_2 a^{2-2n}|\sim\mathcal{O}(1)$, $k_{3,2}$ in \eqref{k32 of a} can fortunately take most of the values. Also taking the dominant role of $a^{m-1}v_1$ in \eqref{iteration m-th} into consideration, we will place emphasis on this situation (see Fig.~\ref{fig:edge state}(f) as an example). 
\end{itemize}

In this section, we find that the edge states are universal in long chains. To be more explicit, edge states with $|a|<1$ exist as long as $\{k_{3,1}, k_{3,2}\}$ are away from some special values. Moreover, for these edge states we know $|\Delta a|\sim {a}^{4n-4}$ for $a\approx -\frac{k_1}{k_2}$ and $|\Delta a|\sim {a}^{2n-2}$ for $a\not\approx -\frac{k_1}{k_2}$. In particular, we suppose the eigenstate $u^{(k)}$ localized at the left end as
\begin{equation}
\label{matrix:uk localized}
  \left(
  \begin{array}{c}
    u^{(k)}_{2j-1} \\
    u^{(k)}_{2j}
  \end{array}
  \right)
  =\left(
  \begin{array}{c}
    a^{j-1}\sqrt{k_1+k_2/a}+c_2a^{1-j}\sqrt{k_1+k_2a}\\
    \sigma (a^{j-1}\sqrt{k_1+k_2a}+c_2a^{1-j}\sqrt{k_1+k_2/a})
  \end{array}
  \right), \quad 1\leq j\leq n.
\end{equation}
The states localized at the right end can be defined in the same way. Suppose $(\omega^{(k)})^2$ is in the bandgap $(2k_1,2k_2)$ and is away from the band edges, then it can be assumed that
\begin{equation}
\label{eq: A1}
[A1]:\quad a\in(-1+\delta_1,-\frac{k_1}{k_2}), \quad  2k_1+\delta_2<(\omega^{(k)})^2<2k_2-\delta_2
\end{equation}
where $0<\delta_1\sim O(1)\sim \delta_2>0$. Accordingly we can calculate its norm as
\begin{equation}
|u^{(k)}|^{2}=\|u^{(k)}\|_{2}^{2}\approx|\frac{2k_1+k_2(a+1/a)}{1-a^2}|\sim\mathcal{O}(1).
\end{equation}

\section{Estimation for states with eigenfrequencies near the edges of the spectral band}
\label{sec:estimate}

In Section.~\ref{sec:edgestates}, we have studied the edge states with frequencies outside the bands in long linear diatomic chains. Now we move to the eigenstates with frequencies inside the bands and especially provide estimates for those with frequencies near the band edges. This knowledge will later enable us to investigate the continuation of states in the weakly nonlinear regime. 

First we notice that the eigenfrequencies of the linear localized states should be outside the frequency bands and this rule also holds for nonlinear localized states. Suppose with generic choice of $k_{3,1}>0$ and $k_{3,2}>0$, a linear diatomic chain usually bears two edge states. If the number of eigenfrequencies in optical band ($(\omega^{(j)})^2\in (2k_2, 2k_1+2k_2)$) and acoustic band ($(\omega^{(j)})^2\in (0, 2k_1)$) is $n_{1}$ and $n_{2}$ respectively, then we assume $n_{1}+n_{2}+2=2n$.

\subsection{Eigenfrequencies in the optical band}
\label{sec:optical band}
To start with, we assume $a\neq\pm 1$ and consider the eigenfrequencies in the optical band. Suppose  
\begin{equation}
2k_2<(\omega^{(1)})^2<(\omega^{(2)})^2<\dots<(\omega^{(n_1)})^2<2k_1+2k_2
\end{equation}
and $a=e^{i\theta} (\theta \in (\pi,2\pi))$, then the frequencies have the form
\begin{equation}
\omega^2=k_1+k_2-\sigma\sqrt{k_1^2+k_2^2+2k_1 k_2 \cos\theta}
\end{equation}
with $\sigma=-1$ and accordingly
\begin{equation}
\pi<\theta^{(1)}<\theta^{(2)}<\dots<\theta^{(n_1)}<2\pi.
\end{equation}
When $\theta\in (\pi, 2\pi)$, we define $\alpha\in (0, \frac{\pi}{2})$ such that 
\begin{equation}
\sqrt{k_1+k_2 e^{-i\theta}}=\rho e^{i\alpha}, \quad \sqrt{k_1+k_2 e^{i\theta}}=\rho e^{-i\alpha}.
\end{equation}
Then the eigenvectors of $\mathcal{T}$ in \eqref{eq:omega} can be written as
\begin{equation}
\begin{split}
  \vec{v}_{1}(a) =
  \left(
  \begin{array}{c}
    v_{11} \\
    \sigma v_{12}
  \end{array}
  \right)=
  \left(
  \begin{array}{c}
    \rho e^{i\alpha} \\
    \sigma \rho e^{-i\alpha}
  \end{array}
  \right); 
  \vec{v}_{2}(a) =
    \left(
  \begin{array}{c}
    v_{21} \\
   \sigma v_{22}
  \end{array}
  \right)=
  \left(
  \begin{array}{c}
    \rho e^{-i\alpha} \\
    \sigma \rho e^{i\alpha}
  \end{array}
  \right)
\end{split}
\label{eq:v1v2}
\end{equation}
Since the states $u^{(j)}$ are real, we suppose 
\begin{equation}
  c_{1}=re^{i\beta},\quad c_{2}=re^{-i\beta}
  \label{eq:c1c2}
\end{equation}
and substituting these into the system \eqref{eq:boundary conditions_1}, \eqref{eq:boundary conditions_2} to obtain
\begin{eqnarray}
  \label{eq:boundary conditions_b1}
  \frac{\cos(\alpha+\beta)}{-\cos(\beta-\alpha)} &=&  \frac{k_{1}}{k_{1}+k_{3,1}-\omega^{2}}\\
  \frac{\cos(\alpha+\beta+(n-1)\theta)}{-\cos(\beta-\alpha+(n-1)\theta)} &=& \frac{k_{1}+k_{3,2}-\omega^{2}}{k_{1}}
\label{eq:boundary conditions_b2}
\end{eqnarray}
Since \eqref{eq:boundary conditions_b1}\eqref{eq:boundary conditions_b2} are invariant under the change $\beta\to\beta+\pi$, it is enough to consider $\beta$ in the range $\beta\in [0,\pi)$. Focusing on the eigenfrequencies near the lower edge of the optical band, we assume $\theta=\pi+\Delta\theta\approx \pi$ hence
\begin{equation}
\alpha=\frac{\pi}{2}+\Delta\alpha\approx\frac{\pi}{2}+\frac{k_2}{2(k_1-k_2)}\Delta \theta,  \quad \omega^2\approx 2 k_2+\frac{k_1 k_2 (\Delta\theta)^2}{2(k_2-k_1)}.
\end{equation}
At the same time, \eqref{eq:boundary conditions_b1} can be rewritten as
\begin{equation}
\label{eq:tan_beta}
\tan\beta=\frac{2k_1+k_{3,1}-\omega^2}{k_{3,1}-\omega^2}\cot\alpha=\frac{2k_1+k_{3,1}-\omega^2}{\omega^2-k_{3,1}}\tan\Delta\alpha. 
\end{equation}
\begin{itemize}
\item Case: $k_{3,1}\not\approx 2k_2$\\
In this case, $|\tan\beta|\leq O(|\Delta\alpha|)\sim O(|\Delta\theta|)\ll O(1)$ hence $\beta\approx 0$ or $\beta\approx \pi$. Suppose we define
\begin{equation}
\delta_1\beta= 
\begin{cases}
\beta, & {\rm if}~~0\leq\beta<\frac{\pi}{2} \\
\beta-\pi, & {\rm if}~~\frac{\pi}{2}<\beta<\pi 
\end{cases}
\end{equation}
Then
\begin{equation}
    \delta_1\beta\approx \frac{2k_1+k_{3,1}-2k_2}{2k_2-k_{3,1}}\Delta\alpha.
\end{equation}
Next we notice that \eqref{eq:boundary conditions_b2} yields
\begin{equation}
\tan((n-1)\Delta\theta)=\frac{(k_1+k_{3,2}-\omega^2)\sin(\delta_1\beta-\Delta\alpha)-k_1\sin(\Delta\alpha+\delta_1\beta)}{k_1\cos(\Delta\alpha+\delta_1\beta)-(k_1+k_{3,2}-\omega^2)\cos(\delta\beta_1-\Delta\alpha)}
\label{eq:n-1_theta}
\end{equation}
where the denominator approaches zero if $k_{3,2}\approx 2k_2$. 
\begin{itemize}
    \item Subcase: $k_{3,2}\not\approx 2k_2$\\
    Now \eqref{eq:n-1_theta} becomes
    \begin{equation}
    \begin{aligned}
    \tan((n-1)\Delta\theta)
    &\approx\frac{(k_{3,2}-2k_2)\delta_1\beta-(2k_1+k_{3,2}-2k_2)\Delta\alpha}{2k_2-k_{3,2}}\\
    &\approx (\frac{k_{3,2}+2k_1-2k_2}{k_{3,1}-2k_2}+\frac{k_{3,1}+2k_1-2k_2}{k_{3,2}-2k_2})\Delta\alpha
    \end{aligned}
    \label{n-1_theta}
    \end{equation}
    hence
    \begin{equation}
    (n-1)\Delta\theta \approx m\pi-(\frac{2k_1+k_{3,1}-2k_2}{2k_2-k_{3,1}}+\frac{2k_1+k_{3,2}-2k_2}{2k_2-k_{3,2}})\Delta\alpha,~m\in\mathbb{Z}.
    \end{equation}
    Comparing the form of $\Delta\theta$ with the eigenfrequencies in the spectal band in order, it can be further inferred that
    \begin{equation}
    \label{eq:Delta theta}
    \Delta\theta^{(k)}=\frac{k\pi}{n-1}+\frac{\tilde{\theta}^{(k)}}{n-1} \approx \frac{k\pi}{n-1},~k=1,2,3,\cdots
    \end{equation}
near the lower edge of the optical band ($k\ll n$). Moreover, in the generic case where $(\frac{2k_1+k_{3,1}-2k_2}{2k_2-k_{3,1}}+\frac{2k_1+k_{3,2}-2k_2}{2k_2-k_{3,2}})\sim O(1)$, expanding \eqref{n-1_theta} at higher orders enables us to obtain
\begin{equation}
\begin{split}
\label{formula:tilde theta k}
\tilde{\theta}^{(k)}
=&(\frac{2k_1+k_{3,1}-2k_2}{2k_2-k_{3,1}}+\frac{2k_1+k_{3,2}-2k_2}{2k_2-k_{3,2}})\frac{k_2}{2(k_2-k_1)}\frac{k\pi}{n-1}\\
&+[(\frac{2k_1+k_{3,1}-2k_2}{2k_2-k_{3,1}}+\frac{2k_1+k_{3,2}-2k_2}{2k_2-k_{3,2}})\frac{k_2}{2(k_2-k_1)}]^2\frac{k\pi}{(n-1)^2}+\mathcal{O}(|\frac{k}{n}|^3)
\end{split}
\end{equation}
and particularly
\begin{equation}
|\tilde{\theta}^{(k)}-k\tilde{\theta}^{(1)}|\leq \mathcal{O}(|\frac{k}{n}|^3), \quad k\ll n.
\label{theta_k}
\end{equation}

\item Subcase: $k_{3,2}\approx 2k_2$\\
Suppose $k_{3,2}=2k_2+\delta k_{3,2}$ and $|\delta k_{3,2}|\ll 1$, then \eqref{eq:n-1_theta} reads
\begin{equation}
\tan((n-1)\Delta\theta)\approx\frac{2 k_1\Delta\alpha }{ \delta k_{3,2}+2k_1\Delta\alpha\delta_1\beta-\frac{k_1 k_2 (\Delta\theta)^2}{2(k_2-k_1)} }
\end{equation}
and
\begin{equation}
(n-1)\Delta\theta\approx m\pi+{\rm arctan}(\frac{2 k_1\Delta\alpha }{ \delta k_{3,2}+2k_1\Delta\alpha\delta_1\beta-\frac{k_1 k_2 (\Delta\theta)^2}{2(k_2-k_1)} }),~m\in\mathbb{Z}.
\end{equation}
This implies near the lower edge of optical band $|\Delta\theta^{(k)}|\leq O(\frac{k}{n})$ for $k\ll n$. To be more specific, the form of $\Delta\theta^{(k)}$ can be approximated in the following different situations:
\begin{itemize}
    \item $\frac{1}{n}\ll |\delta k_{3,2}|\ll 1$: \\
    Then $\frac{2 k_1\Delta\alpha }{ \delta k_{3,2}+2k_1\Delta\alpha\delta_1\beta-\frac{k_1 k_2 (\Delta\theta)^2}{2(k_2-k_1)} }\approx\frac{2k_1\Delta\alpha}{\delta k_{3,2}}\ll 1$ thus $\Delta\theta^{(k)}\approx \frac{k\pi}{n-1}$.
    \item $|\delta k_{3,2}|\ll \frac{1}{n}$: \\
    Then $\frac{2 k_1\Delta\alpha }{ \delta k_{3,2}+2k_1\Delta\alpha\delta_1\beta-\frac{k_1 k_2 (\Delta\theta)^2}{2(k_2-k_1)} }\gg 1$ thus $\Delta\theta^{(k)}\approx \frac{(k-\frac{1}{2})\pi}{n-1}$.
    \item $|\delta k_{3,2}|\sim \frac{1}{n}$: \\
    Then $\frac{2 k_1\Delta\alpha }{ \delta k_{3,2}+2k_1\Delta\alpha\delta_1\beta-\frac{k_1 k_2 (\Delta\theta)^2}{2(k_2-k_1)} }\approx\frac{2k_1\Delta\alpha}{\delta k_{3,2}}\sim O(1)$ thus $\Delta\theta^{(k)}\approx \frac{(k-1)\pi+\gamma}{n-1}$ where $\gamma\approx{\rm arctan}(\frac{2 k_1\Delta\alpha }{ \delta k_{3,2} })$ is away from $0$ and $\frac{\pi}{2}$.
\end{itemize}

\end{itemize}
\item Case: $k_{3,1}\approx 2k_2$\\
Since $k_{3,1}\approx 2k_2\not\approx k_{3,2}$ and $k_{3,1}\not\approx 2k_2\approx k_{3,2}$ are basically the same, it only remains to discuss $k_{3,1}\approx 2k_2\approx k_{3,2}$. Since this is a special case and the corresponding estimate of $\Delta\theta$ heavily depends on the choice of $\{ k_{3,1},k_{3,2}\}$, here we only list a few examples instead of demonstrating the discussion on all situations.  Suppose $k_{3,1}=2k_2+\delta k_{3,1}$ and $k_{3,2}=2k_2+\delta k_{3,2}$, then \eqref{eq:tan_beta} now reads $\tan\beta\approx\frac{2k_1}{\frac{k_1 k_2(\Delta\theta)^2}{2(k_2-k_1)}-\delta k_{3,1}}\tan\Delta\alpha$ hence $|\beta|\gg |\Delta\alpha|\sim O(|\Delta\theta|)$.
\begin{itemize}
    \item Example (I): $|\delta k_{3,1}|\ll \frac{1}{n^2} \gg |\delta k_{3,2}|$\\
    If we write $\beta=\frac{\pi}{2}+\delta_2\beta$, then $\delta_2\beta\approx\frac{\Delta\theta}{2}$ or $|\Delta\theta|\ll\frac{1}{n}$. For the former \eqref{eq:n-1_theta} becomes 
    \begin{equation}
    \tan((n-1)\Delta\theta)\approx\frac{ -\frac{k_1 k_2 (\Delta\theta)^2}{k_2-k_1} }{ 2k_1\Delta\alpha  }\approx 0
    \end{equation}
    thus $\Delta\theta^{(k)}\approx \frac{k\pi}{n-1}$ or $\Delta\theta^{(k)}\approx \frac{(k-1)\pi}{n-1}$ for $1\leq k\ll n$.
    \item Example (II): $|\delta k_{3,1}|\gg \frac{1}{n} \ll |\delta k_{3,2}|$\\
    This means $|\Delta\alpha|\ll |\delta_1\beta|\ll 1$. Then \eqref{eq:n-1_theta} becomes 
    \begin{equation}
    \tan((n-1)\Delta\theta)\approx\frac{ 2k_1\Delta\alpha-\delta k_{3,2}\delta_1\beta }{ \delta k_{3,2}  }\approx 0
    \end{equation}
    thus $\Delta\theta^{(k)}\approx \frac{k\pi}{n-1}$ or $\Delta\theta^{(k)}\approx \frac{(k-1)\pi}{n-1}$ for $1\leq k\ll n$.
    \item Example (III): $|\delta k_{3,1}|\gg \frac{1}{n}$ and $|\delta k_{3,2}|\ll\frac{1}{n^2}$\\
    This is similar to Example (II) but here
    \begin{equation}
    \tan((n-1)\Delta\theta)\approx\frac{ 2k_1\Delta\alpha }{ \delta k_{3,2}+2k_1\Delta\alpha\delta_1\beta }
    \end{equation}
    which leads to $\Delta\theta^{(k)}\approx \frac{(k-\frac{1}{2})\pi}{n-1}$ for $1\leq k\ll n$.
\end{itemize}

\end{itemize}

\begin{figure}[!htp]
\begin{minipage}{0.5\linewidth}
\leftline{(a)}
\centerline{\includegraphics[height = 5cm, width = 6cm]{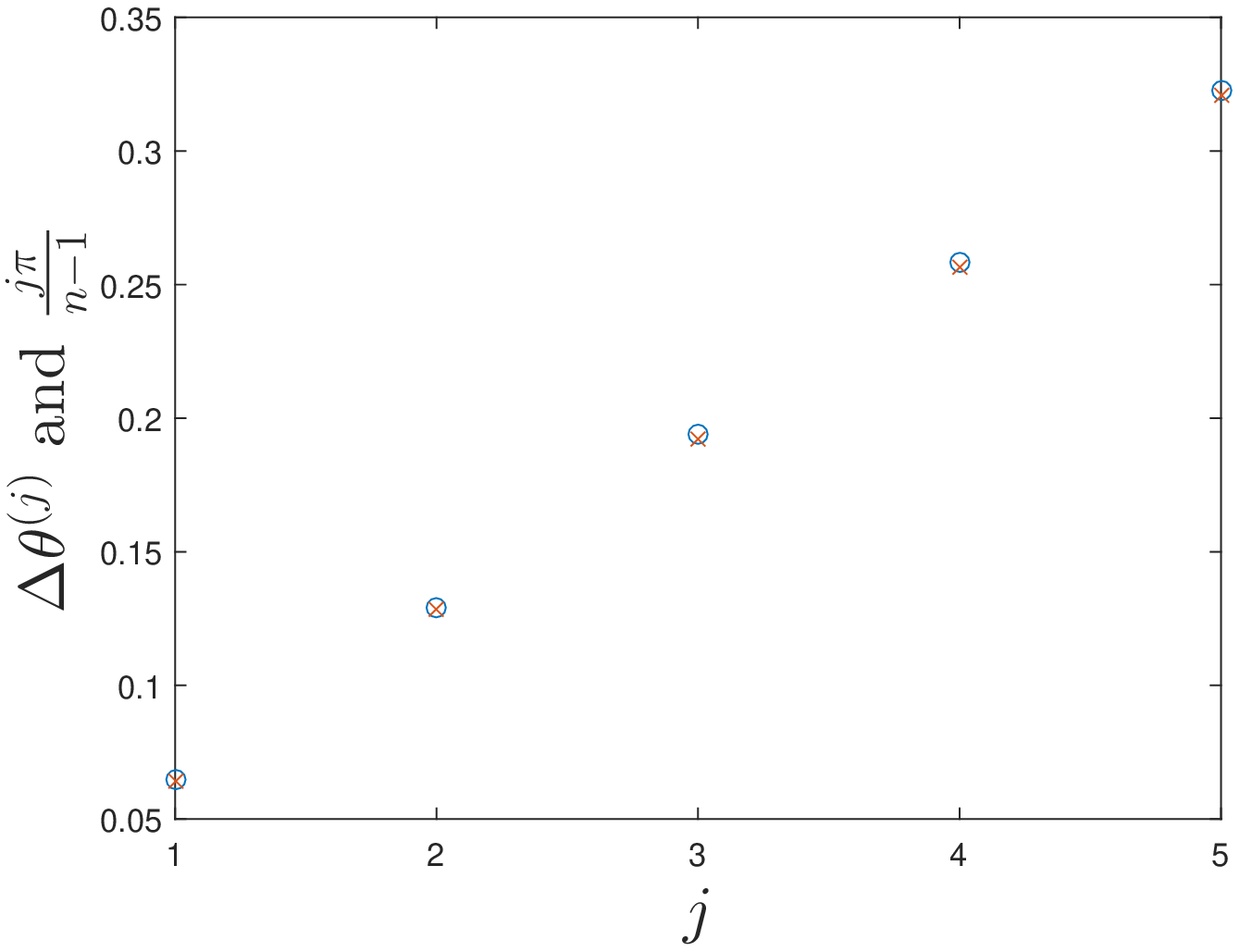}}
\end{minipage}
\hfill
\begin{minipage}{0.5\linewidth}
\leftline{(b)}
\centerline{\includegraphics[height = 5cm, width = 6cm]{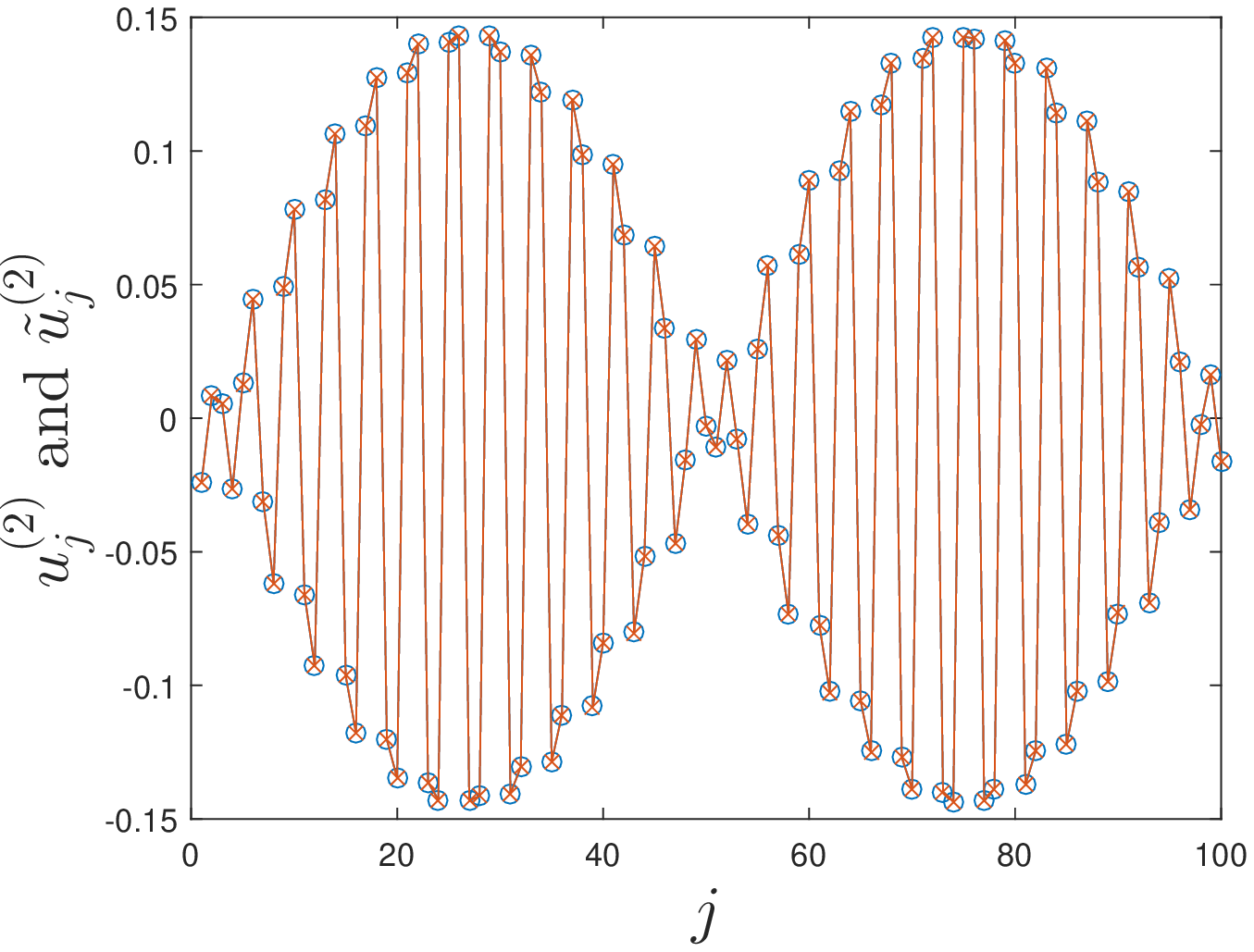}}
\end{minipage}

\begin{minipage}{0.5\linewidth}
\leftline{(c)}
\centerline{\includegraphics[height = 5cm, width = 6cm]{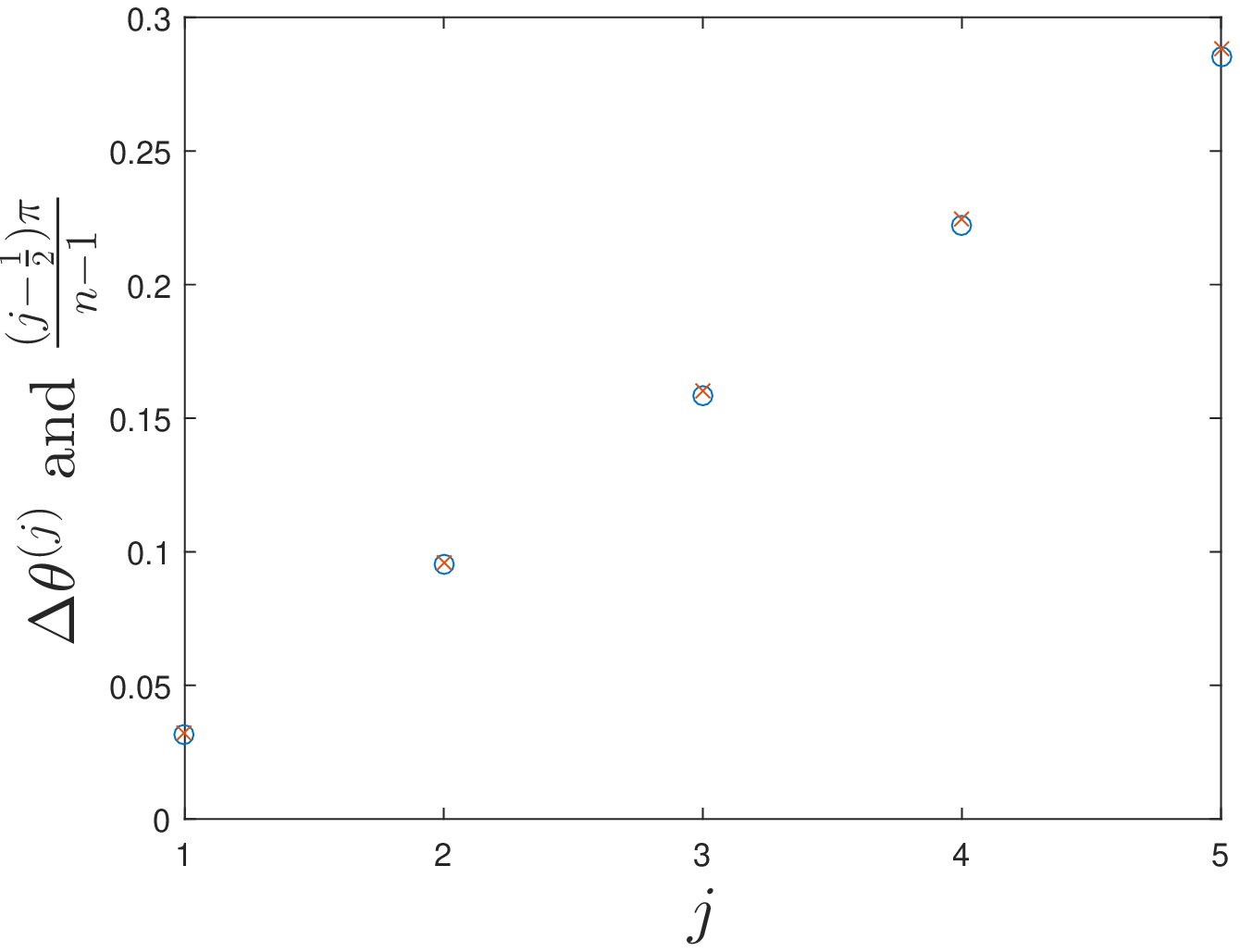}}
\end{minipage}
\hfill
\begin{minipage}{0.5\linewidth}
\leftline{(d)}
\centerline{\includegraphics[height = 5cm, width = 6cm]{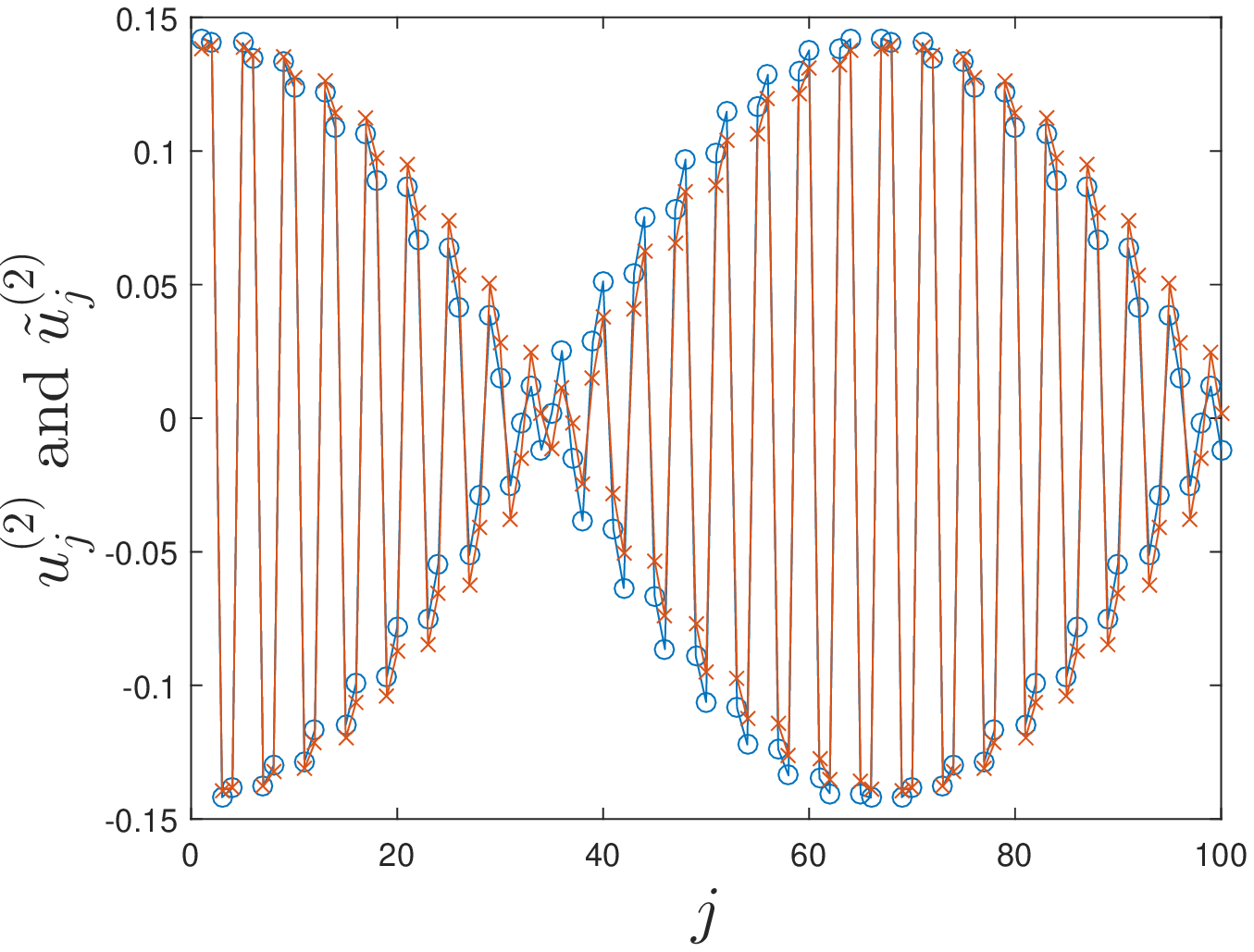}}
\end{minipage}
\caption{
Here we show the approximations for $\Delta\theta^{(j)}$ and the eigenstates near the lower edge of optical band ($\omega^2\approx 2k_2$) in chains with $n=50$, $k_1=0.6$,
$k_{2}=1.4$ and $k_{3,2}=1.6$. In panel (a), we plot numerically obtained $\Delta\theta^{(j)}$ (denoted by blue``$\circ$") for $k_{3,1}=2\not\approx 2k_2$ and compare them with the approximations $\frac{j\pi}{n-1}$ (denoted by red``$\times$''). In panel (b), the corresponding eigenstate $u^{(2)}$ (denoted by blue $``\circ"$) and its approximation $\tilde{u}^{(2)}$ (denoted by red $``\times"$) are illustrated respectively. The panels (c) and (d) in the bottom row follow the same structure as the top row. The only difference is that now we choose $k_{3,1}=2k_2$ so that $\Delta\theta^{(j)}$ in panel (c) are compared with $\frac{(j-\frac{1}{2})\pi}{n-1}$.
}
\end{figure}

In the same spirit, the eigenfrequencies near the upper edge of the optical band can be obtained:
\begin{remark}
\label{remark:op upper}
Suppose $1\leq j\ll n$ and $\omega^2\approx 2k_1+2k_2$.    
\begin{itemize}
    \item When $k_{3,1}\not\approx 2k_2 \not\approx k_{3,2}$, 
    \begin{equation}
    \label{k31 k32 away from 2k2}
        \theta^{(n_1+1-j)}\approx 2\pi-\frac{j}{n-1}\pi.
    \end{equation}
    \item When $k_{3,1}\not\approx 2k_2$ and $|k_{3,2}-2k_2|\ll \frac{1}{n}$, 
    \begin{equation}
    \label{upper edge of optical band}
        \theta^{(n_1+1-j)}\approx 2\pi-\frac{j-\frac{1}{2}}{n-1}\pi.
    \end{equation}
\end{itemize}
\end{remark}

\subsection{Eigenfrequencies in the acoustic band}
\label{sec:acoustic band}
For simplicity of the numbering, we assume the chain bears two edgestates with their eigenfrequencies in the bandgap (which roughly corresponds to the condition $k_{3,1}, k_{3,2}\in (0, 2k_2)$). If we denote the frequencies in the acoustic band by
\begin{equation}
0<(\omega^{(2n)})^2<(\omega^{(2n-1)})^2<\cdots<(\omega^{(2n-n_2+1)})^2<2k_1
\end{equation}
with the form
\begin{equation}
\omega^2=k_1+k_2-\sqrt{k_1^2+k_2^2+2k_1 k_2 \cos(\theta)},
\end{equation}
then 
\begin{equation}
\pi<\theta^{(2n-n_2+1)}<\theta^{(2n-n_2+2)}<\cdots<\theta^{(2n)}<2\pi.
\end{equation}
Revisiting the strategies in Sec.~\ref{sec:optical band}, we can also achieve the estimates for eigenfrequencies near the edges of the acoustic band:
\begin{remark}
Suppose $1\leq j\ll n$,     
\begin{itemize}
    \item when $k_{3,1}\not\approx 0 \not\approx k_{3,2}$, 
    \begin{eqnarray}
        \theta^{(2n-n_2+j)} &\approx& \pi+\frac{j}{n-1}\pi, \\
        \theta^{(2n-j+1)} &\approx& 2\pi-\frac{j}{n-1}\pi.
    \end{eqnarray}
    \item When $k_{3,1}\not\approx 0$ and $|k_{3,2}|\ll \frac{1}{n}$, 
  \begin{eqnarray}
 \theta^{(2n-n_2+j)} &\approx& \pi+\frac{j-\frac{1}{2}}{n-1}\pi, \\
        \theta^{(2n-j+1)} &\approx& 2\pi-\frac{j-\frac{1}{2}}{n-1}\pi.
    \end{eqnarray}
\end{itemize}
\end{remark}

\subsection{The forms and approximations of eigenstates}
\label{form of eigenstates}

After knowing the eigenfrequencies $\omega^{(k)}$ (and $\theta^{(k)}$), the corresponding eigenstates $u^{(k)}$ can be explicitly expressed. 
Here we only show the generic case where $k_{3,1}$ and $k_{3,2}$ are away from $0$ and $2k_2$ while other special cases can be similarly derived.

\begin{itemize}
  \item When $(\omega^{(k)})^{2}\in(2k_{2},2k_{1}+2k_{2})$:
   
According to the periodic structure of diatomic chains and \eqref{eq:eigenvector}\eqref{eq:v1v2}\eqref{eq:c1c2}, the eigenvector $u^{(k)}$ of $\mathcal{L}$ for eigenvalue $-(\omega^{(k)})^2$ has the form
\begin{equation}
\label{matrix:uk op}
  \left(
  \begin{array}{c}
    u^{(k)}_{2j-1} \\
    u^{(k)}_{2j}
  \end{array}
  \right)
  =\left(
  \begin{array}{c}
    \cos(\alpha^{(k)}+\beta^{(k)}+(j-1)\theta^{(k)})\\
    -\cos(-\alpha^{(k)}+\beta^{(k)}+(j-1)\theta^{(k)})
  \end{array}
  \right), \quad 1\leq j\leq n_1
\end{equation}
which near the lower edge of the optical band ($k\ll n_1$) can be written as
\begin{equation}
\label{matrix:uk op lower}
\begin{aligned}
  \left(
  \begin{array}{c}
    u^{(k)}_{2j-1} \\
    u^{(k)}_{2j}
  \end{array}
  \right)
  &=(-1)^{j}\left(
  \begin{array}{c}
    \sin(\Delta\alpha^{(k)}+\Delta\beta^{(k)}+(j-1)\Delta\theta^{(k)})\\
    \sin(-\Delta\alpha^{(k)}+\Delta\beta^{(k)}+(j-1)\Delta\theta^{(k)})
  \end{array}
  \right), \\
\theta^{(k)}&=\pi+\Delta\theta^{(k)}, \quad \alpha^{(k)}=\frac{\pi}{2}+\Delta\alpha^{(k)}, \quad \beta^{(k)}=\Delta\beta^{(k)} 
\end{aligned}
\end{equation}
and near the upper edge ($n_1-k\ll n_1$) yields
\begin{equation}
\label{matrix:uk op upper}
\begin{aligned}
  \left(
  \begin{array}{c}
    u^{(k)}_{2j-1} \\
    u^{(k)}_{2j}
  \end{array}
  \right)
  &=\left(
  \begin{array}{c}
    -\sin(\Delta\alpha^{(k)}+\Delta\beta^{(k)}+(j-1)\Delta\theta^{(k)})\\
    \sin(-\Delta\alpha^{(k)}+\Delta\beta^{(k)}+(j-1)\Delta\theta^{(k)})
  \end{array}
  \right), \\
\theta^{(k)}&=2\pi+\Delta\theta^{(k)}, \quad \alpha^{(k)}=\Delta\alpha^{(k)}, \quad \beta^{(k)}=\frac{\pi}{2}+\Delta\beta^{(k)} .
\end{aligned}
\end{equation}
Since $\Delta\theta^{(k)}\approx k\Delta\theta^{(1)}$ when $1\leq k\ll n$, we can approximate $u^{(k)}$ by $\tilde{u}^{(k)}$ where 
\begin{equation}
  \left(
  \begin{array}{c}
    \tilde{u}^{(k)}_{2j-1} \\
    \tilde{u}^{(k)}_{2j}
  \end{array}
  \right)
  =(-1)^{j}\left(
  \begin{array}{c}
    \sin(k(\Delta\alpha^{(1)}+\Delta\beta^{(1)}+(j-1)\Delta\theta^{(1)}))\\
    \sin(k(-\Delta\alpha^{(1)}+\Delta\beta^{(1)}+(j-1)\Delta\theta^{(1)}))
  \end{array}
  \right), \quad 1\leq k\ll n_1.
 \label{eq:approx eigenvector}
\end{equation}
In this situation we recall $|\Delta\alpha^{(k)}-k\Delta\alpha^{(1)}|\leq O(|\Delta\theta^{(k)}|^3)$, $|\Delta\beta^{(k)}-k\Delta\beta^{(1)}|\leq O(|\Delta\theta^{(k)}|^3)$ and $|\Delta\theta^{(k)}-k\Delta\theta^{(1)}|\leq \mathcal{O}(\frac{k^3}{n^4})$, then obtain that 
\begin{equation}
\label{eq:approx eigenvector err}
|\tilde{u}^{(k)}_j-u^{(k)}_j|\leq \mathcal{O}(\frac{k^3}{n^3}).
\end{equation}

\item When $(\omega^{(k)})^{2}\in(0,2k_{1})$:

Similarly the eigenvector $u^{(k)}$ ($2n-n_2+1\leq k\leq 2n$) has the form
\begin{equation}
\label{matrix:uk ac}
  \left(
  \begin{array}{c}
    u^{(k)}_{2j-1} \\
    u^{(k)}_{2j}
  \end{array}
  \right)
  =\left(
  \begin{array}{c}
    \cos(\alpha^{(k)}+\beta^{(k)}+(j-1)\theta^{(k)})\\
    \cos(-\alpha^{(k)}+\beta^{(k)}+(j-1)\theta^{(k)})
  \end{array}
  \right), \quad 1\leq j\leq n
\end{equation}
which near the upper edge of the acoustic band ($k-(2n-n_2)\ll n_2$) becomes
\begin{equation}
\label{matrix:uk ac upper}
\begin{aligned}
\left(
  \begin{array}{c}
    u^{(k)}_{2j-1} \\
    u^{(k)}_{2j}
  \end{array}
  \right)
  &=(-1)^{j}\left(
  \begin{array}{c}
    \sin(\Delta\alpha^{(k)}+\Delta\beta^{(k)}+(j-1)\Delta\theta^{(k)})\\
    -\sin(-\Delta\alpha^{(k)}+\Delta\beta^{(k)}+(j-1)\Delta\theta^{(k)})
  \end{array}
  \right), \\
\theta^{(k)}&=\pi+\Delta\theta^{(k)}, \quad \alpha^{(k)}=\frac{\pi}{2}+\Delta\alpha^{(k)}, \quad \beta^{(k)}=\Delta\beta^{(k)}  
\end{aligned}
\end{equation}
and near the lower edge ($2n-k\ll n_2$) reads
\begin{equation}
\label{matrix:uk ac lower}
\begin{aligned}
\left(
  \begin{array}{c}
    u^{(k)}_{2j-1} \\
    u^{(k)}_{2j}
  \end{array}
  \right)
  &=\left(
  \begin{array}{c}
    -\sin(\Delta\alpha^{(k)}+\Delta\beta^{(k)}+(j-1)\Delta\theta^{(k)})\\
    -\sin(-\Delta\alpha^{(k)}+\Delta\beta^{(k)}+(j-1)\Delta\theta^{(k)})
  \end{array}
  \right), \\
\theta^{(k)}&=2\pi+\Delta\theta^{(k)}, \quad \alpha^{(k)}=\Delta\alpha^{(k)}, \quad \beta^{(k)}=\frac{\pi}{2}+\Delta\beta^{(k)}.  
\end{aligned}
\end{equation}

\end{itemize}

\subsection{Some properties of the eigenstates}

\begin{itemize}
\item Norms of the eigenstates

For the generic case  $k_{3,1}\not\approx 2k_{2}\not\approx k_{3,2}$, we have the following results
\begin{lemma}
\label{norm of eigenstates}
When $(\omega^{(k)})^{2}\in(0,2k_{1})\cup(2k_{2},2k_{1}+2k_{2})$, $\forall \varepsilon>0$,
\begin{equation}
  ||u^{(k)}|^{2}-n|\lesssim n^{\varepsilon} ;
\label{norm of uk in spectral band}
\end{equation}
\begin{proof}
See Appendix \ref{sec:norm lemma}.
\end{proof}
\end{lemma}

\item Estimates on the inner products involving the eigenstates

\begin{definition}
\label{definition of Hadamard product}
Suppose $\vec{x}=(x_{1},x_{2},\cdots,x_{n})^{\top}$ and $\vec{y}=(y_{1},y_{2},\cdots,y_{n})^{\top}$ are two vectors. Let
"$\circ$" denotes the Hadamard product such that
\begin{equation*}
  \vec{x}\circ\vec{y}=(x_{1}y_{1},x_{2}y_{2},\cdots,x_{n}y_{n})^{\top}
\end{equation*}
In what follows, we use $\vec{x}\vec{y}$ to represent $\vec{x}\circ\vec{y}$ when not causing confusion.
\end{definition}

Here the assumption on $k_{3,1}$ and $k_{3,2}$ is the same as that in 
 Sec.~\ref{form of eigenstates}. Thus,
\begin{lemma}
\label{lemma: inner product}
(1). If $(\omega^{(a)})^{2},(\omega^{(b)})^{2},
(\omega^{(c)})^{2},(\omega^{(d)})^{2}\in (0,2k_1)\cup(2k_{2},2k_{1}+2k_{2})$, then 
\begin{equation}
    |(u^{(a)}u^{(b)}u^{(c)},u^{(d)})|\leq 2n.
\end{equation}

(2). If at least one of $(\omega^{(a)})^{2},(\omega^{(b)})^{2},
(\omega^{(c)})^{2},(\omega^{(d)})^{2}$ is not in $[0,2k_1]\cup[2k_{2},2k_{1}+2k_{2}]$, then 
\begin{equation}
    |(u^{(a)}u^{(b)}u^{(d)},u^{(d)})|\leq \mathcal{O}(1).
\end{equation}

(3). Suppose $(\omega^{(a)})^{2},(\omega^{(b)})^{2},
(\omega^{(c)})^{2},(\omega^{(d)})^{2}\in(2k_{2},2k_{1}+2k_{2})$. If $k=\max\{a,b,c,d\}$ and $1\leq k\ll n$, then for $k_{3,1}\neq 2k_{2}\neq k_{3,2}$
\begin{itemize}
    \item $\prod\limits_{j_1=0}^1\prod\limits_{j_2=0}^1\prod\limits_{j_3=0}^1(a+(-1)^{j_1}b+(-1)^{j_2}c+(-1)^{j_3}d)\neq 0$ implies
    \begin{equation}
     |(u^{(a)}u^{(b)}u^{(d)},u^{(d)})|\sim \mathcal{O}(1);
    \end{equation}
    \item $\prod\limits_{j_1=0}^1\prod\limits_{j_2=0}^1\prod\limits_{j_3=0}^1(a+(-1)^{j_1}b+(-1)^{j_2}c+(-1)^{j_3}d)= 0$ implies
    \begin{equation}
    |(u^{(a)}u^{(b)}u^{(d)},u^{(d)})|\sim \mathcal{O}(n)
    \end{equation}
\end{itemize}  
(4). When $(\omega^{(a)})^{2},(\omega^{(b)})^{2},(\omega^{(c)})^{2},(\omega^{(d)})^{2}\in(2k_{2},2k_{1}+2k_{2})$ and $\max\{a,b,c\}\ll n^{1-\varepsilon}< d< n_{1}-a-b-c$,
\begin{equation}
  (u^{(a)}u^{(b)}u^{(c)},u^{(d)})\lesssim\ n^{2\varepsilon}
   \label{eq:inner product with a far frequency}
\end{equation}    
 \begin{proof}
 See \textbf{Appendix} \ref{sec:inner product}.
 \end{proof}   

\end{lemma}
\end{itemize}

\section{Emergence of a nonlinear localized state with frequency near the lower edge of optical band}
\label{add nonlinearility}
When the interactions between the masses gradually become nonlinear, linear states (localized and nonlocalized) discussed in previous sections will in general persist, which can be proved by the Lyapunov-Schmidt Reduction and the Implicit Function Theorem. However, the strength of nonlinearity for the Implicit Function Theorem to hold can be very small in long chains. As the nonlinearity grows, frequencies previously in the band may cross the band edge and new nonlinear localized states 
can emerge. Here we add a cubic term to the linear system \eqref{eq:motion1} as an example to explore this formation mechanism of localized solutions:
\begin{equation}
  \frac{d^{2}}{dt^{2}}q(t)=\mathcal{L}q(t)+(q(t))^{3}
  \label{eq:motion nl0}
\end{equation}
with energy
\begin{equation}
\label{eq:energy}
\begin{split}
E=&\sum_{j=1}^{2n}[\frac{1}{2}(\dot{q}_j)^2-\frac{1}{4}q_j^4] +\sum_{j=1}^{n}\frac{k_1}{2}(q_{2j-1}-q_{2j})^2+\sum_{j=1}^{n-1}\frac{k_2}{2}(q_{2j+1}-q_{2j})^2 \\
&+\frac{k_{3,1}}{2}q_{1}^2+\frac{k_{3,2}}{2}q_{2n}^2
\end{split}
\end{equation}
where $q^{3}$ represents the Hadamard product $q\circ q\circ q$.
Focusing on time-periodic solutions, we 
define $\tau=\omega t$ and $Q(\tau)=q(\frac{\tau}{\omega})=q(t)$.
Then the nonlinear system \eqref{eq:motion nl0} becomes
\begin{equation}
  \omega^{2}\frac{d^{2}}{d\tau^{2}}Q(\tau)=\mathcal{L}Q(\tau)+(Q(\tau))^{3}
  \label{eq:motion nl}
\end{equation}
By the Implicit Function Theorem, linear periodic solutions can possibly be continued to families of nonlinear periodic solutions in $L^{2}_{per}[0,2\pi]$ parametrized by the frequency or the amplitude. In particular, we study the frequency near the lower edge of optical band and write the expansions as follows:
\begin{equation}
\begin{split}
 \omega  &= \omega^{<0>}+\epsilon\omega^{<1>}+\epsilon^{2}\omega^{<2>}+\cdots \\
 Q(\tau) &= \epsilon Q^{<0>}(\tau)+\epsilon^{2}Q^{<1>}(\tau)+\epsilon^{3}Q^{<2>}(\tau)+\cdots
 \end{split}
 \label{eq:series1}
\end{equation}
where $\omega^{< 0 >}=\omega^{(1)}$ and  $Q^{< 0 >}(\tau)=\frac{u^{(1)}}{|u^{(1)}|}(e^{i\tau}+e^{-i\tau})$.

For solutions even in time, we write $Q^{<m>}$ as 
\begin{equation}
  Q^{<m>}(\tau)=\sum_{j=1}^{2n}
  \sum_{k=1}^{+\infty}c_{m,k,j}
  \frac{u^{(j)}}
  {|u^{(j)}|}
  (e^{ik\tau}+e^{-ik\tau})
\label{eq:The form of Q}
\end{equation}
where $c_{0,1,1}=1$ and $c_{m,1,1}=0$ for $m\geq 1$. As a result, we have  $(Q^{<0>}(\tau),Q^{<k>}(\tau))=0$ for $k\geq 1$ where the $L^2$ inner product is given by
$$
(\vec{f}(\tau),\vec{g}(\tau))=\frac{1}{2\pi}\int^{2\pi}_{0}\sum_{j=1}^{2n}f_{j}(\tau)g_{j}(\tau)d\tau.
$$

By exploiting the symmetry of $\mathcal{L}$ and the cubic form of nonlinearity, we simplify \eqref{eq:series1} as in the following lemma.
\begin{lemma}
\label{prop:simple form}
For $m\geq 0$, 
\begin{equation*}
  \omega^{< 2m+1 >}=0,\quad Q^{< 2m+1 >}(\tau)=0.
\end{equation*}
\end{lemma}
\begin{proof}
We prove by induction:

At $\mathcal{O}(\epsilon^{2})$ of \eqref{eq:motion nl}, we have
\begin{equation*}
  ((\omega^{<0>})^{2}\frac{d^{2}}{d\tau^{2}}-\mathcal{L})Q^{<1>}(\tau)+2\omega^{<0>}\omega^{<1>}\frac{d^{2}}{d\tau^{2}}Q^{<0>}(\tau)=0
\end{equation*}
Projecting the equation above onto $span\{u^{(1)}(e^{i\tau}+e^{-i\tau})\}$, then we get $\omega^{<1>}=0$. Moreover, the vanishing of right hand side and $(Q^{<1>},Q^{<0>})=0$ then imply $Q^{<1>}(\tau)=0$.

Suppose $\omega^{<2l+1>}=0$ and $Q^{<2l+1>}(\tau)=0$ for $l=0,1,2,\cdots,m$.
When $l=m+1$, at order $\mathcal{O}(\epsilon^{2m+4})$ of equation \eqref{eq:motion nl}, we have
\begin{equation*}
 ((\omega^{<0>})^{2}\frac{d^{2}}{d\tau^{2}}-\mathcal{L})Q^{< 2m+3>}(\tau)+2\omega^{<0>}\omega^{<2m+3>}\frac{d^{2}}{d\tau^{2}}Q^{<0>}(\tau)
 =0.
\end{equation*}
The right hand side vanishes because every resulting term includes at least one of $Q^{<2l+1>},(1\leq l\leq m)$.
Similar to the above, we can get $\omega^{<2m+3>}=0$ and $Q^{<2m+3>}(\tau)=0$.
Now according to the second mathematical inductive method, we have proved $\omega^{< 2m+1 >}=0$ and $Q^{< 2m+1 >}(\tau)=0$ for $m\geq 0$.
\end{proof}

Now $\epsilon^3 Q^{<2>}$ is the first correction term to $\epsilon Q^{<0>}$ and $Q^{<2>}$ can be calculated as follows:

\begin{proposition}
\label{prop:Q2}
\begin{equation}
\label{eq:Q2}
Q^{<2>}(\tau)=\sum\limits_{j=2}^{2n}c_{2,1,j}\frac{u^{(j)}}
  {|u^{(j)}|}
  (e^{i\tau}+e^{-i\tau})+\sum\limits_{j=1}^{2n}c_{2,3,j}\frac{u^{(j)}}
  {|u^{(j)}|}
  (e^{3i\tau}+e^{-3i\tau})
\end{equation}
where
\begin{equation}
\begin{split}
  &|c_{2,1,3}-\frac{-3(k_{2}-k_{1})n}{16k_{1}k_{2}\pi^{2}}|\leq\mathcal{O}(1);
  |c_{2,1,j}|\leq \mathcal{O}(\frac{1}{j^{2}}), (j\in[2,n^{1-\varepsilon}]\setminus j=3); \\ &|c_{2,1,j}|\leq\mathcal{O}(\frac{1}{n^{2-4\varepsilon}}),
  (j\in(n^{1-\varepsilon},n_{1}-4]);\\
  &|c_{2,1,j}|\leq\mathcal{O}(\frac{1}{n}),((\omega^{(j)})^{2}\in(0,2k_{1})\cup(2k_2,2k_{1}+2k_{2}), j \geq n_{1}-3);\\
  &|c_{2,1,j}|\leq\mathcal{O}(\frac{1}{n^{\frac{3}{2}}}),((\omega^{(j)})^{2}\in(2k_{1},2k_{2})\cup(2k_{1}+2k_{2},+\infty))
\end{split}
\label{eq:c(21j)}
\end{equation}
and  
\begin{equation}
\begin{split}
  &|c_{2,3,1}-\frac{-3}{64k_{2}n}|\leq\mathcal{O}(\frac{1}{n^2});\quad
   |c_{2,3,3}-\frac{1}{64k_{2}n}|\leq\mathcal{O}(\frac{1}{n^{2}});\\
  &|c_{2,3,j}|\leq\mathcal{O}(\frac{1}{n}),((\omega^{(j)})^{2}\in(0,2k_{1})\cup(2k_2,2k_{1}+2k_{2}));\\
  &|c_{2,3,j}|\leq\mathcal{O}(\frac{1}{n^{3/2}}),((\omega^{(j)})^{2}\in(2k_{1},2k_{2})\cup(2k_{1}+2k_{2},+\infty)).
\end{split}
\label{eq:c(23j)}
\end{equation}
\end{proposition}
   
\begin{proof}
At $\mathcal{O}(\epsilon^{3})$ of the equation \eqref{eq:motion nl}, we get
\begin{equation}
\begin{split}
 &\quad((\omega^{<0>})^{2}\frac{d^{2}}{d\tau^{2}}-\mathcal{L})Q^{<2>}(\tau)+2\omega^{<0>}\omega^{<2>}\frac{d^{2}}{d\tau^{2}}Q^{<0>}(\tau)\\
  &=\frac{3(u^{(1)})^{3}}{|u^{(1)}|^{3}}(e^{i\tau}+e^{-i\tau})+\frac{(u^{(1)})^{3}}{|u^{(1)}|^{3}}(e^{3i\tau}+e^{-3i\tau})
 \end{split}
  \label{eq:Q2 derivation}
\end{equation}
If we project \eqref{eq:Q2 derivation} onto ${\rm span}\{u^{(j)}(e^{i\tau}+e^{-i\tau})\}$ and ${\rm span}\{u^{(j)}(e^{3i\tau}+e^{-3i\tau})\}$ respectively,
then we obtain
\begin{equation}
\begin{split}
  c_{2,1,j}&=\frac{1}{|u^{(1)}|^{3}|u^{(j)}|}\frac{3((u^{(1)})^{3},u^{(j)})}{(\omega^{(j)})^{2}-(\omega^{(1)})^{2}}, \\
  c_{2,3,j}&=\frac{1}{|u^{(1)}|^{3}|u^{(j)}|}\frac{((u^{(1)})^{3},u^{(j)})}{(\omega^{(j)})^{2}-3^{2}(\omega^{(1)})^{2}}.
  \end{split}
  \label{eq:Q2(tau)}
\end{equation}
Thus according to \textbf{Lemma} \ref{norm of eigenstates}, \textbf{Lemma} \ref{lemma: inner product}, the following can be calculated:
\begin{itemize}
\item 
$
\frac{(u^{(1)}u^{(1)}u^{(1)},u^{(3)})}{|u^{(1)}|^3 |u^{(3)}|}\approx \frac{-\frac{n}{4}}{n^2}=-\frac{1}{4n}, \quad \frac{(u^{(1)}u^{(1)}u^{(1)},u^{(1)})}{|u^{(1)}|^4}\approx \frac{\frac{3n}{4}}{n^2}=\frac{3}{4n}.
$

\item In general, $|(\omega^{(j)})^2-3^2(\omega^{(1)})^2|>\frac{3^2}{2}(\omega^{(1)})^2\geq \mathcal{O}(1)$. Here we assume that isolated frequencies exist in the band gap rather than above the optical band. 

 \item When $1<j\ll n$, $(\omega^{(j)})^2\approx 2k_2+\frac{k_1 k_2}{2(k_2-k_1)}(\frac{j\pi}{n-1})^2$.

\item When $1<j\ll n$ and $j\neq 3$,
\begin{equation*}
\frac{(u^{(1)}u^{(1)}u^{(1)},u^{(j)})}{|u^{(1)}|^3 |u^{(j)}|}\lesssim \frac{1}{n^2}, \quad (\omega^{(j)})^{2}-(\omega^{(1)})^{2}\approx \frac{k_1 k_2}{2(k_2-k_1)}\frac{(j^2-1)\pi^2}{(n-1)^2}.
\end{equation*}

\item When $\varepsilon>0$ and $n^{1-\varepsilon}<j\leq n_{1}-4$,
\begin{equation*}
  \frac{(u^{(1)}u^{(1)}u^{(1)},u^{(j)})}{|u^{(1)}|^3 |u^{(j)}|}\lesssim \frac{n^{2\varepsilon}}{n^2}, \quad (\omega^{(j)})^{2}-(\omega^{(1)})^{2}> \frac{k_1 k_2}{4(k_2-k_1)}\frac{n^{2-2\varepsilon}\pi^2}{n^2}.
\end{equation*}

\item When $(\omega^{(j)})^{2}\in(0,2k_{1})\cup (2k_2,2k_{1}+2k_{2})$ and $j\geq n_{1}-3$,
\begin{equation*}
    \frac{(u^{(1)}u^{(1)}u^{(1)},u^{(j)})}{|u^{(1)}|^3 |u^{(j)}|}< \frac{(1+\varepsilon_1)2n}{n^2}, \quad(\omega^{(j)})^{2}-(\omega^{(1)})^{2}>\min\{2k_2-2k_1,(1-\varepsilon_2)k_1\}.
\end{equation*}

\item When $(\omega^{(j)})^{2}\in(2k_{1},2k_{2})\cup (2k_{1}+2k_{2},+\infty)$,
\begin{equation*}
 \frac{(u^{(1)}u^{(1)}u^{(1)},u^{(j)})}{|u^{(1)}|^3 |u^{(j)}|}< \frac{\tilde{C}}{n^{3/2}}, \quad (\omega^{(j)})^{2}-(\omega^{(1)})^{2}>\delta_2  
\end{equation*}
\end{itemize}
Then it can be easily checked that \eqref{eq:c(21j)} and \eqref{eq:c(23j)} hold.
\end{proof}
%
%
\begin{remark}
We notice that some results in the proof of Prop.~\ref{prop:Q2} also hold in more general scenarios, such as
\begin{itemize}
    \item If again isolated frequencies can only sit in the band gap, then $|(\omega^{(j)})^2-k^2(\omega^{(1)})^2|>\frac{k^2}{2}(\omega^{(1)})^2\geq \mathcal{O}(k^2)$.
    \item For $m\geq 1$, 
    \begin{equation}
    \label{eq:Q2m}
    Q^{<2m>}(\tau)=\sum\limits_{j=2}^{2n}c_{2m,1,j}\frac{u^{(j)}}{|u^{(j)}|}
    (e^{i\tau}+e^{-i\tau})+\sum\limits_{k=1}^m\sum\limits_{j=1}^{2n}c_{2m,2k+1,j}\frac{u^{(j)}}{|u^{(j)}|}(e^{(2k+1)i\tau}+e^{-(2k+1)i\tau})
\end{equation}
\end{itemize}
\end{remark}
On the other hand, projecting equation \eqref{eq:Q2 derivation} onto space ${\rm span}\{u^{(1)}(e^{i\tau}+e^{-i\tau})\}$ yields
\begin{equation}
 \omega^{<2>}=-\frac{3((u^{(1)})^{3},u^{(1)})}{2|u^{(1)}|^{4}\omega^{<0>}}
  =-\frac{9}{8\sqrt{2k_{2}}n}+\mathcal{O}(\frac{1}{n^{2}})\sim\mathcal{O}(\frac{1}{n}).
  \label{eq:omega2}
\end{equation}
With the knowledge of $\omega^{<2>}$, we now have a linear approximation of $\omega$ as 
\begin{equation}
\begin{split}
\omega &\approx \omega^{<0>}+\epsilon^2\omega^{<2>}\\
& \approx \sqrt{2k_2}+\frac{k_1 k_2}{4\sqrt{2k_{2}}(k_2-k_1)}(\frac{\pi}{n-1})^2- \frac{9}{8 \sqrt{2k_2}n}\epsilon^2 
\label{eq:linear approx}
\end{split}
\end{equation}
where $\epsilon^2$ serves as a variable. This implies that near $\epsilon^2\approx \frac{2k_1 k_2}{9(k_2-k_1)}\frac{\pi^2}{n}\sim\ \mathcal{O}(\frac{1}{n})$ the frequency crosses the lower edge of the optical band and enters the band gap. 
At the same time, as $\epsilon$ grows the state corresponding to this frequency gradually gets localized, producing a family of breathers which do not bear a localized counterpart in the linear limit. However, the expansions in \eqref{eq:series1} (and \eqref{eq:linear approx}) only hold when $\epsilon^2$ is small enough and $|c_{2,1,3}|\sim\mathcal{O}(n)$ prevents $\epsilon^2$ from growing beyond $\mathcal{O}(\frac{1}{n})$. In fact, we will show in the following theorem the radius of convergence for $\epsilon$ is just at $\mathcal{O}(\frac{1}{\sqrt{n}})$.

\begin{theorem}\label{theorem}
Expansions of $Q$ and $\omega$ in \eqref{eq:series1} converge when $0<\epsilon^2<\frac{1}{fn}$.
\end{theorem}
In order to prove the theorem, we first give the following lemma:

\begin{lemma}
\label{lemma: main estimates}
Let 
\begin{equation*}
  \mathbb{D}_{1}=\{j\in\mathbb{Z} | 2\leq j\leq n_1 \};\quad
   \mathbb{D}_{2}=\{j\in\mathbb{Z} | 2n-n_2+1\leq j\leq 2n \};\quad
   \mathbb{D}_{3}=\{n_1+1,2n-n_2\}.
\end{equation*}
and suppose 
$$
\mathcal{R}_m=\frac{C_1 f^m n^m}{m^2+2}
$$
where $f=\frac{40(k_2-k_1)}{k_1 k_2}$ and $C_1>0$. Then there exists $C_2, C_3, C_4>0$ such that 
$$C_1C_3=\frac{1}{16},\quad  \frac{27k_1k_2}{320\sqrt{2k_2}(k_2-k_1)}+\delta<C_1C_2=\frac{7k_1k_2}{80\sqrt{2k_2}(k_2-k_1)}<\frac{\pi^2k_1k_2}{64\sqrt{2k_2}(k_2-k_1)}-\delta
$$ and for $m\geq 1$
\begin{eqnarray*}
&(\rmnum{1}):|\omega^{<2m>}|<\frac{C_2\mathcal{R}_{m}}{n^{2}};\quad
&(\rmnum{2}):\sum\limits_{j\in\mathbb{D}_{1}}|c_{2m,1,j}|<C_3\mathcal{R}_{m}; \\
&(\rmnum{3}):\sum\limits_{j\in\mathbb{D}_{2}}|c_{2m,1,j}|<\frac{C_4\mathcal{R}_{m}}{n};\quad
&(\rmnum{4}):\sum\limits_{j\in\mathbb{D}_{3}}|c_{2m,1,j}|<\frac{C_4\mathcal{R}_{m}}{n^{\frac{5}{2}}};\\
&(\rmnum{5}):\sum\limits_{j=1}^{n_{1}}|c_{2m,2l+1,j}|<\frac{C_4\mathcal{R}_{m}}{n(2l+1)^{2}};\quad
&(\rmnum{6}):\sum\limits_{j\in\mathbb{D}_{2}}|c_{2m,2l+1,j}|<\frac{C_4\mathcal{R}_{m}}
{n(2l+1)^{2}};\\
&(\rmnum{7}):\sum\limits_{j\in\mathbb{D}_{3}}|c_{2m,2l+1,j}|<\frac{C_4\mathcal{R}_{m}}{n^{\frac{5}{2}}(2l+1)^{2}}. \quad &
\end{eqnarray*}
\end{lemma}

\begin{proof}
We prove by induction:
When $m=1$, we have $\mathcal{R}_1=\frac{C_1 f n}{3}$ and
\begin{eqnarray*}
&|\omega^{<2>}|\approx \frac{9}{8\sqrt{2k_2}n}< \frac{C_1C_2f}{3n}=\frac{C_2\mathcal{R}_1}{n^2}; \quad
&\sum\limits_{j\in \mathbb{D}_1}|c_{2,1,j}|\approx \frac{3(k_2-k_1)n}{16 k_1 k_2\pi^2}<\frac{C_1C_3fn}{3}=C_3\mathcal{R}_1;\\
&\sum\limits_{j\in \mathbb{D}_2}|c_{2,1,j}|<\frac{(1+\varepsilon)6}{2k_2-2k_1}<\frac{C_1C_4 f}{3}=\frac{C_4\mathcal{R}_1}{n}; \quad
&\sum\limits_{j\in \mathbb{D}_3}|c_{2,1,j}|<\frac{2\tilde{C}}{n^{3/2}}<\frac{C_1C_4f}{3n^{3/2}}=\frac{C_4\mathcal{R}_1}{n^{5/2}}; \\
&\sum\limits_{j=1}^{n_1}|c_{2,3,j}|<\frac{2}{9k_2}<\frac{C_1 C_4 f}{27}=\frac{C_4\mathcal{R}_1}{3^2 n}; \quad
&\sum\limits_{j\in\mathbb{D}_2}|c_{2,3,j}|<\frac{2}{9k_2}<\frac{C_1 C_4 f}{27}=\frac{C_4\mathcal{R}_1}{3^2 n}; \\
&\sum\limits_{j\in\mathbb{D}_3}|c_{2,3,j}|<\frac{2\tilde{C}}{9k_2 n^{3/2}}<\frac{C_1 C_4 f}{27 n^{3/2}}=\frac{C_4\mathcal{R}_1}{3^2 n^{5/2}} \quad &
\end{eqnarray*}
where the last five inequalities hold when $C_4$ is chosen to be a large (yet $\mathcal{O}(1)$) constant.

Now we assume $(i)$-$(vii)$ hold for $m=1,2,\cdots,k$ and will show that they remain true for $m=k+1$.
Before that we introduce some notations:
\begin{equation}
  \mathcal{S}_{m}=\{(j,l)|j,l\in\mathbb{N}^2,1\leq j \leq 2n, 0\leq l\leq m\}, \quad m\geq 0.
  \label{eq:S_m}
\end{equation}
Here $\mathcal{S}_m$ ($m\geq 1$) consists of four parts $F_{s,m}$ ($0\leq s\leq 3$) where
\begin{equation}
\label{eq:F_s}
\begin{split}
  F_{0,m}&=\{(j,l)\in\mathcal{S}_m | j=1, 1\leq l\leq m\}, \\
  F_{1,m}&=\{(j,l)\in\mathcal{S}_m | j\in\mathbb{D}_1, 0\leq l\leq m\}\cup\{(1,0)\}, \\ 
  F_{s,m}&=\{(j,l)\in\mathcal{S}_m | j\in\mathbb{D}_s, 0\leq l\leq m\}, \quad s=2,3.
\end{split}
\end{equation}
And we define some sets
\begin{equation}
\label{eq:sets}
  \begin{split}
   &\mathbb{E}_{1}(k)=\{(x,y,z)|x,y,z\in \mathbb{N}, x+y+z=k,x<y< z\};\\
   &\mathbb{E}_{1}^0(k)=\{(x,y,z)|x,y,z\in \mathbb{N}, x+y+z=k\};\\
   &\mathbb{E}_{1}^{1}(k)=\{(x,y,z)|x,y,z\in \mathbb{N}, x+y+z=k,x\leq y\leq z\};\\
   &\mathbb{E}_{1}^{1^*}(k)=\{(x,y,z)|x,y,z\in \mathbb{N}^*, x+y+z=k,x\leq y\leq z\};\\
   &\mathbb{E}_{2}(k)=\{(x,y,z)|x,y,z\in \mathbb{N},x+y+z=k+1,x<y\}  ;\\
   &\mathbb{E}_{2}^0(k)=\{(x,y,z)|x,y,z\in \mathbb{N},x+y+z=k+1\}  ;\\
    &\mathbb{E}_{2}^{0^*}(k)=\{(x,y,z)|x,y,z\in \mathbb{N}^*,x+y+z=k+1\}  ;\\
   &\mathbb{E}_{2}^{1}(k)=\{(x,y,z)|x,y,z\in \mathbb{N},x+y+z=k+1,x\leq y\} ;\\
   &\mathbb{E}_{2}^{1^*}(k)=\{(x,y,z)|x,y,z\in \mathbb{N}^{*},x+y+z=k+1,x\leq y\}; \\
   &\mathbb{E}_{2}^{2}(k)=\mathbb{E}_{2}^{0}(k)\backslash \{(0,0,k+1)\};\\
   &\mathbb{E}_{2}^{3}(k)=\mathbb{E}_{2}^{0}(k)\backslash \{(0,0,k+1),(0,k+1,0),(k+1,0,0)\};\\
   &\mathbb{E}_{3}(k)=\{(x,y)|x,y\in \mathbb{N}^{*},x+y=k+1,x<y\};\\
   &\mathbb{E}_{3}^{0}(k)=\{(x,y)|x,y\in \mathbb{N},x+y=k+1\}; \\
   &\mathbb{E}_{3}^{0^*}(k)=\{(x,y)|x,y\in \mathbb{N}^{*},x+y=k+1\};\\
   &\mathbb{E}_{3}^1(k)=\{(x,y)|x,y\in \mathbb{N},x+y=k+1,x\leq y\};\\
   &\mathbb{E}_{3}^{1^*}(k)=\{(x,y)|x,y\in \mathbb{N^{*}},x+y=k+1,x\leq y\}
  \end{split}
\end{equation}

\begin{itemize}
\item $(i)$: $|\omega^{<2k+2>}|$

To estimate $|\omega^{<2k+2>}|$, we consider $\mathcal{O}(\epsilon^{2k+3})$ terms of equation \eqref{eq:motion nl}:
\begin{equation}
  \begin{split}
 &\quad((\omega^{<0>})^{2}\frac{d^{2}}{d\tau^{2}}-\mathcal{L})Q^{<2k+2>}(\tau)
  +\sum_{(r,s,t)\in\mathbb{E}_{2}^2(k)}\omega^{<2r>}\omega^{<2s>}\frac{d^{2}}{d\tau^{2}}Q^{<2t>}(\tau)\\
  &=\sum_{(a_{1},a_{2},a_{3})\in\mathbb{E}_{1}^0(k)} Q^{<2a_{1}>}(\tau)Q^{<2a_{1}>}(\tau)Q^{<2a_{3}>}(\tau).
  \end{split}
  \label{eq:order 2k+3}
\end{equation}
Projecting the equation onto ${\rm span}\{ u^{(1)}(e^{i\tau}+e^{-i\tau}\}$, we get
\begin{equation}
  \begin{split}
  \label{eq:omega 2k+2}
&2|\omega^{<0>}\omega^{<2k+2>}| \\
\leq&\sum_{(a_{1},a_{2},a_{3})\in \mathbb{E}_{1}^{0}(k)}
\sum_{(j_{1},l_{1})\in\mathcal{S}_{a_{1}}}
\sum_{(j_{2},l_{2})\in\mathcal{S}_{a_{2}}}
\sum_{(j_{3},l_{3})\in\mathcal{S}_{a_{3}}}
\frac{4 |({P}(a_1, a_2, a_3,l_1,l_2,l_3,j_1,j_2,j_3), u^{(1)})|}{|u^{(1)}|}\\
&+\sum_{(r,s)\in\mathbb{E}_{3}^{0^*}(k)}|\omega^{<2r>}\omega^{<2s>}|
\end{split}
\end{equation}
where 
\begin{equation}
\vec{P}(a_1, a_2, a_3,l_1,l_2,l_3,j_1,j_2,j_3)
=\prod_{i=1}^{3}(c_{2a_{i},2l_{i}+1,j_{i}}\frac{u^{(j_{i})}}
{|u^{(j_{i})}|}).
\end{equation}
To discuss the first term on the right hand side of the inequality above, we define
\begin{equation}
\label{eq:W}
\begin{split}
  &\mathcal{W}(a_{1},a_{2},a_{3},s_{1},s_{2},s_{3},j)\\
 =&\sum_{(j_{1},l_{1})\in F_{s_1,a_{1}}} \sum_{(j_{2},l_{2})\in F_{s_2,a_2}} \sum_{(j_{3},l_{3})\in F_{s_3,a_3}} \\
 &\frac{|(u^{(j_{1})}u^{(j_{2})}u^{(j_{3})},u^{(j)})|}{|u^{(j_{1})}||u^{(j_{2})}||u^{(j_{3})}||u^{(j)}|} |c_{2a_{1},2l_{1}+1,j_{1}}| |c_{2a_{2},2l_{2}+1,j_{2}}| |c_{2a_{3},2l_{3}+1,j_{3}}|
\end{split}
\end{equation}
and derive the following proposition
\begin{proposition}
(1). If $a_1=a_2=0$ and $a_3\neq 0$, then 
\begin{equation}
\mathcal{W}(0,0,a_3,1,1,1,j)<\frac{(1+\varepsilon)2C_3}{n} \mathcal{R}_{a_3}, \quad j\notin\mathbb{D}_3
\end{equation}
and 
\begin{equation}
\mathcal{W}(0,0,a_3,s_1,s_2,s_3,\tilde{j})\ll\frac{1}{n}\mathcal{R}_{a_3}
\end{equation}
for $\tilde{j}\in\mathbb{D}_3$ or $(s_1,s_2,s_3)\neq (1,1,1)$.

If $a_1=0$ and $a_2 a_3\neq 0$, then 
\begin{equation}
\mathcal{W}(0,a_2,a_3,1,1,1,j)<\frac{(1+\varepsilon)2C_3^2}{n}\mathcal{R}_{a_2} \mathcal{R}_{a_3}, \quad j\notin\mathbb{D}_3
\end{equation}
and 
\begin{equation}
\mathcal{W}(0,a_2,a_3,s_1,s_2,s_3,\tilde{j})\ll\frac{1}{n}\mathcal{R}_{a_2} \mathcal{R}_{a_3}
\end{equation}
for $\tilde{j}\in\mathbb{D}_3$ or $(s_1,s_2,s_3)\neq (1,1,1)$.

If $a_1 a_2 a_3 \neq 0$, then 
\begin{equation}
\mathcal{W}(a_1,a_2,a_3,1,1,1,j)<\frac{(1+\varepsilon)2C_3^3}{n} \mathcal{R}_{a_1}\mathcal{R}_{a_2}\mathcal{R}_{a_3}, \quad j\notin\mathbb{D}_3;
\end{equation}
and 
\begin{equation}
\mathcal{W}(a_1,a_2,a_3,s_1,s_2,s_3,\tilde{j})\ll\frac{1}{n}\mathcal{R}_{a_1}\mathcal{R}_{a_2}\mathcal{R}_{a_3}
\end{equation}
for $\tilde{j}\in\mathbb{D}_3$ or $(s_1,s_2,s_3)\neq (1,1,1)$.

(2). Utilizing the inequalities in \ref{prop:inequalities_ab_abc}, we have
\begin{equation}
\begin{split}
&\sum_{(a_{1},a_{2},a_{3})\in \mathbb{E}_{1}^{0}(k)}
\sum_{(j_{1},l_{1})\in\mathcal{S}_{a_{1}}}
\sum_{(j_{2},l_{2})\in\mathcal{S}_{a_{2}}}
\sum_{(j_{3},l_{3})\in\mathcal{S}_{a_{3}}}
\frac{|({P}(a_1, a_2, a_3,l_1,l_2,l_3,j_1,j_2,j_3), u^{(j)})|}{|u^{(j)}|} \\
=&\sum_{(a_{1},a_{2},a_{3})\in \mathbb{E}_{1}^{0}(k)}
\sum_{s_1=0}^3 \sum_{s_2=0}^3 \sum_{s_3=0}^3
\mathcal{W}(a_1,a_2,a_3,s_1,s_2,s_3,j) \\
<&6\sum_{(a_{1},a_{2},a_{3})\in \mathbb{E}_{1}^{1^*}(k)}\frac{(1+\varepsilon)2C_3^3}{n} \mathcal{R}_{a_1}\mathcal{R}_{a_2}\mathcal{R}_{a_3}\\
&+6\sum_{(a_{2},a_{3})\in \mathbb{E}_{3}^{1^*}(k-1)}\frac{2(1+\varepsilon)C_3^2}{n}\mathcal{R}_{a_2} \mathcal{R}_{a_3}+
3\frac{2(1+\varepsilon)C_3}{n}\mathcal{R}_{k}\\
<&\frac{\mathcal{R}_{k+1}}{n^2}\cdot\frac{12C_3}{f}(1+8C_1C_3+32C_1^2C_3^2)(1+\varepsilon).
\end{split}
\end{equation}

\end{proposition}
On the other hand, the second term on the right hand side of equation \eqref{eq:omega 2k+2} can be estimated as
\begin{equation}
\begin{split}
\sum_{(r,s)\in\mathbb{E}_{3}^{0^*}(k)}|\omega^{<2r>}\omega^{<2s>}|\leq& \sum_{(r,s)\in\mathbb{E}_{3}^{1^*}(k)}\frac{2C_1^2C_2^2f^{k+1}n^{k-3}}{(r^2+2)(s^2+2)}\\
<&\frac{8C_1C_2^2\mathcal{R}_{k+1}}{n^4}\ll \frac{\mathcal{R}_{k+1}}{n^2}.
\end{split}
\end{equation}
Therefore 
\begin{equation}
    \begin{split}
|\omega^{<2k+2>}|<&\frac{\mathcal{R}_{k+1}}{n^2}\cdot\frac{24C_3}{f\sqrt{2k_2}}(1+8C_1C_3+32C_1^2C_3^2)(1+\varepsilon)\\
=&\frac{\mathcal{R}_{k+1}}{C_1 n^2}\cdot\frac{39k_1k_2}{640(k_2-k_1)\sqrt{2k_2}}(1+\varepsilon) <\frac{C_2\mathcal{R}_{k+1}}{n^2}.
    \end{split}
\end{equation}

\item $(ii)$: $\sum_{j\in\mathbb{D}_1}|c_{2k+2,1,j}|$

Projecting \eqref{eq:order 2k+3} onto ${\rm span}\{ u^{(j)}(e^{i\tau}+e^{-i\tau})\}$ yields
\begin{equation}
\begin{split}
&\sum_{j\in\mathbb{D}_1}|c_{2k+2,1,j}|\\
<&4\sum_{(a_{1},a_{2},a_{3})\in \mathbb{E}_{1}^{0}(k)}
\sum_{s_1=0}^3 \sum_{s_2=0}^3 \sum_{s_3=0}^3\sum_{j\in\mathbb{D}_1}
\frac{\mathcal{W}(a_1,a_2,a_3,s_1,s_2,s_3,j)}{|(\omega^{(j)})^2-
(\omega^{(1)})^2|} \\
&+\sum_{(r,s,t)\in\mathbb{E}_2^2}
\frac{|\omega^{<2r>}\omega^{<2s>}|\sum_{j\in\mathbb{D}_1} |c_{2t,1,j}|}{|(\omega^{(j)})^2-(
\omega^{(1)})^2|}.
\end{split}
\end{equation}
On the right hand side, the first term yields 
\begin{equation}
\begin{split}
&\sum_{(a_{1},a_{2},a_{3})\in \mathbb{E}_{1}^{0}(k)}
\sum_{s_1=0}^3 \sum_{s_2=0}^3 \sum_{s_3=0}^3 ( \sum_{j=2}^{[n^{1-\varepsilon}]}+\sum_{j=[n^{1-\varepsilon}]+1}^{n_1})
\frac{\mathcal{W}(a_1,a_2,a_3,s_1,s_2,s_3,j)}{|(\omega^{(j)})^2-(\omega^{(1)})^2|}\\
<& \frac{\pi^2}{6}\cdot \frac{2(k_2-k_1)}{k_1 k_2\pi^2}\cdot {\mathcal{R}_{k+1}}\cdot\frac{12C_3}{f}(1+8C_1C_3+32C_1^2C_3^2)(1+\varepsilon)
\end{split}
\end{equation}
where $(\omega^{(j)})^2-(\omega^{(1)})^2\approx\frac{k_1k_2(j^2-1)\pi^2}{2(k_2-k_1)(n-1)^2}$ for $2\leq j\leq n^{1-\varepsilon}$ and $(\omega^{(j)})^2-(\omega^{(1)})^2>\frac{k_1k_2\pi^2}{4(k_2-k_1)n^{2\varepsilon}}$ for $n^{1-\varepsilon}<j<n_1$.
At the same time, the second term is bounded from above as
\begin{equation}
\begin{split}
&\sum_{(r,s,t)\in\mathbb{E}_2^2(k)}\frac{|\omega^{<2r>}\omega^{<2s>}|\sum_{j\in\mathbb{D}_1} |c_{2t,1,j}|}{|(\omega^{(j)})^2-(\omega^{(1)})^2|} \\ 
<& \sum_{(s,t)\in\mathbb{E}_3^{0^*}(k)}2\omega^{<0>}\frac{|\omega^{<2s>}|\sum_{j\in\mathbb{D}_1} |c_{2t,1,j}|}{|(\omega^{(j)})^2-(\omega^{(1)})^2|}
+ \sum_{(r,s,t)\in\mathbb{E}_2^{0^*}(k)}\frac{|\omega^{<2r>}\omega^{<2s>}|\sum_{j\in\mathbb{D}_1} |c_{2t,1,j}|}{|(\omega^{(j)})^2-(\omega^{(1)})^2|} \\
<& \sum_{(s,t)\in\mathbb{E}_3^{1^*}(k)}4\sqrt{2k_2}(1+\varepsilon)C_2C_3\mathcal{R}_{s}\mathcal{R}_{t}\frac{2(k_2-k_1)}{k_1k_2\pi^2} \\
<& \mathcal{R}_{k+1}\frac{32\sqrt{2k_2}C_1C_2C_3(k_2-k_1)}{k_1k_2\pi^2}(1+\varepsilon).
\end{split}
\end{equation}
As a result, 
\begin{equation}
\begin{split}
    \sum_{j\in\mathbb{D}_1}|c_{2k+2,1,j}|
    <& C_3\mathcal{R}_{k+1}\frac{(1+\varepsilon)(k_2-k_1)}{k_1k_2\pi^2}(32\sqrt{2k_2}C_1C_2+\frac{16\pi^2}{f}(1+8C_1C_3+32C_1^2C_3^2)) \\
    <&C_3\mathcal{R}_{k+1}(\frac{1}{3}+\frac{13}{20})(1+\varepsilon)<C_3\mathcal{R}_{k+1}.
\end{split}
\end{equation}

\item $(iii)$: $\sum_{j\in\mathbb{D}_2}|c_{2k+2,1,j}|$

At the next step, we again project \eqref{eq:order 2k+3} onto ${\rm span}\{ u^{(j)}(e^{i\tau}+e^{-i\tau}\}$ and consider $\sum_{j\in\mathbb{D}_2}|c_{2k+2,1,j}|$ as
\begin{equation}
\begin{split}
&\sum_{j\in\mathbb{D}_2}|c_{2k+2,1,j}|\\
<&4\sum_{(a_{1},a_{2},a_{3})\in \mathbb{E}_{1}^{0}(k)}
\sum_{s_1=0}^3 \sum_{s_2=0}^3 \sum_{s_3=0}^3\sum_{j\in\mathbb{D}_2}
\frac{\mathcal{W}(a_1,a_2,a_3,s_1,s_2,s_3,j)}{|(\omega^{(j)})^2-(\omega^{(1)})^2|} \\
&+\sum_{(r,s,t)\in\mathbb{E}_2^2}\frac{|\omega^{<2r>}\omega^{<2s>}|\sum_{j\in\mathbb{D}_2} |c_{2t,1,j}|}{|(\omega^{(j)})^2-(\omega^{(1)})^2|}
\end{split}
\end{equation}
where
\begin{equation}
\begin{split}
&\sum_{(a_{1},a_{2},a_{3})\in \mathbb{E}_{1}^{0}(k)}
\sum_{s_1=0}^3 \sum_{s_2=0}^3 \sum_{s_3=0}^3 \sum_{j\in\mathbb{D}_2}
\frac{\mathcal{W}(a_1,a_2,a_3,s_1,s_2,s_3,j)}{|(\omega^{(j)})^2-(\omega^{(1)})^2|}\\
<& \frac{\mathcal{R}_{k+1}}{2(k_2-k_1)n}\cdot\frac{12C_3}{f}(1+8C_1C_3+32C_1^2C_3^2)
\end{split}
\end{equation}
and 
\begin{equation}
\begin{split}
&\sum_{(r,s,t)\in\mathbb{E}_2^2(k)}\frac{|\omega^{<2r>}\omega^{<2s>}|\sum_{j\in\mathbb{D}_2} |c_{2t,1,j}|}{|(\omega^{(j)})^2-(\omega^{(1)})^2|} \\ 
<& \sum_{(s,t)\in\mathbb{E}_3^{0^*}(k)}2\omega^{<0>}\frac{|\omega^{<2s>}|\sum_{j\in\mathbb{D}_2} |c_{2t,1,j}|}{|(\omega^{(j)})^2-(\omega^{(1)})^2|}
+ \sum_{(r,s,t)\in\mathbb{E}_2^{0^*}(k)}\frac{|\omega^{<2r>}\omega^{<2s>}|\sum_{j\in\mathbb{D}_2} |c_{2t,1,j}|}{|(\omega^{(j)})^2-(\omega^{(1)})^2|} \\
<& \sum_{(s,t)\in\mathbb{E}_3^{1^*}(k)}4\sqrt{2k_2}(1+\epsilon)C_2C_4\mathcal{R}_{s}\mathcal{R}_{t}\frac{1}{2(k_2-k_1)n^3}\ll \frac{\mathcal{R}_{k+1}}{n}.
\end{split}
\end{equation}
Then with a relatively large constant $C_4$, we have $\sum_{j\in\mathbb{D}_2}|c_{2k+2,1,j}|<\frac{C_4\mathcal{R}_{k+1}}{n}$.

\item $(iv)$-$(vii)$

The proof for $(iv)$-$(vii)$ is similar to that for $(iii)$ and again $C_4$ should be chosen large. In fact, it suffices to let $C_4$ be the maximum of the requested values in the discussion on $m=1$ and $(iii)$-$(vii)$.
\end{itemize}

\end{proof}

\begin{proof}
(\textbf{Theorem \ref{theorem}})
According to the definition of the norm of $Q^{(2m)}(\tau)$, we have
\begin{equation}
\begin{split}
  &\quad\|Q^{<2m>}(\tau)\|^{2}
  =\frac{1}{2\pi}
  \sum_{(j,l)\in\mathcal{S}_{m}}
  |c_{2m,2l+1,j}|^{2}
  \\
  &<\frac{1}{2\pi}(\sum_{\lambda=1}^{3}\sum_{j\in\mathbb{D}_{\lambda}}
  |c_{2m,1,j}|
  +\sum_{l=1}^{m}\sum_{j=1}^{n_{1}}
  |c_{2m,2l+1,j}|
  +\sum_{\lambda=2}^{3}\sum_{l=1}^{m}\sum_{j\in\mathbb{D}_{\lambda}}
  |c_{2m,2l+1,j}|)^{2}
  \end{split}
  \label{norm:Q(tau)}
\end{equation}
Combining the conclusions from ($\rmnum{2}$) to ($\rmnum{7}$) in \textbf{Lemma \ref{lemma: main estimates}}, we know
\begin{equation}
  \|Q^{<2m>}(\tau)\|^{2}
  <2C_{3}^{2}\mathcal{R}_{m}^{2}.
\end{equation}
Thus, 
the radius of $Q(\tau)$ denoted by $r_{1}$
in \eqref{eq:series1} satisfies
\begin{equation}
\frac{1}{r^{2}_{1}}=\lim_{m\rightarrow +\infty}
\frac{\mathcal{R}_{m+1}}
{\mathcal{R}_{m}} 
=fn
\end{equation}

Similarly, we known from ($\rmnum{1}$)
in \textbf{Lemma \ref{lemma: main estimates}}, the radius of series $\omega$ denoted by $r_{2}$ satisfies $r_{1}=r_{2}\sim\mathcal{O}(\frac{1}{\sqrt{n}})$.

\end{proof}

Here \textbf{Theorem \ref{theorem}} guarantees that the expansions of $Q(\tau)$ and $\omega$ in \eqref{eq:series1} hold up to $\epsilon\sim O(\frac{1}{\sqrt{n}})$, beyond which the frequency crosses the lower edge of optical band and most of the assumptions fail.
As the frequency exits the optical band and enters the band gap, numerical results show that the corresponding state $Q$ grows more middle-localized. It should be noticed that this is a typical mechanism for the emergence of localized states as their frequencies should be outside the bands. As another example, if the linear chain \eqref{eq:motion} is equipped with $k_{3,1}\not\approx 2k_2$ and $|k_{3,2}-2k_2|\ll \frac{1}{n}$, then near the lower edge of the optical band ($(\omega^{(k)})^{2}\approx 2k_2$) we have $\theta^{(k)}\approx \frac{(k-\frac{1}{2})\pi}{n-1}$. When cubic nonlinearity is adopted and \eqref{eq:motion nl} is considered, then in the optical band a frequency will cross the lower edge and its state $u$ will become end-localized (a nonlinear edge state). It is worth noting that most of our discussion on the emergence of nonlinear middle-localized state can also apply to the case of this nonlinear edge state. Unfortunately, due to space constraints, we do not include that part in this work and will leave it in future reports.

To understand the process of a state growing localized, a rough but straightforward perspective is considering the approximation of $Q$ for small $\epsilon$
\begin{equation}
\label{eq:tilde Q}
 Q\approx\tilde{Q}(\tau)
 = (\epsilon \frac{u^{(1)}}{|u^{(1)}|}+\epsilon^3 c_{2,1,3} \frac{u^{(3)}}{|u^{(3)}|})(e^{i\tau}+e^{-i\tau})
\end{equation}
where $c_{2,1,3}<0$. In the generic case $k_{3,1}\not\approx 2k_2\not\approx k_{3,2}$, the spatial envelopes of $u^{(1)}$ and $u^{(3)}$ are approximately $\sin x$ and $\sin 3x$ ($x\in (0,\pi)$), respectively. Then adding $\epsilon^2 c_{2,1,3}\sin 3x$ to $\sin x$ basically raises its middle and flattens its tails, making it more middle-localized. In the same spirit, when $k_{3,1}\not\approx 2k_2$ and $|k_{3,2}-2k_2|\ll \frac{1}{n}$, the spatial envelopes of $u^{(1)}$ and $u^{(2)}$ look like $\sin \frac{x}{2}$ and $\sin \frac{3x}{2}$ ($x\in (0,\pi)$), respectively. Then again adding $\epsilon^2 c_{2,1,2}\sin \frac{3x}{2}$ makes $\sin\frac{x}{2}$ more right-localized.

In order to better quantitatively study the localized property, we define 
\begin{equation}
p_{l_1,l_2}[f]=\frac{|(I_{l_1,l_2}, f^2)|}{\|f\|^2}
\end{equation}
where $f$ is a vector and $I_{l_1,l_2}$ is the vector with $l_1$-th to $l_2$-th elements being one while other elements being zero.
If the interval $[l_1,l_2]$ locates in the middle of the chain, as the state $Q$ gets localized, $p_{l_1,l_2}[Q(0)]$ usually increases accordingly. 
\begin{remark}
When $l_1$ and $l_2$ in $I_{l_1,l_2}$ take different values, in many cases the change of $p_{l_1,l_2}[Q(0)]$ over $\epsilon$ (or energy of $Q$) also has an increasing trend (See panel (c) in Fig.~\ref{figure:nonlinear} for instance). Moreover, even when $I_{l_1,l_2}$ is replaced by some other weight vector $W$, $\frac{|( W, Q(0)^2 )|}{\| Q(0)\|^2}$ may still increase as $\epsilon$ grows. As another example, the increasing trend of $\frac{|( W, Q(0)^2 )|}{\| Q(0)\|^2}$ with 
\begin{equation}
\label{eq:Weight}
W_j=
\begin{cases}
0, & \quad 1\leq j\leq 26~{\rm or}~275\leq j\leq 300; \\
1, & \quad 53\leq j\leq 248; \\
\frac{j-26}{27}, & \quad 27\leq j\leq 52; \\
\frac{275-j}{27}, & \quad 249\leq j\leq 274
\end{cases}
\end{equation}
is demonstrated in panel (c) of Fig.~\ref{figure:nonlinear}.
With this being said, we will just consider $p_{m_1,2n+1-m_1}[Q(0)]$ for a given $m_1$ in the following discussion.
\end{remark}

On the other hand, when $\epsilon$ is small, we can calculate $p_{m_1,2n+1-m_1}[\tilde{Q}(0)]$ explicitly as
\begin{equation}
p_{m_1,2n+1-m_1}[\tilde{Q}(0)]=\frac{ \sum_{j=m_1}^{2n+1-m_1}(\frac{u_j^{(1)}}{|u^{(1)}|}+\epsilon^2 c_{2,1,3}\frac{u_j^{(3)}}{|u^{(3)}|})^2 }{1+\epsilon^4 c_{2,1,3}^2}
\end{equation}
where
\begin{itemize}
    \item $\tilde{Q}$ corresponds to the frequency near the lower edge of the optical band;
    \item $m_1$ is chosen such that $m_1\approx\frac{n}{6}$;
    \item $\Delta\theta^{(j)}\approx\frac{j\pi}{n}$, $c_{2,1,3}\approx -\frac{3(k_2-k_1)n}{16k_1k_2\pi^2}$, $|u^{(1)}|^2\approx n$ and $|u^{(3)}|^2\approx n$.
\end{itemize}
With the estimates above we can obtain the following form of $p_{m_1,n+1-m_1}[\tilde{Q}(0)]$ as
\begin{equation}
p_{m_1,2n+1-m_1}[\tilde{Q}(0)]=\frac{\frac{1}{|u^{(1)}|^{2}}(q-\frac{T_{1,+}}{4}-\frac{T_{1,-}}{4})+
\frac{2\epsilon^{2}c_{2,1,3}}{|u^{(1)}||u^{(3)}|}T
+\frac{\epsilon^{4}c_{2,1,3}^{2}}{|u^{(3)}|^{2}}(q-\frac{T_{3,+}}{4}-\frac{T_{3,-}}{4})}{1+\epsilon^4 c_{2,1,3}^2}
\end{equation}
where $q=n+1-m_1$,
\begin{equation}
\label{T_j,pm}
\begin{split}  T_{j,\pm}&=\frac{(\cos2\Delta\theta^{(j)}-1)[\cos2(\Delta p_{j,\pm}+(n-\frac{m_{1}-1}{2})\Delta\theta^{(j)})-\cos2(\Delta p_{j,\pm}+\frac{m_{1}-1}{2}\Delta\theta^{(j)})]}
{1-\cos2\Delta\theta^{(j)}}\\
      &\quad+
 \frac{\sin2\Delta\theta^{(j)}[\sin2(\Delta p_{j,\pm}+(n-\frac{m_{1}-1}{2})\Delta\theta^{(j)})-\sin2(\Delta p_{j,\pm}+\frac{m_{1}-1}{2}\Delta\theta^{(j)})]}
{1-\cos2\Delta\theta^{(j)}}\\
  &\approx\frac{-2\sin m_{1}\Delta\theta^{(j)}\sin2\Delta\theta^{(j)}}
 {2(\Delta\theta^{(j)})^{2}}\approx\frac{-2\sin m_{1}\Delta\theta^{(j)}}
 {\Delta\theta^{(j)}},(j=1,3)
\end{split}
\end{equation}
and
\begin{equation}
\label{T}
\begin{split}
    T&=\sum^{\frac{2n+1-m_{1}}{2}}_{j=\frac{m_{1}+1}{2}}
    \sin (\Delta p_{1,+}+(j-1)\Delta\theta^{(1)})
    \sin (\Delta p_{3,+}+(j-1)\Delta\theta^{(3)})\\
    &\quad+\sum^{\frac{2n+1-m_{1}}{2}}_{j=\frac{m_{1}+1}{2}}
    \sin (\Delta p_{1,-}+(j-1)\Delta\theta^{(1)})
    \sin (\Delta p_{3,-}+(j-1)\Delta\theta^{(3)})\\
    &\approx \frac{2\sin(\Delta\theta^{(1)}+\Delta\theta^{(3)})
    \sin \frac{m_{1}}{2}(\Delta\theta^{(1)}+\Delta\theta^{(3)})}
    {(\Delta\theta^{(1)}+\Delta\theta^{(3)})^{2}}
    -\frac{2\sin(\Delta\theta^{(1)}-\Delta\theta^{(3)})
    \sin \frac{m_{1}}{2}(\Delta\theta^{(1)}-\Delta\theta^{(3)})}
    {(\Delta\theta^{(1)}-\Delta\theta^{(3)})^{2}}\\
    &\approx \frac{n(\sqrt{3}-2)}{4\pi}.
\end{split}
\end{equation}
where $\Delta p_{j,\pm}=\Delta\beta^{(j)}\pm\Delta\alpha^{(j)}$ as defined in \eqref{Delta_p}.

To be more specific, 
\begin{equation}
\label{eq:p_approx}
p_{m_1,2n+1-m_1}[\tilde{Q}(0)]\approx \frac{(\frac{5}{6}+\frac{1}{2\pi})+n\epsilon^2\frac{3(k_2-k_1)(2-\sqrt{3})}{32k_1k_2\pi^3}+n^2\epsilon^4(\frac{5}{6}+\frac{1}{3\pi})[\frac{3(k_2-k_1)}{16k_1k_2\pi^2}]^2}{1+n^2\epsilon^4[\frac{3(k_2-k_1)}{16k_1k_2\pi^2}]^2}.
\end{equation}
Since $\frac{3(k_2-k_1)(2-\sqrt{3})}{32k_1k_2\pi^3}>0$, $p_{m_1,2n+1-m_1}[\tilde{Q}(0)]$ is an increasing function for small $\epsilon$ ($n\epsilon^2\ll 1$). In other words, this is an explicit evidence that the eigenmode gradually becomes localized as its frequency tends to exit the band (see the panel (b) of Fig.~\ref{figure:nonlinear} for the comparison between $p_{m_1,2n+1-m_1}[\tilde{Q}(0)]$ and $p_{m_1,2n+1-m_1}[{Q}(0)]$). When $\epsilon$ is large enough for the frequency to touch the band edge, \textbf{Theorem \ref{theorem}} and the approximation for $Q$ no longer stay valid, but numerical results show that the state $Q$ continues to become localized. An example of the change of $Q$ in shape and in $p_{m_1,2n+1-m_1}[Q(0)]$ is illustrated in the panels (c) and (d) of Fig.~\ref{figure:nonlinear}.

\begin{figure}[!htp]
\begin{minipage}{0.5\linewidth}
\leftline{(a)}
\centerline{\includegraphics[height = 5cm, width = 6cm]{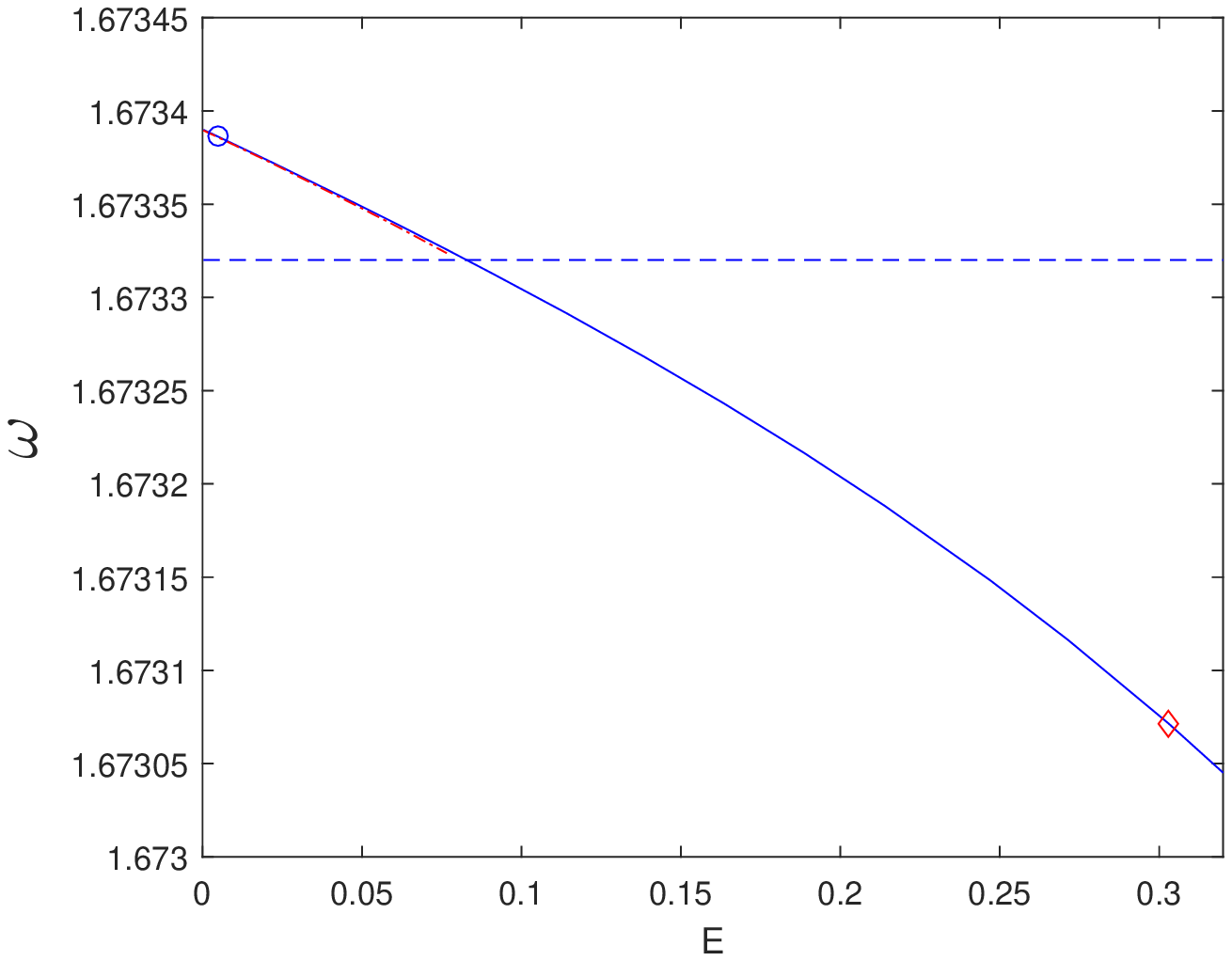}}
\end{minipage}
\hfill
\begin{minipage}{0.5\linewidth}
\leftline{(b)}
\centerline{\includegraphics[height = 5cm, width = 6cm]{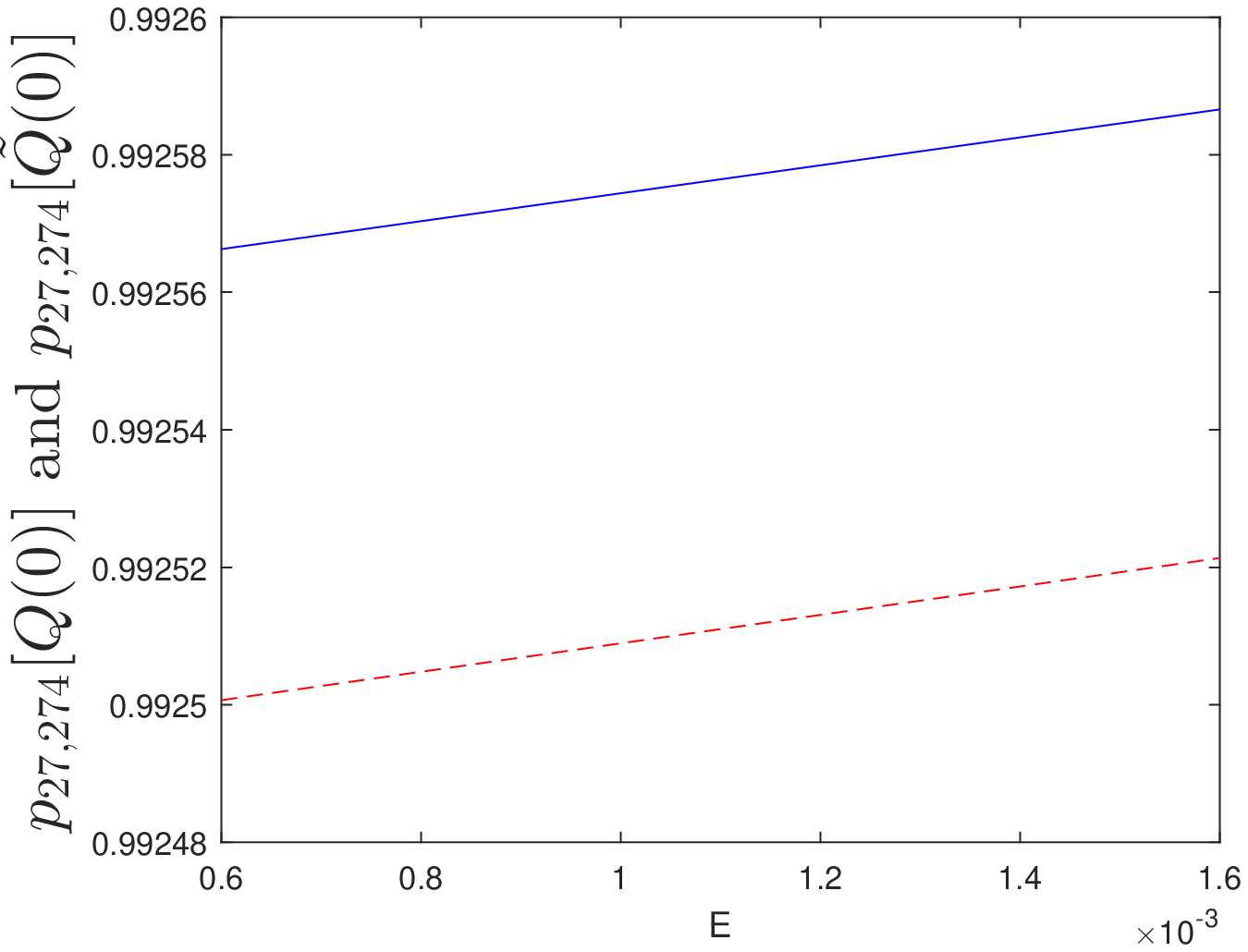}}
\end{minipage}

\begin{minipage}{0.5\linewidth}
\leftline{(c)}
\centerline{\includegraphics[height = 5cm, width = 6cm]{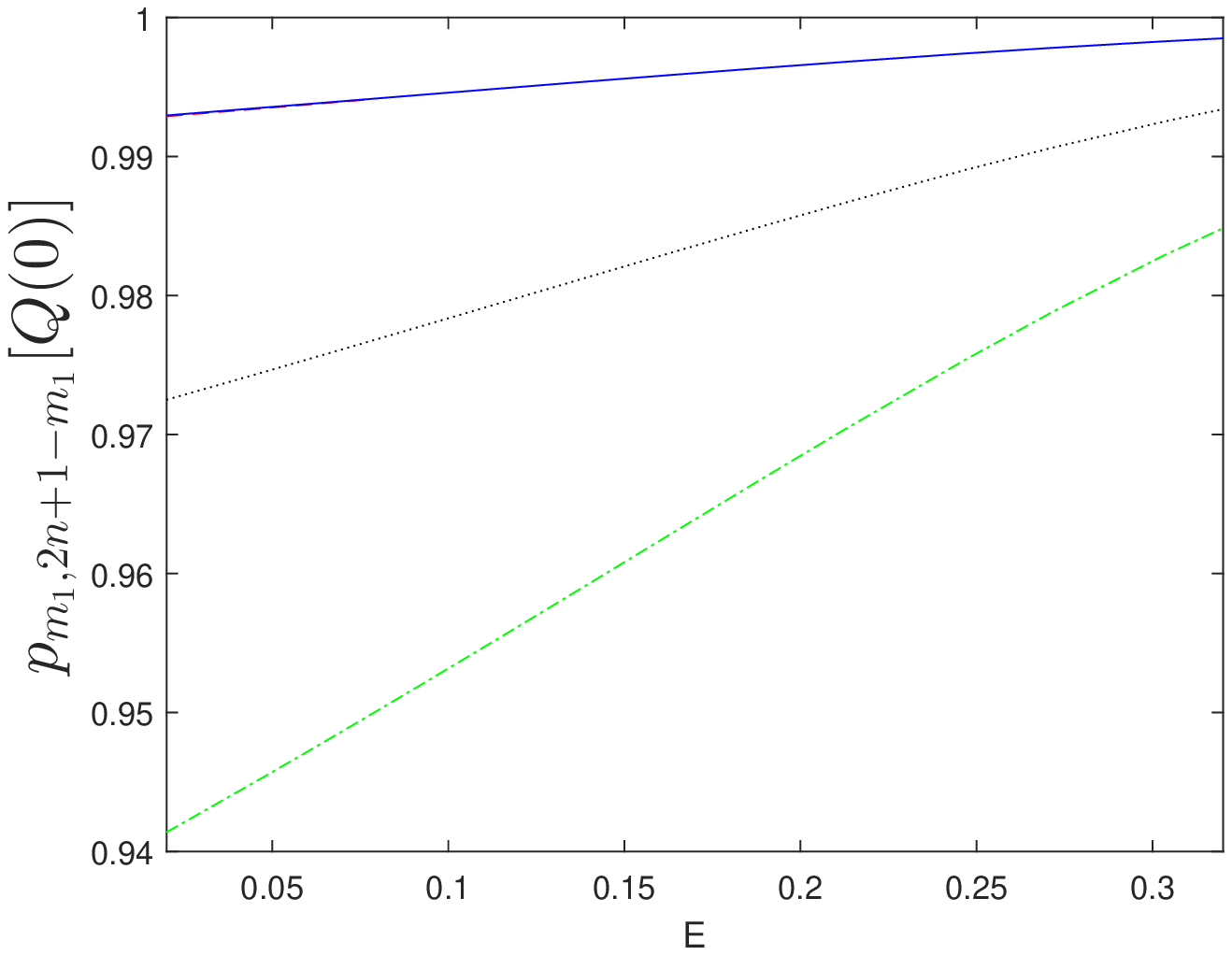}}
\end{minipage}
\hfill
\begin{minipage}{0.5\linewidth}
\leftline{(d)}
\centerline{\includegraphics[height = 5cm, width = 6cm]{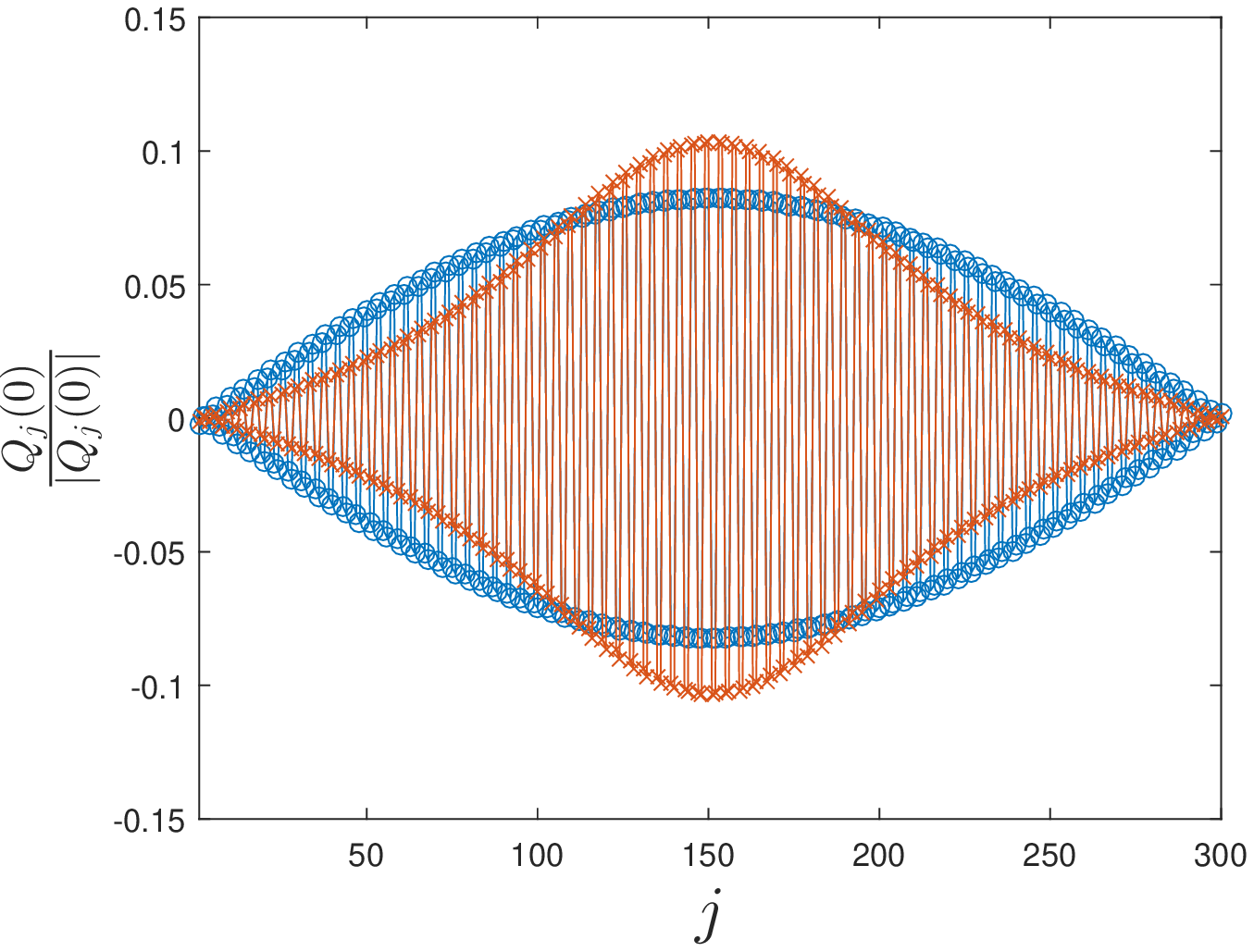}}
\end{minipage}

\caption{Here we show the change of nonlinear time-periodic state and its frequency for the system \eqref{eq:motion nl0} with $n=150$, $k_1=0.6$,
$k_{2}=1.4$, $k_{3,1}=2$ and $k_{3,2}=1.6$.  
The panel (a) shows the change of frequency $\omega$ (blue solid line) near the lower edge of the optical band (blue horizontal dashed line) over the growth of energy in \eqref{eq:energy} and its approximation from \eqref{eq:linear approx} in red ``$\cdot -$'' line. 
In panel (b) we plot the comparison between $p_{m_1,2n+1-m_1}[{Q}(0)]$ (blue solid line) and the approximation of $p_{m_1,2n+1-m_1}[\tilde{Q}(0)]$ in \eqref{eq:p_approx} (red dashed line) for small energy (or $\epsilon$). In panel (c) blue solid line and green ``$\cdot -$'' line represent  $p_{m_1,2n+1-m_1}[{Q}(0)]$ over energy with $m_1=27$ and $m_1=53$ respectively. In addition, $\frac{|( W, Q(0)^2 )|}{\| Q(0)\|^2}$ with $W$ defined in \eqref{eq:Weight} and the approximation of $p_{m_1,2n+1-m_1}[\tilde{Q}(0)]$ are shown in black dotted line and red dashed line respectively.
The panel (d) shows the normalized nonlinear state in blue circles (red crosses) with its frequency highlighted by ``o"( ``$\diamond$") in panel(a).
}
\label{figure:nonlinear}
\end{figure}

\section{Conclusions and future challenges}
\label{sec:conclusion}
In the present work, we have considered a class of long diatomic chains and provided a systematic perspective concerning the emergence of nonlinear localized states (or breathers) with frequencies crossing the band edges. To achieve that goal, we have adopted the long chain limit and derived approximated forms of the eigenstates in the linear diatomic chain model. Particularly linear edge states with frequencies outside the bands and linear non-localized states with frequencies near the band edges have been estimated, showing different characteristics in settings with different boundary conditions. Among others two representative situations are $k_{3,1}\not\approx 2k_2\not\approx k_{3,2}$ and $k_{3,1}\not\approx 2k_2\approx k_{3,2}$, which in the nonlinear regime can probably lead to the emergence of middle-localized states and end-localized states (edge states) from the lower edge of the optical band. For the former more generic case, we have conducted nonlinear continuation of the linear eigenstate with frequency closet to the band edge and obtained an approximation of the nonlinear state. Based on our knowledge of the linear eigenstates, it can be proved that the approximation continued from the linear limit remains valid before its amplitude reaches $O(\frac{1}{\sqrt{n}})$ (chain length~$\sim O(n)$) and the frequency crosses the band edge. Moreover, growing of the corresponding localized form has been explicitly identified and numerically illustrated.

The long-chain assumption in this work has connected the symmetry of infinite chains with the finiteness of real chains and offered a frame for asymptotic analysis.
Although we have mainly focused on the emergence of middle-localized states, the corresponding procedure can also be applied to the study of nonlinear edge states. Moreover, the large-lattice limit has the potential to play an role in investigating the existence and stability of localized states in lattices with more complex structures or in higher dimensions. 
Nevertheless, there are numerous open questions under this topic that merit further
exploration in the near future. For instance, we have established the relation between the boundary conditions and the existence of different localized states. It would be favorable to redescribe this relation in the frame of topology and relate it to the bulk-boundary correspondence. Next challenge is the analysis of the nonlinear states with frequencies just exiting the bands or outside the bands since the expansion from the linear limit diverges in these cases. Besides, another question of interest concerns the identification of localization. We have employed the ratio between the weighted and unweighted norms to represent the extent of localization which is dependent on the choice of weights. A more general and universal characterization for localization could help revealing the connection between the localization and the frequency. 
Some of these directions are currently under consideration and will be reported in future studies.

\vspace{5mm}

{\it Acknowledgements.} HX gratefully acknowledges that this work is partially supported by NSFC (Grant No. 11801191).

\begin{appendices}

\section{Appendix: On the Special cases of edge states}
\label{sec:special cases proof}

Proof of Lemma~\ref{lemma: k31_approx_k2}:

In \eqref{eq:boundary conditions_1}\eqref{eq:boundary conditions_2} with $n\gg 1$, if $k_{3,1}=k_2+\delta k_{3,1}$ and $|\delta k_{3,1}|\ll 1$, then $k_{3,2}\not\approx k_2$ is a necessary condition to the existence of a left edge state with $a\approx -\frac{k_1}{k_2}$ and $|c_2 a^{2-2n}|\leq O(1)$.

\begin{proof}
We follow the notations in Section~\ref{subsubsec:k31_approx_k2}.  
\begin{itemize}
\item If $|\delta a|\ll |\delta\tilde{a}|\ll 1$ (here $\delta\tilde{a}<\delta a<0$), then $c_2\approx -\frac{k_2^2}{k_1}\sqrt{\frac{k_2\Delta a}{k_1(k_1^2-k_2^2)}}$. This is the situation where $\delta a$ and $\delta k_{3,1}$ are small and the results are basically the same as those in the subcase $k_{3,1}=k_2$. 
\begin{itemize}
\item When $\sqrt{\delta\tilde{a}}\sim \sqrt{\Delta a}\ll a^{2n-2}$:\\
Then $1\ll k_{3,2}\approx \frac{k_1 a^{2n-2}}{\sigma c_2}\sim O(\frac{a^{2n-2}}{\sqrt{\Delta a}})\ll \frac{a^{2n-2}}{\sqrt{\delta a}}$. Since $c_2 a^{2-2n}\ll 1$, the corresponding state $u$ is a left edge state.
\item When $\sqrt{\delta\tilde{a}}\sim \sqrt{\Delta a}\sim a^{2n-2}$:\\
Then $k_{3,2}\approx \omega^2-k_1+ \frac{k_1 a^{2n-2}}{\sigma c_2}\leq O(1)$ and $k_{3,2}\not\approx \omega^2-k_1\approx k_2$. Since $c_2 a^{2-2n}\sim O(1)$, the eigenstate $u$ in this case can also be considered as left localized.
\item When $\sqrt{\delta\tilde{a}}\sim \sqrt{\Delta a}\gg a^{2n-2}$:\\
Then $k_{3,2}\approx \omega^2-k_1+ \frac{k_1( v_{11} a^{2n-2}+c_2 v_{12})}{\sigma c_2 v_{11}}\approx \omega^2-k_1\approx k_2$ but now $|c_2 a^{2-2n}|\gg 1$.
\end{itemize}

\item If $|\delta a|\sim |\delta\tilde{a}|$, then $c_2\sim O(\sqrt{\delta a}-\sqrt{\Delta a})\sim O(\sqrt{\delta a})$. At first we consider a special situation that $\Delta a=\delta a+\delta \tilde{a}\ll \delta a$. In this case $c_2\approx \frac{k_2^2}{k_1}\sqrt{\frac{k_2\delta a}{k_1(k_1^2-k_2^2)}}\sim O(\sqrt{\delta a}) \gg O(\sqrt{\Delta a})$. Another (and more generic) situation is $O(\Delta a)\sim O(\delta a)\sim O(\delta\tilde{a})$ and accordingly $c_2\sim O(\sqrt{\delta a})\sim O(\sqrt{\Delta a})$.
\begin{itemize}
\item When $\delta k_{3,1}\sim \sqrt{\delta a}\ll a^{2n-2}$, $k_{3,2}\approx \omega^2-k_1+\frac{k_1 a^{2n-2}}{\sigma c_2}\sim O(\frac{ a^{2n-2}}{\delta k_{3,1}})$.
\item When $\delta k_{3,1}\sim \sqrt{\delta a}\sim a^{2n-2}$, $k_{3,2}\approx \omega^2-k_1+\frac{k_1 a^{2n-2}}{\sigma c_2}\leq O(1)$ and $k_{3,2}\not\approx \omega^2-k_1\approx k_2$.
\item When $\delta k_{3,1}\sim \sqrt{\delta a}\gg a^{2n-2}$, $k_{3,2}\approx \omega^2-k_1+\frac{k_1( v_{11} a^{2n-2}+c_2 v_{12})}{\sigma c_2 v_{11}}\approx \omega^2-k_1\approx k_2$ but here the eigenstate $u$ may not be left localized.
\end{itemize}

\item If $|\delta\tilde{a}|\ll |\delta a|\ll 1$, then $c_2\approx -\frac{k_2^2}{2k_1}\sqrt{\frac{k_2}{k_1(k_1^2-k_2^2)}}\frac{\delta \tilde{a}}{\sqrt{\delta a}}\ll O(|\sqrt{\delta a}|)\sim O(|\sqrt{\Delta a}|)$. 
\begin{itemize}
\item When $|\frac{\delta \tilde{a}}{\sqrt{\delta a}}|\ll a^{2n-2}$, $c_2 a^{2-2n}\ll 1$ hence the corresponding state is a left edge state. To characterize $k_{3,2}$, we consider the following cases:
\begin{itemize}
\item $\frac{\delta\tilde{a}}{\delta a}\ll a^{2n-2}$, $k_{3,2}\approx \frac{k_1 v_{11}}{\sigma v_{12}}\sim O(\frac{1}{\sqrt{\delta a}})$.
\item $\frac{\delta\tilde{a}}{\delta a}\sim a^{2n-2}$, $k_{3,2}\approx \omega^2-k_1+\frac{k_1 v_{11} a^{2n-2}}{\sigma(v_{12}a^{2n-2}+c_2 v_{11})} \geq O(\frac{1}{\sqrt{\delta a}})$.
\item $\frac{\delta\tilde{a}}{\delta a}\gg a^{2n-2}$, $k_{3,2}\approx \omega^2-k_1+\frac{k_1 a^{2n-2}}{\sigma c_2} \sim O(\frac{a^{2n-2}\sqrt{\delta a}}{\delta\tilde{a}})$ hence $O(\frac{a^{2n-2}}{\sqrt{\delta a}})\ll k_{3,2}\ll O(\frac{1}{\sqrt{\delta a}})$.
\end{itemize}
\item When $|\frac{\delta \tilde{a}}{\sqrt{\delta a}}| \gg a^{2n-2}$ (hence $|\sqrt{\delta a}|\sim |\delta k_{3,1}| \gg  a^{2n-2}$), ${k}_{3,2}\approx \omega^2-k_1+ \sigma\frac{ k_1 a^{2n-2}}{c_2}\approx \omega^2-k_1\approx k_2$. Here we have $|c_2 a^{2-2n}|\sim |\frac{\delta \tilde{a}}{\sqrt{\delta a}}a^{2-2n}|\gg 1$.
\item When $|\frac{\delta \tilde{a}}{\sqrt{\delta a}}| \sim a^{2n-2}$ (hence $|\sqrt{\delta a}|\sim |\delta k_{3,1}| \gg  a^{2n-2}$), $k_{3,2}\approx \omega^2-k_1+\frac{k_1 a^{2n-2}}{\sigma c_2}\leq O(1)$ and $k_{3,2}\not\approx \omega^2-k_1\approx k_2$. 
\end{itemize}

\end{itemize}
The discussion above has already proven Lemma~\ref{lemma: k31_approx_k2}.
In order to better characterize $a$ (or $\Delta a$) of the eigenstate from given $k_{3,1}$ and $k_{3,2}$, the results for left edge states with $k_{3,1}\approx k_2$ can be rephrased as follows:
\begin{itemize}
    \item If $k_{3,1}=k_2$:
    \begin{itemize}
        \item if $k_{3,2}\gg 1$, then $k_{3,2}\approx-\frac{k_1^2 a^{2n-2}\sqrt{k_1(k_1^2-k_2^2)}}{\sigma k_2^2\sqrt{\Delta a}}$ and $\sqrt{\Delta a}\sim O(\frac{a^{2n-2}}{k_{3,2}})$;
        \item if $k_{3,2}\leq O(1)$ and $k_{3,1}\not\approx k_2$, then $k_{3,2}\approx k_2-\frac{k_1^2 a^{2n-2}\sqrt{k_1(k_1^2-k_2^2)}}{\sigma k_2^2\sqrt{\Delta a}}$ and
        $\sqrt{\Delta a}\sim O(a^{2n-2})$. 
    \end{itemize}
    \item If $k_{3,1}=k_2+\delta k_{3,1}$:
    \begin{itemize}     
        \item If $\delta k_{3,1}\ll a^{2n-2}$ and $k_{3,2}\gg O(\frac{a^{2n-2}}{\delta k_{3,1}})$, then $k_{3,2}\approx \frac{\sqrt{k_1(k_1^2-k_2^2)} a^{2n-2}}{\sigma(\sqrt{k_2\Delta a}a^{2n-2}-\frac{k_2^5 (\Delta a-\frac{k_1 (k_1^2-k_2^2)}{k_2^5}\delta k_{3,1}^2)}{2k_1^2\delta k_{3,1}\sigma\sqrt{k_1(k_1^2-k_2^2)}})}$ and $\sqrt{\Delta a}\sim O(\delta k_{3,1})$.
        \item If $\delta k_{3,1}\ll a^{2n-2}$ and
        $k_{3,2}\sim O(\frac{a^{2n-2}}{\delta k_{3,1}})$, then $k_{3,2}\approx \frac{k_1^2 a^{2n-2}}{\delta k_{3,1}-\sigma k_2^2\sqrt{\frac{k_2\Delta a}{k_1(k_1^2-k_2^2)}}}$ and $\sqrt{\Delta a}\leq O(\delta k_{3,1})$.
        \item If $\delta k_{3,1}\ll a^{2n-2}$ and
        $O(1)\ll k_{3,2}\ll O(\frac{a^{2n-2}}{\delta k_{3,1}})$, then $k_{3,2}\approx \frac{k_1^2 a^{2n-2}}{-\sigma k_2^2\sqrt{\frac{k_2\Delta a}{k_1(k_1^2-k_2^2)}}}$ and $O(\delta k_{3,1})\ll\sqrt{\Delta a}\ll O(a^{2n-2}) $.
        \item If $\delta k_{3,1}\ll a^{2n-2}$, 
        $k_{3,2}\leq O(1)$ and $k_{3,2}\not\approx k_2$, then $k_{3,2}\approx \omega^2-k_1+\frac{k_1^2 a^{2n-2}}{-\sigma k_2^2\sqrt{\frac{k_2\Delta a}{k_1(k_1^2-k_2^2)}}}$ and $O(\delta k_{3,1})\ll\sqrt{\Delta a}\sim O(a^{2n-2}) $.
        \item If $\delta k_{3,1}\sim a^{2n-2}$, 
        $k_{3,2}\leq O(1)$ and $k_{3,2}\not\approx k_2$, then $k_{3,2}\approx \omega^2-k_1+\frac{k_1^2 a^{2n-2}}{\delta k_{3,1}-\sigma k_2^2\sqrt{\frac{k_2\Delta a}{k_1(k_1^2-k_2^2)}}}$ and $\sqrt{\Delta a}\leq O(\delta k_{3,1})\sim O(a^{2n-2}) $.
        \item If $\delta k_{3,1}\gg a^{2n-2}$, 
        $k_{3,2}\leq O(1)$ and $k_{3,2}\not\approx k_2$, then $k_{3,2}\approx \omega^2-k_1-\frac{2k_1^3(k_2^3-k_2^2)\delta k_{3,1}}{k_2^5(\Delta a-\frac{k_1 (k_1^2-k_2^2)}{k_2^5}\delta k_{3,1}^2)}$ and $\sqrt{\Delta a}\sim O(\delta k_{3,1})$.
        \item If $\delta k_{3,1}\geq O(a^{2n-2})$ and $k_{3,2}\gg O(1)$, then $k_{3,2}\approx \frac{\sqrt{k_1(k_1^2-k_2^2)} a^{2n-2}}{\sigma(\sqrt{k_2\Delta a}a^{2n-2}-\frac{k_2^5 (\Delta a-\frac{k_1 (k_1^2-k_2^2)}{k_2^5}\delta k_{3,1}^2)}{2k_1^2\delta k_{3,1}\sigma\sqrt{k_1(k_1^2-k_2^2)}})}$ and $\sqrt{\Delta a}\sim O(\delta k_{3,1})$.
    \end{itemize}
\end{itemize}

\end{proof}


Proof of Remark~\ref{remark: band edge states}:

\begin{proof}
First we consider the case $a=1$, then $\omega^2=(k_1+k_2)(1-\sigma)$ and
\begin{eqnarray}
\frac{u_1}{u_2}&=&\sigma\frac{c_1+c_2}{c_1-c_2}=\frac{k_1}{k_1+k_{3,1}-\omega^2}, \\
\frac{u_{2n-1}}{u_{2n}}&=&\sigma\frac{c_1+c_2-c_2(n-1)\frac{2(k_1+k_2)}{k_2}}{c_1-c_2-c_2(n-1)\frac{2(k_1+k_2)}{k_2}}=\frac{k_1+k_{3,2}-\omega^2}{k_1}.
\end{eqnarray}
\begin{itemize}
\item If $c_1=0$, then 
\begin{itemize}
\item $\sigma=1$: $k_{3,1}<0$.
\item $\sigma=-1$: $k_{3,1}=2k_1+2k_2$ and $k_{3,2}=2k_2+\frac{2k_1}{1+(n-1)\frac{2k_1+2k_2}{k_2}}\approx 2k_2+\frac{k_2k_2}{n(k_1+k_2)}\approx k_2$.
\end{itemize}
\item If $c_1=1$, then
\begin{itemize}
    \item $\sigma=1$: $k_{3,1}=k_{3,2}=0$.
    \item $\sigma=-1$: $\omega^2=2k_1+2k_2$, $c_2=\frac{2k_2-k_{3,1}}{k_{3,1}-2k_1-2k_2}$, $k_{3,2}=2k_2-\frac{2c_2 k_1}{1-c_2-c_2(n-1)\frac{2k_1+2k_2}{k_2}}$.
    \begin{itemize}
        \item $c_2=0$: $k_{3,1}=k_{3,2}=2k_2$.
        \item $c_2<0$: $k_{3,1}<2k_2$ or $k_{3,1}>2k_1+2k_2$, $k_{3,2}>2k_2$ and  $k_{3,2}-2k_2\approx \frac{-c_2 k_1}{1-c_2(n-1)\frac{2k_1+2k_2}{k_2}}\leq O(\frac{1}{n})$.
        \item $c_2>0$: $2k_2<k_{3,1}<2k_1+2k_2$,
        \begin{itemize}
            \item $c_2\gg \frac{1}{n}$, $k_{3,2}\approx 2k_2+\frac{k_1 k_2}{n(k_1+k_2)}\approx 2k_2$.
            \item $c_2\ll \frac{1}{n}$, $k_{3,1}-2k_2\ll O(\frac{1}{n})$, $k_{3,2}\approx 2k_2-2c_2 k_1\approx 2k_2$.
            \item $c_2\sim O(\frac{1}{n})$: $k_{3,1}-2k_2\sim O(\frac{1}{n})$. If $|1-c_2-c_2(n-1)\frac{2k_1+2k_2}{k_2}|\sim O(1)$, then $k_{3,2}-2k_2\sim O(\frac{1}{n})$; If $|1-c_2-c_2(n-1)\frac{2k_1+2k_2}{k_2}|\ll 1$, then $|k_{3,2}-2k_2|\geq O(\frac{1}{n})$.
        \end{itemize}
    \end{itemize}
\end{itemize}
\end{itemize}
Then it comes to the case $a=-1$ where $\omega^2=(1-\sigma)k_1+(1+\sigma)k_2$ and
\begin{eqnarray}
\frac{u_1}{u_2}&=&\sigma\frac{c_1+c_2}{c_1-c_2}=\frac{k_1}{k_1+k_{3,1}-\omega^2}, \\
\frac{u_{2n-1}}{u_{2n}}&=&\sigma\frac{c_1+c_2-c_2(n-1)\frac{2(k_2-k_1)}{k_2}}{c_1-c_2-c_2(n-1)\frac{2(k_2-k_1)}{k_2}}=\frac{k_1+k_{3,2}-\omega^2}{k_1}.
\end{eqnarray}
\begin{itemize}
\item If $c_1=0$, then 
\begin{itemize}
\item $\sigma=1$: $k_{3,1}=2k_1-2k_2$, $k_{3,2}=2k_2-\frac{2k_1}{1+(n-1)\frac{2k_2-2k_2}{k_2}}\approx 2k_2-\frac{2k_1k_2}{n(k_2-k_1)}$.
\item $\sigma=-1$: $k_{3,1}=2k_1$ and $k_{3,2}=\frac{2k_1}{1+(n-1)\frac{2k_2-2k_1}{k_2}}\approx \frac{k_2k_2}{n(k_2-k_1)}\approx 0$.
\end{itemize}
\item If $c_1=1$, then
\begin{itemize}
    \item $\sigma=1$: $\omega^2=2k_2$, $c_2=\frac{2k_2-k_{3,1}}{k_{3,1}+2k_1-2k_2}$, $k_{3,2}=2k_2+\frac{2c_2k_1}{1-c_2-c_2(n-1)\frac{2k_2-2k_1}{k_2}}$.
    \begin{itemize}
        \item $c_2=0$: $k_{3,1}=k_{3,2}=2k_2$.
        \item $c_2<0$: $k_{3,1}>2k_2$ or $k_{3,1}<2k_2-2k_1$, $0<2k_2-k_{3,2}\leq O(\frac{1}{n}$.
        \item $c_2>0$: $2k_2-2k_1<k_{3,1}<2k_2$, 
        \begin{itemize}
            \item $c_2\gg \frac{1}{n}$, $k_{3,2}\approx 2k_2-\frac{k_1 k_2}{n(k_2-k_1)}\approx 2k_2$.
            \item $c_2\gg \frac{1}{n}$, $k_{3,2}\approx 2k_2+2c_2k_1\approx 2k_2$.
            \item $c_2\sim O(\frac{1}{n})$: $0<2k_2-k_{3,1}\sim O(\frac{1}{n})$. If $|1-c_2-c_2(n-1)\frac{2k_2-2k_1}{k_2}|\sim O(1)$, then $|k_{3,2}2-2k2|\sim O(\frac{1}{n})$; If $|1-c_2-c_2(n-1)\frac{2k_2-2k_1}{k_2}|\ll O(1)$, then $|k_{3,2}-2k_2|\gg O(\frac{1}{n})$.
        \end{itemize}
    \end{itemize}
    \item $\sigma=-1$: $\omega^2=2k_1$, $c_2=\frac{-k_{3,1}}{k_{3,1}-2k_1}$, $k_{3,2}=-\frac{2c_2 k_1}{1-c_2-c_2(n-1)\frac{2k_2-2k_1}{k_2}}$.
    \begin{itemize}
        \item $c_2=0$: $k_{3,1}=k_{3,2}=0$.
        \item $c_2<0$: $k_{3,1}>2k_1$, $0<k_{3,2}\leq \frac{1}{n}$.
        \item $c_2>0$: $0<k_{3,1}<2k_1$, $1-c_2-c_2(n-1)\frac{2k_2-2k_1}{k_2}<0$ hence $c_2>\frac{1}{1+(n-1)\frac{2k_2-k_1}{k_2}}\approx \frac{k_2}{2n(k_2-k_1)}$. This means $k_{3,1}\geq O(\frac{1}{n})$.
        \begin{itemize}
            \item $c_2\gg \frac{1}{n}$, $k_{3,2}\approx \frac{k_1 k_2}{n(k_2-k_1)}\approx 0$.
            \item $c_2\sim O(\frac{1}{n})$: $k_{3,1}\sim O(\frac{1}{n})$. If $|1-c_2-c_2(n-1)\frac{2k_2-2k_1}{k_2}|\sim O(1)$, then $k_{3,2}\sim O(\frac{1}{n})$; If $|1-c_2-c_2(n-1)\frac{2k_2-2k_1}{k_2}|\ll O(1)$, then $k_{3,2}\gg O(\frac{1}{n})$.
        \end{itemize}
    \end{itemize}
\end{itemize}
\end{itemize}
Comparing $|k_{3,1}-k_2(1-\sigma a)|$ and $|k_{3,2}-k_2(1-\sigma a)|$ in each scenario above concludes the proof.
    
\end{proof}


Proof of Remark~\ref{remark:k31_approx_2k2}:

When $|k_{3,1}-2k_2|\geq O(1) \leq |k_{3,2}-2k_2|$, there does not exist any eigenstate $u$ near $\omega^2=2k_2$ with $-1<a\approx -1$.  

\begin{proof}
    
Since 
\begin{equation}
\label{eq:c2_k31_a}
c_2\approx\frac{\frac{k_1 k_2 \Delta a}{\sqrt{k_1-k_2}}-\delta k_{3,1}\sqrt{k_1-k_2}+\frac{k_1 k_2 (\Delta a)^2}{\sqrt{k_1-k_2}}}{\frac{k_1 k_2 \Delta a}{\sqrt{k_1-k_2}}+\delta k_{3,1}\sqrt{k_1-k_2}}
\end{equation}
depends on the relation between $\delta k_{3,1}$ and $\Delta a$, the following three cases will be investigated and $\Delta k_{3,2}$ from (\eqref{eq:k32_original}) will be obtained accordingly.
\begin{itemize}
\item $|\delta k_{3,1}|\ll |\Delta a|$ ($c_2\approx 1$): 
\begin{itemize}
\item $|\Delta a|\ll \frac{1}{n}$ ($a^{-2n}\approx 1+2n\Delta a$):
\begin{itemize}
\item $|\delta k_{3,1}|\leq O(|\Delta a|^2)$: Then $|c_2-1|\sim O(|\Delta a|)$ and $\Delta k_{3,2}\sim O(n|\Delta a|^2)$.
\item $|\delta k_{3,1}|\gg O(|\Delta a|^2)$: Then $|c_2-1|\sim O(|\frac{\delta k_{3,1}}{\Delta a}|)$.
\begin{itemize}
\item $|\frac{\delta k_{3,1}}{\Delta a}|\ll n|\Delta a|$: Then $1-c_2 a^{2-2n}\approx -2n\Delta a$ and $\Delta k_{3,2}\sim O(n|\Delta a|^2)$.
\item $|\frac{\delta k_{3,1}}{\Delta a}|\gg n|\Delta a|$: Then $1-c_2 a^{2-2n}\sim O(\frac{|\delta k_{3,1}}{\Delta a}|)$ and $\Delta k_{3,2}\sim O(|\delta k_{3,1}|)$.
\item $|\frac{\delta k_{3,1}}{\Delta a}|\sim O(n|\Delta a|$): Then $|1-c_2 a^{2-2n}|\leq O(\frac{|\delta k_{3,1}}{\Delta a}|)$ and $|\Delta k_{3,2}|\leq O(|\delta k_{3,1}|)$.
\end{itemize}
\end{itemize}
\item $|\Delta a|\sim \frac{1}{n}$ ($1\not\approx a^{-2n}\sim O(1)$): Then $\Delta k_{3,2}\sim O(\Delta a)$.
 \item $|\Delta a|\gg \frac{1}{n}$ ($a^{-2n}\gg 1$): Then $\Delta k_{3,2}\sim O(\Delta a)$.
\end{itemize}
\end{itemize}

\begin{itemize}
\item $|\delta k_{3,1}|\sim O(|\Delta a|)$ ($\pm 1\not\approx c_2\sim O(1)$): 
\begin{itemize}
\item $|\Delta a|\ll \frac{1}{n}$ ($a^{-2n}\approx 1+2n\Delta a$): Then $1-c_2 a^{2-2n}\sim 1+c_2 a^{2-2n}\sim O(1)$ and $\Delta k_{3,2}\sim O(\Delta a)$.
\item $|\Delta a|\sim \frac{1}{n}$ ($1\not\approx a^{-2n}\sim O(1)$): 
\begin{itemize}
    \item $1-c_2 a^{2-2n}\approx 0$: Then $|\Delta k_{3,2}|\ll O(|\Delta a|)$.
    \item $1+c_2 a^{2-2n}\approx 0$: Then $|\Delta k_{3,2}|\gg O(|\Delta a|)$.
    \item $c_2 a^{2-2n}\not\approx \pm 1$: Then $|\Delta k_{3,2}|\sim O(|\Delta a|)$.
\end{itemize}
 \item $|\Delta a|\gg \frac{1}{n}$ ($a^{-2n}\gg 1$): Then $|1+c_2 a^{2-2n}|\gg 1 \ll |1-c_2 a^{2-2n}|$ and $\Delta k_{3,2}\sim O(\Delta a)$.
\end{itemize}
\end{itemize}

\begin{itemize}
\item $|\delta k_{3,1}|\gg |\Delta a|$ ($c_2\approx -1$): 
\begin{itemize}\item $|\Delta a|\ll \frac{1}{n}$ ($a^{-2n}\approx 1+2n\Delta a$ and $c_2+1\approx O(|\frac{\Delta a}{\delta k_{3,1}}|)$):
\begin{itemize}
\item $|\delta k_{3,1}|\gg O(\frac{1}{n})$: Then $|1+c_2 a^{2-2n}|\sim O(|n\Delta a|)$ and $\Delta k_{3,2}\sim O(\frac{1}{n})$.
\item $|\delta k_{3,1}|\sim O(\frac{1}{n})$: Then $|v_{12}+v_{11}c_2 a^{2-2n}|\leq O(|n\Delta a|)$ and $\Delta k_{3,2}\geq O(\frac{1}{n})$.
\item $|\delta k_{3,1}|\ll O(\frac{1}{n})$: Then $|1+c_2 a^{2-2n}|\sim O(|\frac{\Delta a}{\delta k_{3,1}}|)$ and $\Delta k_{3,2}\sim O(|\delta k_{3,1}|)$.
\end{itemize}
\item $|\Delta a|\sim \frac{1}{n}$ ($1\not\approx a^{-2n}\sim O(1)$): Then $\Delta k_{3,2}\sim O(\Delta a)$.
 \item $|\Delta a|\gg \frac{1}{n}$ ($a^{-2n}\gg 1$): Then $\Delta k_{3,2}\sim O(\Delta a)$.
\end{itemize}
\end{itemize}
If $\delta k_{3,1}\geq O(1)$ and $|\Delta a|\ll 1$, then only the situation $|\delta k_{3,1}|\gg |O(\Delta a)|$ can take place and particularly $|\Delta a|\ll 1$ implies $|\Delta k_{3,2}|\ll 1$. 

\end{proof}

\section{Appendix: The norms of eigenvectors}
\label{sec:norm lemma}

Proof of Lemma~\ref{norm of eigenstates}:

When $(\omega^{(k)})^{2}\in(0,2k_{1})\cup(2k_{2},2k_{1}+2k_{2})$, $\forall\varepsilon>0$,
\begin{equation*}
  ||u^{(k)}|^{2}-n|\lesssim n^{\varepsilon}.
\end{equation*}

\begin{proof}
We only prove the Lemma for the frequencies in the optical band ($(\omega^{(k)})^{2}\in(2k_{1},2k_{1}+2k_{2})$) and the result for the acoustic band will similarly hold.
When $(\omega^{(k)})^{2}\in(2k_{1},2k_{1}+2k_{2})$, by the form of $u^{(k)}$, we have
\begin{equation}
\begin{split}
  |u^{(k)}|^{2}
  =n&+\frac{1}{4}\sum_{j=0}^{n-1}
  (e^{2i(\alpha^{(k)}+\beta^{(k)}+j\theta^{(k)})}+e^{-2i(\alpha^{(k)}+\beta^{(k)}+j\theta^{(k)})})\\
  &+\frac{1}{4}\sum_{j=0}^{n-1}
 (e^{2i(\beta^{(k)}-\alpha^{(k)}+j\theta^{(k)})}+e^{-2i(\beta^{(k)}-\alpha^{(k)}+j\theta^{(k)})})
 \end{split}
 \label{form: norm of u_k} 
\end{equation}
where
\begin{equation}
\begin{split}
 &\quad\sum_{j=0}^{n-1}
  |(e^{2i(\beta^{(k)}\pm\alpha^{(k)}+j\theta^{(k)})}+e^{-2i(\beta^{(k)}\pm\alpha^{(k)}+j\theta^{(k)})})|\\
 &=|\frac{\cos 2p_{k,\pm}-\cos2(p_{k,\pm}+n\theta^{(k)})-\cos2(p_{k,\pm}+\theta^{(k)})+\cos2(p_{k,\pm}+(n-1)\theta^{(k)})}
  {1-\cos 2\theta^{(k)}}|
 \label{form: norm of u_k 2}
 \end{split}
\end{equation}
and $p_{k,\pm}=\beta^{(k)}\pm\alpha^{(k)}$. Next, we estimate $|u^{(k)}|$ with $k$ in different ranges: 
\begin{itemize}
  \item When $(\omega^{(k)})^{2}\in(2k_{2},2k_{1}+2k_{2})$ and $1\leq k\leq n^{1-\varepsilon/2}$ or $n_{1}-n^{1-\varepsilon/2}\leq k\leq n_{1}$,
This is the situation when the corresponding frequency is near the lower or the upper edge of the spectral band. Again we only prove the former case and the latter can be similarly obtained. 
By  \eqref{eq:Delta theta}, Remark~\ref{remark:op upper}, \eqref{matrix:uk op upper} and \eqref{matrix:uk op lower}, both the numerator and the denominator of \eqref{form: norm of u_k 2} are of $\mathcal{O}(|\Delta\theta^{(k)}|^2)=\mathcal{O}(\frac{k^2}{n^2})$.
Thus
\begin{equation}
\label{norm of eigenstates1}
  ||u^{(k)}|^{2}-n|\lesssim 1.
\end{equation}
\end{itemize}

\begin{itemize}
  \item  When $(\omega^{(k)})^{2}\in(2k_{2},2k_{1}+2k_{2})$ and $n^{1-\varepsilon/2}< k< n_{1}-n^{1-\varepsilon/2}$, then
$$
  \pi+\mathcal{O}(n^{-\varepsilon/2})\sim\theta^{(n^{1-\varepsilon/2})}<\theta^{(k)}<\theta^{(n_1-n^{1-\varepsilon/2})}\sim 2\pi-\mathcal{O}(n^{-\varepsilon/2}).
$$  
Thus \eqref{form: norm of u_k 2} becomes
\begin{equation}
\begin{split}
 &|\frac{\cos 2p_{k,\pm}-\cos2(p_{k,\pm}+n\theta^{(k)})-\cos2(p_{k,\pm}+\theta^{(k)})+\cos2(p_{k,\pm}+(n-1)\theta^{(k)})}
  {1-\cos 2\theta^{(k)}}|\\
 &\leq\frac{4}{1-\cos2\theta^{(k)}}=\frac{2}{\sin^{2}\theta^{(k)}}\lesssim n^{\varepsilon}
 \end{split}
\end{equation}
As a result,
\begin{equation}
\label{norm of eigenstates2}
  ||u^{(k)}|^{2}-n|\lesssim n^{\varepsilon}.
\end{equation}
\end{itemize}
\end{proof}

\section{Appendix: Discussion on the inner product $(u^{(a)}u^{(b)}u^{(c)},u^{(d)})$}
\label{sec:inner product}

If $(\omega^{(a)})^{2},(\omega^{(b)})^{2},
(\omega^{(c)})^{2},(\omega^{(d)})^{2}\in (0,2k_1)\cup(2k_{2},2k_{1}+2k_{2})$, then by \eqref{matrix:uk op}\eqref{matrix:uk ac} we know 
\begin{equation}
\begin{split}
&\quad((u^{(a)}u^{(b)}u^{(c)},u^{(d)})) \\
&=\sum\limits_{j=0}^{n-1}[\prod\limits_{k=a,b,c,d}\cos(\alpha^{(k)}+\beta^{(k)}+j\theta^{(k)})+\prod\limits_{k=a,b,c,d}\sigma^{(k)}\cos(-\alpha^{(k)}+\beta^{(k)}+j\theta^{(k)})]\\
&=\frac{1}{16}\sum\limits_{j=0}^{n-1}\sum\limits_{j_1=0}^{1}\sum\limits_{j_2=0}^{1}\sum\limits_{j_3=0}^{1}\sum\limits_{j_4=0}^{1}[
e^{i[(-1)^{j_1}(p_{a,+}+j\theta^{(a)})+(-1)^{j_2}(p_{b,+}+j\theta^{(b)})+(-1)^{j_3}(p_{c,+}+j\theta^{(c)})+(-1)^{j_4}(p_{d,+}+j\theta^{(d)})]}\\
&+\sigma^{(a)}\sigma^{(b)}\sigma^{(c)}\sigma^{(d)}e^{i[(-1)^{j_1}(p_{a,-}+j\theta^{(a)})+(-1)^{j_2}(p_{b,-}+j\theta^{(b)})+(-1)^{j_3}(p_{c,-}+j\theta^{(c)})+(-1)^{j_4}(p_{d,-}+j\theta^{(d)})]}]
\end{split}
\label{general form of innner product}
\end{equation}
where $\sigma^{(k)}=\begin{cases}
1, & 2n-n_2+1\leq k\leq 2n \\
-1, & 1\leq k\leq n_1
\end{cases}$.
Then Lemma~\ref{lemma: inner product} (1) is obvious from the above, namely if all of the four eigenvectors are not localized then
$$
|(u^{(a)}u^{(b)}u^{(c)},u^{(d)})|\leq 2n.
$$
Next we consider the case where one of the frequencies in $(u^{(a)}u^{(b)}u^{(c)},u^{(d)})$ is in the bandgap (again under the assumption $(A1)$). Without loss of generality, we suppose $(\omega^{(d)})^2\in (2k_1+\delta_2,2k_2-\delta_2)$ and then obtain
\begin{eqnarray*}
|(u^{(a)}u^{(b)}u^{(c)},u^{(d)})|\leq \sum\limits_{j=1}^{2n}|u^{(d)}_j|\approx \frac{|\sqrt{k_1+k_2/\tilde{a}}+\sqrt{k_1+k_2\tilde{a}}|}{1-|\tilde{a}|}\leq\mathcal{O}(1)
\end{eqnarray*}
where $\tilde{a}$ is defined as in \eqref{eq:5} correpsonding to $\omega^{(d)}$. If $u^{(c)}$ and $u^{(d)}$ are eigenvectors localized at the left and right end respectively, then 
\begin{eqnarray*}
|(u^{(a)}u^{(b)}u^{(c)},u^{(d)})|\leq (|u^{(c)}|,|u^{(d)}|)\leq |u^{(c)}| |u^{(d)}|\leq\mathcal{O}(1)
\end{eqnarray*}
In the same spirit, since the localized state $u^{(d)}$ also satisfies $\sum\limits_{j=1}^{2n}|u^{(d)}_j|^3\leq\mathcal{O}(1)$ and $\sum\limits_{j=1}^{2n}|u^{(d)}_j|^4\leq\mathcal{O}(1)$, Lemma~\ref{lemma: inner product} (2) can be obtained. Moreover, suppose the number of frequencies from $(\omega^{(a)})^2, (\omega^{(b)})^2, (\omega^{(c)})^2, (\omega^{(d)})^2$ in the bandgap is $l$, then there exists $\tilde{C}>0$ such that 
\begin{equation}
    \frac{|(u^{(a)}u^{(b)}u^{(c)},u^{(d)})|}{|u^{(a)}||u^{(b)}||u^{(c)}||u^{(d)}|}<\frac{\tilde{C}}{n^{\frac{4-l}{2}}}, \quad 1\leq l\leq 4.
\end{equation}
In what follows we assume all the eigenfrequencies are near an edge of a spectral band to achieve a finer estimate for the inner product. To begin with, we prove the following proposition:
\begin{proposition}
\label{prop:inner product 1}
Suppose $k=\max\{a,b,c,d\}$. If $1\leq k\ll n$ and 
$$
(-1)^{j_1}a+(-1)^{j_2}b+(-1)^{j_3}c+(-1)^{j_4}d\neq 0,
$$
then 
\begin{equation}
|\sum\limits_{j=0}^{n-1}e^{i[((-1)^{j_1}p_{a,\pm}+(-1)^{j_2}p_{b,\pm}+(-1)^{j_3}p_{c,+\pm}+(-1)^{j_4}p_{d,\pm})+j((-1)^{j_1}\theta^{(a)}+(-1)^{j_2}\theta^{(b)}+(-1)^{j_3}\theta^{(c)}+(-1)^{j_4}\theta^{(d)})]}+c.c.|\lesssim 1
\label{eq:prop inner product 1}
\end{equation}
where "$c.c.$" stands for complex conjugate.
\end{proposition}
\begin{proof}
Since $k\ll n$, at the leading order we have
$$
\Delta\theta^{(k)}\approx k\Delta\theta^{(1)}, \quad \Delta\alpha^{(k)}\approx k\Delta\alpha^{(1)}, \quad \Delta\beta^{(k)}\approx k\Delta\beta^{(1)},
$$
thus
\begin{equation}
\begin{split}
\theta^*:=&(-1)^{j_1}\theta^{(a)}+(-1)^{j_2}\theta^{(b)}+(-1)^{j_3}\theta^{(c)}+(-1)^{j_4}\theta^{(d)} \\ 
\approx &((-1)^{j_1}+(-1)^{j_2}+(-1)^{j_3}+(-1)^{j_4})\pi\\
&+((-1)^{j_1}a+(-1)^{j_2}b+(-1)^{j_3}c+(-1)^{j_4}d)\Delta\theta^{(1)},
\end{split}
\end{equation}
\begin{equation}
\begin{split}
p^*_{\pm}:=&(-1)^{j_1}p_{a,\pm}+(-1)^{j_2}p_{b,\pm}+(-1)^{j_3}p_{c,+\pm}+(-1)^{j_4}p_{d,\pm} \\ \approx &\pm((-1)^{j_1}+(-1)^{j_2}+(-1)^{j_3}+(-1)^{j_4})\frac{\pi}{2}\\
&+((-1)^{j_1}a+(-1)^{j_2}b+(-1)^{j_3}c+(-1)^{j_4}d)(\Delta\beta^{(1)}\pm\Delta\alpha^{(1)}).
\end{split}
\end{equation}
Then the left hand side of \eqref{eq:prop inner product 1} can be written as
\begin{equation}
\begin{split}
  &|\frac{\cos p_{\pm}^* -\cos(p_{\pm}^*+n\theta^*) -\cos(p_{\pm}^*+\theta^*)+ \cos(p_{\pm}^*+(n-1)\theta^*)}
        {1-\cos\theta^*}|\\
\approx &|\frac{\cos l^*\Delta p_{1,\pm} -\cos l^*(\Delta p_{1,\pm}+2\Delta\theta^{(1)}+\pi) -\cos l^*(\Delta p_{1,\pm}+\Delta\theta^{(1)})+ \cos l^*(\Delta p_{1,\pm}+\Delta\theta^{(1)}+\pi)}
        {1-\cos l^*\Delta\theta^{(1)}}| 
\end{split}
\label{eq:prop inner product 2}
\end{equation}
where 
\begin{equation}
\label{Delta_p}
\Delta p_{k,\pm}=\Delta\beta^{(k)}\pm\Delta\alpha^{(k)}, \quad l^*=(-1)^{j_1}a+(-1)^{j_2}b+(-1)^{j_3}c+(-1)^{j_4}d.
\end{equation}
Since $0\neq l^*\ll n$, the denominator of \eqref{eq:prop inner product 2} is of the order $\mathcal{O}(|\frac{l^*}{n}|^2)$ and the numerator will not exceed this order. As a result, \eqref{eq:prop inner product 2} is at most $\mathcal{O}(1)$.
\end{proof}
Then the first part of Lemma~\ref{lemma: inner product} (3) directly follows from Proposition~\ref{prop:inner product 1}. To prove the second part of  Lemma~\ref{lemma: inner product} (3), we introduce the following proposition

\begin{proposition}
\label{lm: approx inner priduct}
Suppose $(\omega^{(a)})^{2},(\omega^{(b)})^{2},
(\omega^{(c)})^{2},(\omega^{(d)})^{2}\in(2k_{2},2k_{1}+2k_{2})$, If $k:=\max\{a,b,c,d\}\ll n$, then
\begin{equation}
|(u^{(a)}u^{(b)}u^{(c)},u^{(d)})-
(
\tilde{u}^{(a)}
\tilde{u}^{(b)}
\tilde{u}^{(c)},
\tilde{u}^{(d)})|
\lesssim
\frac{k^{3}}{n^{2}}.
\label{eq:approx inner product}
\end{equation}
\end{proposition}

\begin{proof}
According to the forms of eigenvectors and \eqref{eq:approx eigenvector err}, it can be estimated that
\begin{equation}
\begin{split}
 &
 |(u^{(a)}u^{(b)}u^{(c)},u^{(d)})-(\tilde{u}^{(a)}
\tilde{u}^{(b)}
\tilde{u}^{(c)},
\tilde{u}^{(d)})|\\
\leq& |(u^{(a)}u^{(b)}u^{(c)},u^{(d)})-(\tilde{u}^{(a)}u^{(b)}u^{(c)},
u^{d})|
+|(\tilde{u}^{(a)}u^{(b)}u^{(c)},
u^{(d)})-(\tilde{u}^{(a)}
\tilde{u}^{(b)}u^{(c)},
u^{(d)})|\\
&+|(\tilde{u}^{(a)}\tilde{u}^{(b)}
u^{(c)},
u^{(d)})-(\tilde{u}^{(a)}\tilde{u}^{(b)}\tilde{u}^{(c)},
u^{(d)})|
+|(\tilde{u}^{(a)}\tilde{u}^{(b)}\tilde{u}^{(c)},
u^{(d)})-(\tilde{u}^{(a)}\tilde{u}^{(b)}\tilde{u}^{c},
\tilde{u}^{(d)})|\\
\lesssim& \frac{k^3}{n^2}
\end{split}
\end{equation}
\end{proof}

\begin{corollary}
\label{cor:inner priducts equivalence}
Suppose $(\omega^{(a)})^{2},(\omega^{(b)})^{2},
(\omega^{(c)})^{2},(\omega^{(d)})^{2}\in(2k_{2},2k_{1}+2k_{2})$, If $k:=\max\{a,b,c,d\}\ll n$, then
\begin{equation}
  (\tilde{u}^{(a)}\tilde{u}^{(b)}\tilde{u}^{(c)},\tilde{u}^{(d)})\sim\mathcal{O}(n)\Longleftrightarrow
  (u^{(a)}u^{(b)}u^{(c)},u^{(d)})\sim\mathcal{O}(n).
 \label{eq:inner products equivalence}
\end{equation}
\end{corollary}
Now it suffices to consider $(\tilde{u}^{(a)}
\tilde{u}^{(b)}
\tilde{u}^{(c)},
\tilde{u}^{(d)})\sim \mathcal{O}(n)$ and this is easier than the discussion on $(u^{(a)}u^{(b)}u^{(c)},u^{(d)})$. Taking advantage of the form of $\tilde{u}^{(k)}$ in \eqref{eq:approx eigenvector}, we immediately find that $
(-1)^{j_1}a+(-1)^{j_2}b+(-1)^{j_3}c+(-1)^{j_4}d=0
$ implies that the terms $e^{\pm i[((-1)^{j_1}a+(-1)^{j_2}b+(-1)^{j_3}c+(-1)^{j_4}d)(\Delta p_{1,\pm}+j\Delta\theta^{(1)})]}$ in $(\tilde{u}^{(a)}
\tilde{u}^{(b)}
\tilde{u}^{(c)},
\tilde{u}^{(d)})$ are one. Due to the symmetry, we only discuss that case with $j_1=0$ and it can be checked that
\begin{itemize}
    \item If $a+b+c+d=0$, then no other terms of the form $a+(-1)^{j_2}b+(-1)^{j_3}c+(-1)^{j_4}d$ can be zero.
    \item If $a-b-c-d=0$, then no other terms of the form
    $a+(-1)^{j_2}b+(-1)^{j_3}c+(-1)^{j_4}d$ can be zero.
    \item If $a+b-c-d=0$, then $a-b+c-d$ ($a-b-c+d$) can also be zero when $b=c$ ($b=d$). It should be noticed that $\mathcal{O}(n)$ terms $\sum\limits_{j=0}^{n-1}e^{\pm i[(a+b-c-d)(\Delta p_{1,\pm}+j\Delta\theta^{(1)})]}$ and $\sum\limits_{j=0}^{n-1}e^{\pm i[(a-b+c-d)(\Delta p_{1,\pm}+j\Delta\theta^{(1)})]}$ (or $\sum\limits_{j=0}^{n-1}e^{\pm i[(a-b-c+d)(\Delta p_{1,\pm}+j\Delta\theta^{(1)})]}$) will add up rather than cancel each other.
    \item If $a+(-1)^{j_2}b+(-1)^{j_3}c+(-1)^{j_4}d\neq 0$, then $e^{ i[(a+(-1)^{j_2}b+(-1)^{j_3}c+(-1)^{j_4}d)(\Delta p_{1,\pm}+j\Delta\theta^{(1)})]}+c.c.$ will have an $\mathcal{O}(1)$ upper bound by the proof of Prop.~\ref{prop:inner product 1}.
\end{itemize}
As a result, we have proved the following Lemma~\ref{lemma:approx inner product order n} hence the second part of Lemma~\ref{lemma: inner product} (3).
\begin{lemma}
\label{lemma:approx inner product order n}
Suppose $k=\max\{a,b,c,d\}$. If $1\leq k\ll n$ then 
$$
(-1)^{j_1}a+(-1)^{j_2}b+(-1)^{j_3}c+(-1)^{j_4}d=0 \Longleftrightarrow (\tilde{u}^{(a)}
\tilde{u}^{(b)}
\tilde{u}^{(c)},
\tilde{u}^{(d)})\sim \mathcal{O}(n).
$$
\end{lemma}

~\\
Proof of Lemma~\ref{lemma: inner product} (4):

When $(\omega^{(a)})^{2},(\omega^{(b)})^{2},(\omega^{(c)})^{2},(\omega^{(d)})^{2}\in(2k_{2},2k_{1}+2k_{2})$ and $\max\{a,b,c\}\ll n^{1-\varepsilon}< d< n_{1}-a-b-c$,
\begin{equation*}
  (u^{(a)}u^{(b)}u^{(c)},u^{(d)})\sim\mathcal{O}(n^{2\varepsilon})
\end{equation*}

\begin{proof}
Here we write $\theta^{(l)}=\pi+\Delta\theta^{(l)}, l=a,b,c,d$ where $|\Delta\theta^{(l)}|\ll 1$ for $l=a,b,c$. However $|\Delta\theta^{(d)}|$ is not necessarily small and in fact $|\Delta\theta^{(d)}|\in (\frac{2n^{1-\varepsilon}}{n-1}\pi,\pi-\frac{a+b+c+\frac{3}{4}}{n-1}\pi)$. Then in $(u^{(a)}u^{(b)}u^{(c)},u^{(d)})$ we have
\begin{equation}
\begin{split}
&|\sum\limits_{j=0}^{n-1}e^{i[((-1)^{j_1}p_{a,\pm}+(-1)^{j_2}p_{b,\pm}+(-1)^{j_3}p_{c,+\pm}+(-1)^{j_4}p_{d,\pm})+j((-1)^{j_1}\theta^{(a)}+(-1)^{j_2}\theta^{(b)}+(-1)^{j_3}\theta^{(c)}+(-1)^{j_4}\theta^{(d)})]}+c.c.|\\
=  &|\frac{\cos p_{\pm}^* -\cos(p_{\pm}^*+n\theta^*) -\cos(p_{\pm}^*+\theta^*)+ \cos(p_{\pm}^*+(n-1)\theta^*)}
        {1-\cos\theta^*}|\\
\leq & \frac{4}{2\sin^2(\frac{(-1)^{j_1}\Delta\theta^{(a)}+(-1)^{j_2}\Delta\theta^{(b)}+(-1)^{j_3}\Delta\theta^{(c)}+(-1)^{j_4}\Delta\theta^{(d)}}{2})}.
\end{split}
\end{equation}
where 
\begin{equation*}
\begin{split}
&|(-1)^{j_1}\Delta\theta^{(a)}+(-1)^{j_2}\Delta\theta^{(b)}+(-1)^{j_3}\Delta\theta^{(c)}+(-1)^{j_4}\Delta\theta^{(d)}| \\
\leq&|\Delta\theta^{(a)}+\Delta\theta^{(b)}+\Delta\theta^{(c)}+\Delta\theta^{(d)}|<\pi-\frac{n^{1-\varepsilon}/2}{n-1} \pi   
\end{split}
\end{equation*}
and 
\begin{equation*}
\begin{split}
&|(-1)^{j_1}\Delta\theta^{(a)}+(-1)^{j_2}\Delta\theta^{(b)}+(-1)^{j_3}\Delta\theta^{(c)}+(-1)^{j_4}\Delta\theta^{(d)}|\\
\geq&|-\Delta\theta^{(a)}-\Delta\theta^{(b)}-\Delta\theta^{(c)}+\Delta\theta^{(d)}|\gtrsim\frac{n^{1-\varepsilon}}{n}.
\end{split}
\end{equation*}
That is to say
$$
\sin^2(\frac{(-1)^{j_1}\Delta\theta^{(a)}+(-1)^{j_2}\Delta\theta^{(b)}+(-1)^{j_3}\Delta\theta^{(c)}+(-1)^{j_4}\Delta\theta^{(d)}}{2})\gtrsim n^{-2\varepsilon}
$$
hence 
\begin{equation}
 |(u^{(a)}u^{(b)}u^{(c)},u^{(d)})|\lesssim n^{2\varepsilon}.
\end{equation}

\end{proof}

\section{Appendix: Some inequalities}

\begin{proposition}
\label{prop:inequalities_ab_abc}
The following inequalities hold:
\begin{eqnarray}
\label{eq:inequality_E31*}
  &\sum\limits_{(a,b)\in \mathbb{E}_{3}^{1^*}(k)}
    \frac{1}{(a^{2}+2)(b^{2}+2)}
    <\frac{2\pi^{2}-8}{3}\cdot\frac{1}{(k+1)^{2}+2}<\frac{4}{(k+1)^{2}+2};\\
\label{eq:inequality_E21*}    
  &\sum\limits_{(a,b,c)\in\mathbb{E}_{2}^{1^*}(k)}
   \frac{1}{(a^{2}+2)(b^{2}+2)(c^{2}+2)}<(\frac{2\pi^2-8}{3})(\frac{4\pi^2-4}{3})\frac{1}{(k+1)^{2}+2}
   <\frac{48}{(k+1)^{2}+2};\\
\label{eq:inequality_E11}   
  &\sum\limits_{(a,b,c)\in
   \mathbb{E}_{1}^{1}(k)}
   \frac{1}{(a^{2}+2)(b^{2}+2)(c^{2}+2)}
   <(\frac{2\pi^2-2}{3})^2\frac{1}{k^2+2}<\frac{36}{k^{2}+2}.
\end{eqnarray}
\end{proposition}

\begin{proof}
For the first inequality \eqref{eq:inequality_E31*}, we have
\begin{equation}
\begin{split}
 &\sum_{(a,b)\in \mathbb{E}_{3}^{1^*}(k)}\frac{1}{(a^{2}+1)(b^{2}+1)}=\sum_{a=1}^{[\frac{k+1}{2}]}\frac{1}{(a^{2}+2)((k+1-a)^{2}+2)}\\
 &<\frac{1}{(\frac{k+1}{2})^{2}+2}[\frac{1}{1^{2}+2}+\frac{1}{2^{2}+2}+\cdots+\frac{1}{(\frac{k+1}{2})^{2}+2}]\\
 &<\frac{4}{(k+1)^{2}+2}(\frac{\pi^{2}}{6}-1+\frac{1}{3})<\frac{4}{(k+1)^{2}+2}.
\end{split}
\end{equation}
Immediately we also know 
\begin{equation}
\label{eq:inequalitiy_E31_E32}
\begin{split}
&\sum_{(a,b)\in \mathbb{E}_{3}^{0^*}(k)}\frac{1}{(a^{2}+2)(b^{2}+2)}
    \leq\sum_{(a,b)\in \mathbb{E}_{3}^{0}(k)}\frac{1}{(a^{2}+2)(b^{2}+2)} \\
    &<2\times 4\times(\frac{\pi^{2}}{6}-1+\frac{1}{3}+\frac{1}{2})\frac{1}{(k+1)^{2}+2}<\frac{12}{(k+1)^{2}+2}.
\end{split}
\end{equation}
The second inequality \eqref{eq:inequality_E21*} can be derived as: 
\begin{equation*}
\begin{split}
  &\sum_{(a,b,c)\in \mathbb{E}_{2}^{1^*}(k)}\frac{1}{(a^{2}+2)(b^{2}+2)(c^{2}+2)}=\sum_{c=1}^{k+1-2}\frac{1}{c^2+2}\sum_{(a,b)\in\mathbb{E}_3^{1^*}(k+1-c)}\frac{1}{(a^{2}+2)(b^{2}+2)}\\
  &<\frac{2\pi^2-8}{3}\sum_{c=1}^{k+1-2}\frac{1}{(c^2+2)((k+1-c)^2+2)}
<\frac{2\pi^2-8}{3}\sum_{(c,d)\in\mathbb{E}_3^{0^*}(k)}\frac{1}{(c^2+2)(d^2+2)}\\
  &<(\frac{2\pi^2-8}{3})(\frac{4\pi^2-4}{3})\frac{1}{(k+1)^2+2}<\frac{48}{(k+1)^2+2}.
\end{split}
\end{equation*}
And we prove the third inequality \eqref{eq:inequality_E11} as
\begin{equation*}
\begin{split}
  &\quad\sum_{(a,b,c)\in \mathbb{E}_{1}^{1}(k)}\frac{1}{(a^{2}+2)(b^{2}+2)(c^{2}+2)}=\sum_{a=0}^{[\frac{k}{3}]}\frac{1}{a^2+2}\sum_{b=a}^{[\frac{k-a}{2}]}\frac{1}{(b^{2}+2)((k-a-b)^{2}+2)}\\
  &<\frac{2\pi^2-2}{3}\sum_{a=0}^{[\frac{k}{3}]}\frac{1}{(a^2+2)((k-a)^2+2)}
  <(\frac{2\pi^2-2}{3})^2\frac{1}{k^2+2}<\frac{36}{k^2+2}.
\end{split}
\end{equation*}
\end{proof}

\end{appendices}

\end{document}